\documentclass[british,superscript,utf8,nomove,x11names,shortlabels]{article}
\usepackage{amsmath}
\usepackage{amsthm}
\usepackage{amssymb}
\usepackage{unicode-math}
\usepackage[a4paper]{geometry}
\usepackage{color}
\usepackage{babel}
\usepackage{float}
\usepackage{units}
\usepackage{mathtools}
\usepackage{url}
\usepackage{graphicx}
\usepackage{rotfloat}
\usepackage[bookmarks=true,bookmarksnumbered=false,bookmarksopen=false,
 breaklinks=true,pdfborder={0 0 0},pdfborderstyle={},backref=false,colorlinks=true]
 {hyperref}
\hypersetup{pdftitle={Variational Kinetics},
 pdfauthor={Ronan Fleming},
 colorlinks=true,citecolor=magenta,linkcolor=blue}

\makeatletter

\theoremstyle{plain}
\newtheorem{thm}{\protect\theoremname}
\newenvironment{lyxcode}
	{\par\begin{list}{}{
		\setlength{\rightmargin}{\leftmargin}
		\setlength{\listparindent}{0pt}
		\raggedright
		\setlength{\itemsep}{0pt}
		\setlength{\parsep}{0pt}
		\normalfont\ttfamily}%
	 \item[]}
	{\end{list}}

\@ifundefined{date}{}{\date{}}
\catcode`\^^99=\active
\protected\def^^99{%
  \ifmmode
    \text{\texttrademark}%
  \else
    \texttrademark{}%
  \fi
}

\usepackage{bbm}

\makeatother

\usepackage[style=numeric,url=false,isbn=false,eprint=false,doi=false,firstinits=true,sorting=none]{biblatex}
\providecommand{\theoremname}{Theorem}

\begin{document}
\title{Variational kinetics: \\
elementary reaction kinetics via conic optimisation}
\author{\noindent Ronan M.T. Fleming\textsuperscript{}$^{1,2,3}$\thanks{To whom correspondence should be addressed: ronan.mt.fleming@gmail.com},
Ines Thiele$^{1,2,4,5}$}
\maketitle
\begin{center}
$^{1}$Digital Metabolic Twin Centre, $^{2}$School of Medicine,$^{3}$Institute
for Clinical Trials,\\
 $^{4}$School of Microbiology, $^{5}$Ryan Institute, \\
University of Galway, University Road, Galway, Ireland,
\par\end{center}
\begin{abstract}
Genome-scale modelling methods primarily predict reaction fluxes,
whereas established high throughput experimental technologies primarily
measure molecular species concentrations. This apparently paradoxical
situation has arisen because implementing the nonlinear constraints
that represent reaction kinetic rate equations is challenging without
resorting to convenient yet inaccurate approximations or to expansions
that are valid only near a reference state. We present a mathematically
and computationally tractable solution to this problem. First, we
introduce a mathematical reformulation of established knowledge of
metabolic reactions and reaction kinetics in matrix–vector notation.
We then present \emph{variational kinetics}, a novel approach that
satisfies steady state reaction kinetics at genome scale by exponential
conic optimisation. The nonlinear rate law constraints are relaxed
to exponential cones, which renders the feasible set convex, and satisfaction
of elementary kinetics is recovered by minimising a strictly concave
merit function over that set, which attains zero if, and only if,
every rate law holds. We establish that a particular sequence of conic
optimisation problems converges to a stationary point of this merit
function, and that every such stationary point is a steady state satisfying
elementary kinetics. Moiety conservation, thermodynamic constraints
on elementary kinetic parameters, regularised steady states and linear
optimisation of external reaction rates are each accommodated within
the same conic formulation. We demonstrate the approach computationally
on a genome-scale metabolic model.

\begin{figure}[H]
\includegraphics[width=1\textwidth]{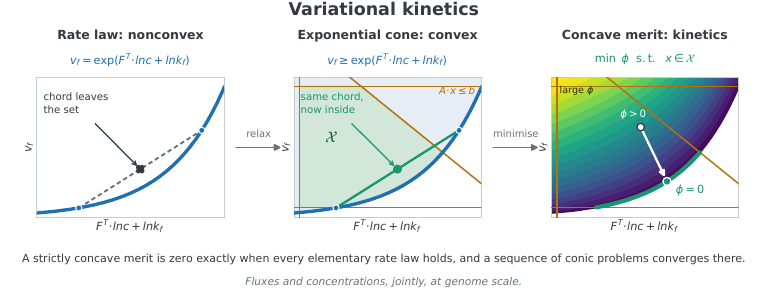}
\end{figure}
\end{abstract}


\section{Introduction}

Metabolism is the network of enzyme-catalysed reactions through which
cells extract energy and synthesise the molecules required for life.
Reconstructed at the scale of an entire genome, this network underpins
genome-scale models of metabolism, which have become standard tools
for interpreting biochemical, genetic and clinical data and for engineering
cellular chemistry. A metabolic network admits two complementary quantitative
descriptions: the rates at which its reactions proceed, that is their
fluxes, and the concentrations of the metabolites they interconvert
\parencite{fleming_conditions_2016}. A model that predicted both,
consistently, would let measurements of one be interpreted in terms
of the other — but the two descriptions have proven far easier to
obtain apart than together.

Constraint-based methods, such as flux balance analysis \parencite{orth_what_2010},
predict fluxes from network stoichiometry together with an optimisation
principle. They scale readily to genome-scale networks, but deliberately
omit reaction kinetics and therefore cannot predict metabolite concentrations.
High-throughput technologies, by contrast, increasingly measure concentrations
rather than fluxes. The quantities that are easiest to predict are
thus not the quantities that are easiest to measure. Perfectly closing
this gap requires introduction of reaction-kinetic rate laws that
couple concentrations to fluxes in a thermodynamically consistent
manner. However, the corresponding constraints are nonlinear, the
resulting feasible set is non-convex \parencite{qian_stoichiometric_2003},
and therefore satisfying them at genome-scale has proven both mathematically
and computationally demanding.

There are many modelling methods that introduce various aspects of
reaction-kinetics, thermodynamic consistency, or both, into genome-scale
models. We do not attempt an exhaustive reference list, rather we
refer to certain approaches as examples of the challenges associated
with different classes of methods. Thermodynamic constraints on the
direction of reactions can be applied by introduction of thermochemical
constraints \parencite{smith_thermo-flux_2026}, but this relies on
extrapolation from measured thermochemical data \parencite{noor_consistent_2013}
and sufficient metabolomic data. Thermodynamic constraints on combinations
of reaction directions can be efficiently satisfied at genome-scale
using linear optimisation \parencite{desouki_cyclefreeflux_2015}
and combined with constraints relating the ratio of unidirectional
fluxes to thermodynamic driving force \parencite{beard_relationship_2007}
using non-linear yet convex optimisation \parencite{fleming_variational_2012},
which can also be biased using available omics data \parencite{amestica-toledo_thermodynamically_2026}.
However, these methods omit constraints relating molecular species
abundance (enzymes, metabolites) to absolute reaction rates.

A complementary line of work takes a flux distribution as given and
predicts the accompanying metabolite concentrations by convex optimisation:
the max-min driving force method \parencite{noor_pathway_2014} selects
concentrations that maximise the smallest thermodynamic driving force
along a pathway, while enzyme cost minimisation \parencite{liebermeister_enzyme_2015}
selects them to minimise the total enzyme demand implied by the rate
laws. These formulations are convex and scale well, but they presuppose
a known flux distribution and enforce only thermodynamic feasibility,
or an enzyme-cost optimum, rather than the elementary rate law itself,
so they determine concentrations for a given flux rather than fluxes
and concentrations together. Provided that kinetic constraints and
kinetic parameters are formulated in a thermodynamically consistent
manner \parencite{lubitz_parameter_2010}, introduction of reaction-kinetic
rate laws obviates the need to apply separate thermodynamic constraints.
As introduction of reaction-kinetic rate laws is challenging, a variety
of different methods have instead attempted to represent various aspects
of reaction kinetics using approximations to rate laws. Reaction rates
can be bounded from above using the product of $k_{cat}$ times the
concentration of the catalysing enzyme and this, together with an
upper bound for the available enzyme pool, results in a scalable linear
approximation to reaction kinetics \parencite{domenzain_reconstruction_2023,adadi_prediction_2023,sanchez_improving_2017}
but does not take into account non-linear kinetic effects such as
saturation, thermodynamic driving force, or regulation. Rate laws
may be approximated by a log-linear expansion about a reference state,
in which the logarithm of a reaction rate is a linear function of
the logarithms of the reactant concentrations \parencite{savageau_biochemical_1970}.
The closely related linear-logarithmic (lin-log) kinetics \parencite{hatzimanikatis_effects_1997}
instead make the rate itself, scaled by enzyme level, a linear function
of the logarithmic concentrations, taking the elasticities of metabolic
control analysis as its parameters. However, these approximations
share a common compromise as expansions are built about a fixed reference
state and lose accuracy away from it (cf. Supplementary Figure \ref{fig:kineticFormats}).
An alternative stepwise approach is to represent reaction-kinetic
rate laws but anchor a model ensemble to reference steady states,
e.g., obtained from fluxomic data across many strains, then prune
it using perturbation data \parencite{tran_ensemble_2008,khodayari_kinetic_2014,khodayari_genome-scale_2016},
but this requires sufficient reference data at genome-scale, and the
retained ensemble is not uniquely identified, so its predictions remain
sensitive to the choice of reference state and to the sampling and
pruning procedure.

Reaction kinetics rate laws, and various approximations thereof, can
be added to genome-scale kinetic models then the feasible set of fluxes,
concentrations and parameters can be sampled in a manner consistent
with experimental data \parencite{miskovic_production_2010,toumpe_scalable_2026},
however the feasible set is non-linear and non-convex so there are
no guarantees that the numerical sample distribution will match the
theoretically desired distribution. Alternatively, one can embed nonlinear
rate laws directly into a constraint-based model and predict network
states using mixed-integer nonlinear optimisation algorithms that
exploit the mathematical properties particular to systems of reaction
rate laws \parencite{bekiaris_cobra-k_2026}, however reliable convergence
is not guaranteed as the network grows. Mass-action stoichiometric
simulation builds dynamic models by mapping measured concentrations
and fluxes onto the network \parencite{jamshidi_mass_2010,haiman_masspy_2021},
but it requires simultaneous measurements of concentrations and fluxes
to fix the model, and its mass-action form neglects enzyme saturation
and allosteric regulation, so it is accurate only near the state at
which it was parameterised. A common tension runs through these approaches:
one must either approximate the kinetics and forfeit biochemical fidelity,
depend on a pre-specified reference state and extensive parameterisation
and forfeit uniqueness and identifiability, or confront a nonconvex
optimisation that does not reliably scale. What remains missing is
a formulation of steady-state elementary reaction kinetics that is
at once mathematically exact, independent of any reference state,
and expressed as a tractable optimisation that determines reaction
fluxes and metabolite concentrations together.

Here we present variational kinetics, an approach that reformulates
steady-state elementary reaction kinetics as an optimisation problem
and solves it at genome scale through a sequence of exponential conic
optimisation problems. There is prior reason to expect such a reformulation
to be attainable: the thermodynamic and kinetic relationships that
govern an elementary reaction are log-linear in the chemical potentials
of its species, and the exponential cone captures exactly this log-linear
structure, so it was plausible in advance that elementary kinetics
could be cast as conic optimisation without approximation. We first
restate established results on metabolic reactions and reaction kinetics
in a consistent matrix-vector notation; we then develop the variational
kinetics formulation together with the numerical method used to solve
it, and characterise the conditions under which the iterative scheme
converges to a steady state; finally we illustrate the approach on
a genome-scale model of dopaminergic neuronal metabolism. Our aim
throughout is to answer a single question: can thermodynamically and
kinetically consistent reaction fluxes and metabolite concentrations
be predicted at genome scale from network structure and kinetic parameters
alone, without approximate rate laws and without a pre-specified reference
state?

\section{\protect\label{sec:Notation}Notation}

Throughout , $\mathbb{R}$, $\mathbb{R}^{n}$, and $\mathbb{R}^{m\times n}$
denote the field of real numbers, the vector space of $n$-tuples
of real numbers, and the space of $m\times n$ matrices with entries
in $\mathbb{R}$, respectively. Similarly, $\mathbb{Z}$, $\mathbb{Z}^{n}$,
$\mathbb{Z}^{m\times n}$ stand for integer numbers, the vector space
of $n$-tuples of integer number, and the space of matrices with entries
in $\mathbb{Z}$, respectively. $\mathbb{R}_{\ge0}^{n}$ and $\mathbb{R}_{>0}^{n}$
denote non-negative real $n$-tuples and positive real $n$-tuples
in $\mathbb{R}^{n}$, respectively, and $\mathbb{Z}_{\ge0}^{n}$ and
$\mathbb{Z}_{>0}^{n}$ denote non-negative integer $n$-tuples and
positive integer $n$-tuples in $\mathbb{Z}^{n}$, respectively. We
use Householder notation, that is a matrix is denoted by uppercase
Roman, such as $A\in\mathbb{R}^{m\times n}$. $A_{\textrm{i}}$ and
$A_{\textrm{:,j}}$ denote the $i^{th}$ row and the $j^{th}$ column
of $A$, respectively, where $i\in1,\ldots,m$ and $j\in1,\ldots,n$.
Note that subscript indexes are lower case Roman letters set in normal
font like $\textrm{i}$, rather than italic $i$. $A^{T}$ denotes
the transpose of a matrix $A$. $\mathbf{1}$ denotes a vector of
all ones and $I$ denote an identity matrix, with dimensions appropriate
to the circumstance. A calligraphic, uppercase, roman letter, e.g.,
$\mathcal{A}$, denotes a set, multiset or sequence, with $\{\cdot,\cdot\}$
denoting an unordered pair, $(\cdot,\cdot)$ denoting an ordered pair
and $(\cdot,\ldots,\cdot)$ denoting a sequence. Let $\left|\mathcal{A}\right|$
denote the cardinality of the set $\mathcal{A}$.The dot product of
$x$ and $y$ is denoted by $x^{T}\cdot y$, the Hadamard product
(element-wise product) is denoted by $x\odot y$, the Hadamard divisor
(element-wise division) is denoted by $x\oslash y$. Such products
and divisors of vectors require both vectors to have compatible dimensions.
Where $x$ is a vector $x^{-1}=\mathbf{1}\oslash x$ both denote a
component-wise inverse. $\left[\,\cdot,\cdot\,\right]$ denotes horizontal
concatenation and $\left[\begin{array}{c}
\cdot\\
\cdot
\end{array}\right]$ or $\left[\,\cdot;\cdot\,\right]$ denote vertical concatenation.
Also, $\exp(x)$ of a vector $x$ means component-wise exponential.
Where a diagonal matrix formed from a vector is required it is written
$\textrm{diag}(\cdot)$, while the Hadamard product is used wherever
both operands are vectors of the same dimension. $\nabla$ is the
gradient of a scalar valued function, or, for a vector valued function,
the matrix whose columns are the gradients of the components. The
expression $f(a\mid b)$ means that the vector valued function $f$
has a vector variable argument $a$, given a vector of parameters
$b$. Let 
\begin{eqnarray*}
\textrm{sign}(x) & \coloneqq & \begin{cases}
1 & \textrm{if}\;x>0,\\
-1 & \textrm{if}\;x<0,\\
0 & \textrm{if}\;x=0,
\end{cases}
\end{eqnarray*}
and component-wise where $x$ is a vector.

A disadvantage of reformulating nonlinear mathematical models in terms
of conic optimisation is an expansion in the number of variables.
With expansion of terms beyond the 26 letters in the Roman alphabet,
one is forced to compromise on established notation guidelines in
a manner that maintains a reasonable correspondence between symbols
in a scientific paper and the corresponding symbols in computer programming
code that are supposed to have the same meaning. The exponential or
natural logarithm of a vector is meant component-wise and $\exp(\ln(0)):=0$.
For example, let $x$ denote a vector, then $\ln\left(x\right)$ denotes
the component-wise natural logarithm of that vector and $lnx$ denotes
a variable that is envisaged to equal $\ln\left(x\right)$ at the
optimum of a conic optimisation problem, within tolerances specified
by parameters input into a numerical optimisation solver. This approach
enables transparency in representation of correspondence between related
variables, avoids premature exhaustion of the Roman alphabet, while
marginally extending beyond Householder notation guidelines.

The following exceptions to the conventions above are retained, because
each is well established in its own literature, named after a person,
or needed to avoid a typographic ambiguity.
\begin{itemize}
\item $\ell$ denotes the vector of moiety concentrations. A lower case
script ell is used in place of $l$ so that it cannot be confused
with the digit one.
\item $\mathcal{R}$, $\mathcal{T}$ and $\mathcal{P}$ denote the gas constant,
temperature and pressure. These are scalars, not sets, but the calligraphic
forms are conventional in chemical thermodynamics. Note that $\mathcal{R}(\cdot)$
also denotes the range of a matrix, which is a set; the two are distinguished
by the presence of an argument.
\item $\mathcal{L}$ denotes a Lagrangian, after Lagrange. It is a function
rather than a set.
\item $V_{\max}$ and $K_{M}$ denote the limiting rate and the Michaelis
constant of a Michaelis-Menten rate law. These are scalars in upper
case Roman, which is standard in enzyme kinetics.
\item Units of measurement are set in upright type, that is $\textrm{K}$
for kelvin, $\textrm{atm}$ for atmosphere, $\textrm{mol}$ for mole
and $\textrm{L}$ for litre.
\item A symbol that carries one established meaning in the biochemical literature
and a different established meaning in the optimisation literature
is not disambiguated here. For example $F$ denotes the forward stoichiometric
matrix in the sections on reaction kinetics and the matrix of the
affine conic constraint in the sections on conic optimisation, each
being standard in its own field. The section in which a symbol appears
determines which meaning is intended.
\end{itemize}

\section{Mathematical formulation of reaction kinetics}

\subsection{Reaction stoichiometry}

Consider a biochemical network with $m$ molecular species and $n$
reactions. Henceforth, species means molecular species. Typically,
though not always $m<n$. We assume that all net reactions are \emph{reversible}
and that each can be represented by a pair of unidirectional, \emph{forward}
and \emph{reverse,} reactions. With respect to the forward direction,
let the relative quantity, or \emph{stoichiometry}, of species $i$
participating as a substrate or catalyst in a \emph{forward reaction}
$j$, be a whole number entry $F_{\textrm{i,j}}>0$, in a \emph{forward
stoichiometric matrix} $F\in\mathbb{\mathbb{Z}}^{m\times n}$. Likewise,
with respect to the reverse direction, let the stoichiometry\emph{
}of species $i$ participating as a substrate or catalyst in \emph{reverse
reaction} $j$, be an entry $R_{\textrm{i,j}}>0$, in a \emph{reverse
stoichiometric matrix} $R\in\mathbb{\mathbb{Z}}^{m\times n}$. Then,
$N\in\mathbb{Z}^{m\times n}\coloneqq R-F$ is a (net) stoichiometric
matrix, where $N_{\textrm{i,j}}$ is the number of instances of molecule
$i$ that are consumed (negative) or produced (positive) in reaction
$j$. We assume that each column of $N$ corresponds a reaction where
mass is conserved.

\subsection{Elementary reaction kinetics}

Elementary reaction kinetics refers to the study of the individual,
simple steps that occur during a chemical reaction at the molecular
level. Each elementary reaction represents a single molecular event,
such as the breaking or formation of a chemical bond. An elementary
reaction is one for which no reaction intermediates have been detected
or need to be postulated in order to describe the chemical reaction
on a molecular scale. Elementary reactions provide a detailed, step-by-step
description of how a substrate interacts with an enzyme at the molecular
level. Let $v_{f}\in\mathbb{R}_{>0}^{n}$ and $v_{r}\in\mathbb{R}_{>0}^{n}$
denote forward and reverse elementary reaction rates, both of which
are a function of species concentrations $c\in\mathbb{R}_{>0}^{m}$
as well as forward and reverse elementary kinetic parameters, denoted
$k_{f}\in\mathbb{R}_{>0}^{n}$ and $k_{r}\in\mathbb{R}_{>0}^{n}$
respectively. Any reaction rate law assumes that reaction rate is
a function of concentrations and kinetic parameters, so let the net
reaction rate be 
\[
v_{net}(c\mid k_{f},k_{r})\coloneqq v_{f}(c\mid k_{f})-v_{r}(c\mid k_{r}).
\]
Henceforth, we assume that the rate of an elementary reaction is proportional
to the product of the concentrations of each of substrate (or catalyst),
each to the power of their respective stoichiometry in the reaction.
It follows that \emph{elementary kinetics} for the forward and reverse
reaction rates of reaction $j$ are given by
\begin{eqnarray*}
v_{f\textrm{j}}(c\mid k_{f\textrm{j}}) & \coloneqq & k_{f\textrm{j}}\prod c_{\textrm{i}}^{F_{\textrm{i,j}}},\\
v_{r\textrm{j}}(c\mid k_{r\textrm{j}}) & \coloneqq & k_{r\textrm{j}}\prod c_{\textrm{i}}^{R_{\textrm{i,j}}}.
\end{eqnarray*}
Using matrix vector notation and elementary logarithmic and exponential
identities, we formulate elementary kinetics for the forward and reverse
reaction rate vectors as
\begin{eqnarray}
v_{f}(c\mid k_{f}) & = & \exp(\ln(k_{f})+F^{T}\cdot\ln(c)),\label{eq:elementaryKineticsF}\\
v_{r}(c\mid k_{r}) & = & \exp(\ln(k_{r})+R^{T}\cdot\ln(c)).\label{eq:elementaryKineticsR}
\end{eqnarray}
The exponential or natural logarithm of a vector is meant component-wise
\footnote{Strictly, it is not proper to take the logarithm of a unit that has
physical dimensions. This difficulty can be avoided by considering
$c$ as a vector of mole fractions rather than concentrations (Eq.
19.93 in \parencite{berry_physical_2000}).}. A set of concentrations, kinetic parameters and rates that satisfy
Eq. (\ref{eq:elementaryKineticsF}) and Eq. (\ref{eq:elementaryKineticsR})
are said to be \emph{kinetically feasible}. Taking the logarithm of
both sides of Eqs. (\ref{eq:elementaryKineticsF}) and (\ref{eq:elementaryKineticsR})
we have
\begin{eqnarray}
\ln\left(v_{f}\right) & = & \ln(k_{f})+F^{T}\cdot\ln(c),\label{eq:logElementaryKineticsF}\\
\ln\left(v_{r}\right) & = & \ln(k_{r})+R^{T}\cdot\ln(c).\label{eq:logElementaryKineticsR}
\end{eqnarray}
Assuming elementary reaction kinetics, net rate is
\begin{eqnarray}
v_{net}(c\mid k_{f},k_{r}) & = & v_{f}(c\mid k_{f})-v_{r}(c\mid k_{r})=\exp(\ln(k_{f})+F^{T}\cdot\ln(c))-\exp(\ln(k_{r})+R^{T}\cdot\ln(c)).\label{eq:netElementaryReactionFlux}
\end{eqnarray}
Note that net rate is a function of species concentration given kinetic
parameters. That is, for now we assume we are given kinetic parameters
and are only interested in modelling concentrations and rates. Of
course, in reality the situation is more complicated as, at best,
one has experimental estimates of kinetic parameters.

\subsection{Steady state}

With respect to time $t$, the rate of change of species concentrations
is given by the dot product of net reaction stoichiometry and net
reaction rate
\begin{eqnarray*}
\dot{c}\coloneqq\frac{dc}{dt} & = & N\cdot v_{net}(c\mid k_{f},k_{r}),
\end{eqnarray*}
which is an ordinary differential equation. When $\dot{c}_{\textrm{i}}=0$
the rate of production equals the rate of consumption of species $i$,
that is, species $i$ is at a steady state. When $\dot{c}_{\textrm{i}}>0$
the rate of production is greater than the rate of consumption of
species $i$ and when $\dot{c}_{\textrm{i}}<0$ the rate of production
is less than the rate of consumption of species $i$. Assuming elementary
reaction kinetics, Eq. \ref{eq:netElementaryReactionFlux}, we have
\begin{eqnarray}
\dot{c} & = & N\cdot(\exp(\ln(k_{f})+F^{T}\cdot\ln(c))-\exp(\ln(k_{r})+R^{T}\cdot\ln(c))),\nonumber \\
 & = & [N,-N]\cdot\exp\left(\ln\left(\left[\begin{array}{c}
k_{f}\\
k_{r}
\end{array}\right]\right)+[F,R]^{T}\cdot\ln(c)\right)\nonumber \\
 & = & \left([R,F]-[F,R]\right)\cdot\exp\left(\ln\left(\left[\begin{array}{c}
k_{f}\\
k_{r}
\end{array}\right]\right)+[F,R]^{T}\cdot\ln(c)\right)\label{eq:dcdt-1-1-1}
\end{eqnarray}
where the latter are obtained by gathering related terms and substituting
$N=R-F$ to present, in matrix vector format, the fundamental equation
representing evolution of concentration with respect to time according
to elementary reaction kinetics.

\subsection{Mass balance}

Let $B\in\mathbb{Z}^{m\times k}$ be a stoichiometric matrix, where
each column corresponds to an \emph{external reaction,} which is a
modelling construct used to represent the exchange of mass between
a biochemical network and its environment. Let $w\in\mathbb{R}^{k}$
denote net exchange reaction rate. Assuming \emph{mass balance}, the
rate of change of species concentrations is equal to the net rate
of exchange with the environment. By convention, this is expressed
as
\begin{eqnarray*}
\dot{c} & = & -B\cdot w
\end{eqnarray*}
or equivalently
\begin{eqnarray}
N\cdot(v_{f}(c\mid k_{f})-v_{r}(c\mid k_{r}))+B\cdot w & = & 0\nonumber \\
N\cdot v_{net}(c\mid k_{f},k_{r})+B\cdot w & = & 0.\label{eq:massBalance}
\end{eqnarray}
Note that with this convention, if $w_{\textrm{j}}<0$ and $B_{\textrm{i,j}}<0$,
then this means that species $i$ is input from the environment, while
if $w_{\textrm{j}}>0$ and $B_{\textrm{i,j}}<0$, then this means
that species $i$ is output to the environment. Eq. (\ref{eq:massBalance})
means that $\textrm{production + input = consumption + output}$ for
every species. Henceforth, for brevity, we use the term steady state
to mean either a strict steady state, for molecular species not exchanged
across the boundary of the system, or for species that are exchanged
across the boundary (strictly mass balance). For each metabolite,
whether this means strictly steady state or mass balance, is evident
from the context.

\subsection{\protect\label{subsec:Moiety-conservation}Moiety conservation}

Every genome-scale stoichiometric matrix has linearly dependent rows,
that is $r\coloneqq\textrm{rank}(N)<m$. Let $L\in\mathbb{Z}_{}^{(m-r)\times m}$
denote a left nullspace basis for $N$, that is $L\cdot N=0.$ The
number of linearly dependent rows, or row rank deficiency, is $m-r$.
Each linearly dependent row corresponds to a conserved moiety, which
is a chemical substructure that remains invariant with respect to
the chemical transformations in a given network \cite{ghaderi_structural_2020}.
This \emph{moiety conservation} imposes constraints on the relationship
between an initial species concentration vector at time zero $c_{0}\in\mathbb{R}_{\ge0}^{m}$
and all subsequent species concentrations at time $t$, denoted $c\coloneqq c(t)\in\mathbb{R}_{\ge0}^{m}$,
since
\begin{eqnarray}
\int_{0}^{t}N\cdot v_{net}(c)\,dt & = & c-c_{0}\nonumber \\
0 & = & L\cdot(c-c_{0}).\label{eq:moietyConservation}
\end{eqnarray}
Given a stoichiometric matrix and molecular structures for each species,
using atom mapping and graph theoretical algorithms, it is possible
to compute a \emph{non-negative} left nullspace basis $L\in\mathbb{Z}_{\ge0}^{(m-r)\times m}$
where $L_{\textrm{i,j}}$ is equal to the number of instances of conserved
moiety $i$ in metabolite $j$ \cite{ghaderi_structural_2020,rahou_characterisation_2026}.
Hence, Eq. (\ref{eq:moietyConservation}) is referred to as a moiety
conservation constraint \cite{heinrich_regulation_1996} and $L$
a moiety incidence matrix \cite{rahou_characterisation_2026}.

\subsection{Mass balance subject to elementary reaction kinetics}

If one assumes that a system is at a steady state and that all reactions
must satisfy elementary reaction kinetics, then this requires the
solution to the following system of equations
\begin{equation}
N\cdot(\exp(\ln(k_{f})+F^{T}\cdot\ln(c))-\exp(\ln(k_{r})+R^{T}\cdot\ln(c)))+B\cdot w=0.\label{eq:massBalanceKinetics}
\end{equation}
Due to the exponential and logarithmic terms, it is clear that Eq.
(\ref{eq:massBalanceKinetics}) is non-linear. Furthermore, the set
of all solutions to Eq. (\ref{eq:massBalanceKinetics}) is known not
to be \emph{convex}, i.e., if one draws a straight line between two
points in that set, then there may be points along that line which
are not in that set. This non-convexity makes it challenging to find
solutions to Eq. (\ref{eq:massBalanceKinetics}), e.g., for the high
dimensional models that typically arise from genome-scale metabolic
models. Furthermore, due to various reasons, the uncertainty in our
knowledge of many kinetic parameters is large, so in a modelling context
$k_{f}$ and $k_{r}$ are not fixed parameters, but rather variables
that may be penalised from their deviation to a subset of known experimentally
measured kinetic parameters.

\subsection{Moiety conserved elementary reaction kinetics}

Biochemical networks are open systems that are forced away from equilibrium
by exchange of mass with their environment. Typically this is modelled
with a set of exchange reactions, each of which is a modelling construct
(pseudoreaction) that does not conserve mass and either uptakes species
from the environment or secretes species to the environment. However,
to the best of the authors' knowledge, there exist no conditions established
to guarantee that such a system admits a kinetic steady state. Previously,
we established an approach to force a biochemical network away from
equilibrium such that a non-equilibrium steady state still exists
\cite{fleming_mass_2012}. A corresponding existence theorem is proven
below in terms of ordinary differential equation theory.
\begin{thm}
\label{thm:existenceODE}Let the dynamical equation for mass conserved
elementary kinetics be
\begin{equation}
\dot{c}=(R-F)\cdot\left(k_{f}\odot\exp(F^{T}\cdot\ln(c))-k_{r}\odot\exp(R^{T}\cdot\ln(c))\right),\label{eq:odeKinetics}
\end{equation}
where $c=c(t)\in\mathbb{R}_{>0}^{m}$ is a species concentrations
at time $t>0$, $\dot{c}\in\mathbb{R}^{m}$ is the time derivative
of concentrations and $k_{f},k_{r}\in\mathbb{R}_{\ge0}^{n}$ are non-negative
forward and reverse kinetic parameters and $F,R\in\mathbb{Z}_{\ge0}^{m\times n}$
are forward and reverse stoichiometric matrices. Assuming a finite
and strictly positive initial concentration $c(0)\in\mathbb{R}_{>0}^{m}$,
and the existence of at least one strictly positive vector $\ell\in\mathbb{R}_{>0}^{m}$,
such that
\[
(R-F)^{T}\cdot\ell=0.
\]
then there exists at least one finite and non-negative steady state
concentration $c^{\star}\in\mathbb{R}_{>0}^{m}$, such that $\dot{c}=0$.
\end{thm}

\begin{proof}
Consider\ref{eq:odeKinetics} an autonomous ordinary differential
equation
\begin{equation}
\dot{c}_{\textrm{i}}=\sum_{\textrm{j}=1}^{n}(R_{\textrm{i,j}}-F_{\textrm{i,j}})\left(k_{fj}\prod_{\textrm{i}=1}^{m}c_{\textrm{i}}^{F_{\textrm{i,j}}}-k_{rj}\prod_{\textrm{i}=1}^{m}c_{\textrm{i}}^{R_{\textrm{i,j}}}\right),\qquad c(0)>0.\label{eq:odeMonomials}
\end{equation}
By assumption, the system satisfies concentration non-negativity:
if $c_{\textrm{i}}(t)=0$ for some $i$, then $\dot{c}_{\textrm{i}}(t)\ge0$.
Hence the nonnegative orthant $\mathbb{R}_{\ge0}^{m}$ is forward
invariant. Multiplying(\ref{eq:odeKinetics}) by $\ell\in\mathbb{R}_{>0}^{m}$
yields
\[
\ell^{T}\cdot\dot{c}=\ell^{T}\cdot(R-F)\cdot(\cdot)=0,
\]
and therefore
\begin{equation}
\ell^{T}\cdot c(t)=\ell^{T}\cdot c(0)\quad\forall\,t\ge0.\label{P.3}
\end{equation}
Define
\[
\Omega\coloneqq\{\,c\in\mathbb{R}_{\ge0}^{m}:\ell^{T}\cdot c=\ell^{T}\cdot c(0)\,\}.
\]
Because $\ell\succ0$, each component satisfies
\[
0\le c_{\textrm{i}}\le\frac{\ell^{T}\cdot c(0)}{\ell_{\textrm{i}}},
\]
so $\Omega$ is nonempty, closed, bounded, and convex, hence compact.
By(\ref{P.3}) and forward invariance of $\mathbb{R}_{\ge0}^{m}$,
$\Omega$ is invariant under the dynamics.

Write the right-hand side of(\ref{eq:odeMonomials}) as the vector
field $g(c)\coloneqq(R-F)\cdot\left(k_{f}\odot\exp(F^{T}\cdot\ln(c))-k_{r}\odot\exp(R^{T}\cdot\ln(c))\right)$.
The reaction rates, written componentwise as monomials in(\ref{eq:odeMonomials}),
extend continuously to $\mathbb{R}_{\ge0}^{m}$, so $g$ is continuous
on $\Omega$. Since these monomials are smooth on $\mathbb{R}_{>0}^{m}$,
solutions are unique, and for each $\tau>0$ the time-$\tau$ flow
map
\[
\Phi_{\tau}:\Omega\to\Omega,\qquad\Phi_{\tau}(c_{0})=c(\tau),
\]
is a well-defined continuous self-map of $\Omega$.

Fix a sequence $\tau_{n}\downarrow0$. For each $n$, $\Phi_{\tau_{n}}$
maps the compact convex set $\Omega$ continuously into itself, so
Brouwer's fixed point theorem yields $c_{n}\in\Omega$ with
\begin{equation}
\Phi_{\tau_{n}}(c_{n})=c_{n}.\label{P.4}
\end{equation}
A point fixed by $\Phi_{\tau_{n}}$ has a $\tau_{n}$-periodic orbit
and is not yet a steady state, since any nonconstant orbit whose period
divides $\tau_{n}$ is also fixed by $\Phi_{\tau_{n}}$. To extract
an equilibrium, use compactness of $\Omega$ to pass to a subsequence
with $c_{n}\to c^{\star}\in\Omega$. In integral form,
\[
0=\Phi_{\tau_{n}}(c_{n})-c_{n}=\int_{0}^{\tau_{n}}g\!\left(\Phi_{s}(c_{n})\right)ds,
\]
so dividing by $\tau_{n}$ gives the time average
\begin{equation}
\frac{1}{\tau_{n}}\int_{0}^{\tau_{n}}g\!\left(\Phi_{s}(c_{n})\right)ds=0.\label{P.5}
\end{equation}
Because $g$ is bounded on the compact set $\Omega$, for every $s\in[0,\tau_{n}]$
\[
\left\Vert \Phi_{s}(c_{n})-c^{\star}\right\Vert \le\tau_{n}\,\sup_{c\in\Omega}\left\Vert g(c)\right\Vert +\left\Vert c_{n}-c^{\star}\right\Vert \longrightarrow0
\]
uniformly in $s$ as $n\to\infty$. Since $g$ is uniformly continuous
on $\Omega$, the integrand converges uniformly to $g(c^{\star})$,
and hence the average in(\ref{P.5}) converges to $g(c^{\star})$.
Therefore 
\[
g(c^{\star})=0,\qquad\text{i.e.}\qquad\dot{c}^{\star}=0,
\]
so $c^{\star}$ is a steady state of (\ref{eq:odeMonomials}), with
$c^{\star}\succeq0$ and finite (cf.. Figure \ref{fig:Thm1} for an
illustration of this argument). This completes the proof.
\end{proof}
\begin{figure}
\includegraphics[width=1\textwidth]{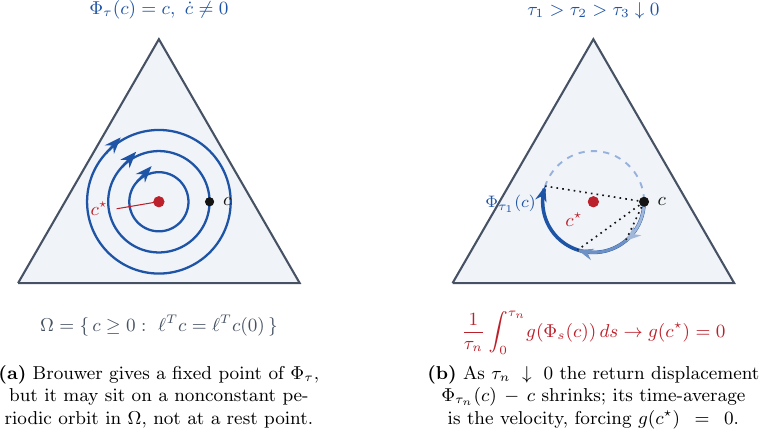}\caption{\protect\label{fig:Thm1}Illustration of a key part of the proof of
Theorem \ref{thm:existenceODE} (a) A single-time statement $\Phi_{\tau}(c^{\star})=c^{\star}$
for one fixed $\tau>0$ is not by itself enough to conclude $\dot{c}^{\star}=0$.
It asserts only that the orbit through $c^{\star}$ returns to $c^{\star}$
after time $\tau$, i.e. that the orbit is $\tau$-periodic. A nonconstant
periodic orbit whose period divides $\tau$ satisfies the same fixed-point
equation while having $\dot{c}\protect\neq0$ throughout, so a fixed
point of $\Phi_{\tau}$ need not be an equilibrium. (b) Letting $\tau_{n}\downarrow0$
removes this gap: the time-averaged velocity over $[0,\tau_{n}]$
vanishes at each $c_{n}$, and in the limit this average becomes the
instantaneous velocity $g(c^{\star})$, forcing $g(c^{\star})=0$.}
\end{figure}

\subsection{\protect\label{subsec:Thermodynamic-constraints-on-kfkr}Thermodynamically
feasible kinetic parameters}

We assume chemical potential $u\in\mathbb{R}^{m}$ is 
\[
u\coloneqq u^{\circ}+\mathcal{R}\mathcal{T}\ln\left(c\right)
\]
where $u^{\circ}\in\mathbb{R}^{m}$ is a vector of standard chemical
potentials. This is a simplification of chemical potential in biochemical
thermodynamics, but the mathematical form is the same for a variety
of more sophisticated formulations (cf Appendix \ref{subsec:Chemical-potential})
. The change in chemical potential for a system of biochemical reactions
is denoted
\[
\Delta u\coloneqq N^{T}\cdot u.
\]
At thermodynamic equilibrium for all reactions, the sum of substrate
chemical potentials equals the sum of product chemical potentials
for each reaction, so the change in chemical potential for all reactions
is zero and therefore 
\[
N^{T}\cdot u=0\Leftrightarrow\frac{-N^{T}\cdot u^{\circ}}{\mathcal{R}\mathcal{T}}=N^{T}\cdot\ln(c_{eq}),
\]
where $c_{eq}\in\mathbb{\mathbb{R}}^{m}$ is a vector of species concentrations
at equilibrium. At thermodynamic equilibrium, without a driving force,
the forward and reverse elementary reaction rates must be equal \emph{(detailed
balance} \cite{berry_physical_2000}). This requirement means that
elementary kinetic parameters are thermodynamically constrained, as
the following sequence of algebraic steps show
\begin{eqnarray}
v_{f}(c_{eq}|k_{f})-v_{r}(c_{eq}|k_{r}) & = & 0\nonumber \\
v_{f}(c_{eq}|k_{f}) & = & v_{r}(c_{eq}|k_{r})\nonumber \\
\exp(\ln(k_{f})+F^{T}\cdot\ln(c_{eq})) & = & \exp(\ln(k_{r})+R^{T}\cdot\ln(c_{eq}))\nonumber \\
\ln(k_{f})+F^{T}\cdot\ln(c_{eq}) & = & \ln(k_{r})+R^{T}\cdot\ln(c_{eq})\nonumber \\
\ln(k_{f})-\ln(k_{r}) & = & R^{T}\cdot\ln(c_{eq})-F^{T}\cdot\ln(c_{eq})\nonumber \\
\ln\left(\frac{k_{f}}{k_{r}}\right) & = & (-F+R)^{T}\cdot\ln(c_{eq})\nonumber \\
\ln\left(\frac{k_{f}}{k_{r}}\right) & = & N^{T}\cdot\ln(c_{eq})\nonumber \\
\ln\left(\frac{k_{f}}{k_{r}}\right) & = & -N^{T}\cdot\frac{u^{\circ}}{\mathcal{R}\mathcal{T}}\label{eq:thermoFeasibleKineticParam}
\end{eqnarray}
Elementary reaction kinetics (\ref{eq:elementaryKineticsR}) coupled
with the thermodynamic constraints in Eq. (\ref{eq:thermoFeasibleKineticParam})
is referred to as \emph{mass action kinetics}, represented by the
following pair of equation systems
\begin{eqnarray}
\dot{c} & = & N\cdot(\exp(\ln(k_{f})+F^{T}\cdot\ln(c))-\exp(\ln(k_{r})+R^{T}\cdot\ln(c))),\nonumber \\
\ln\left(\frac{k_{f}}{k_{r}}\right) & = & -N^{T}\cdot\frac{u^{\circ}}{\mathcal{R}\mathcal{T}}.\label{eq:steadyStateMassActionKinetics}
\end{eqnarray}

\subsection{Thermodynamic constraints on reaction rates}

Thermodynamic constraints on kinetic reactions imply thermodynamic
constraints on reaction rates. To observe this implication, start
with the definition of elementary reaction kinetics for the forward
and reverse rates in Eq. (\ref{eq:elementaryKineticsF}) and (\ref{eq:elementaryKineticsR})
and let $v_{f}=v_{f}(c\mid k_{f})$ and $v_{r}=v_{r}(c\mid k_{r})$,
so we have
\begin{eqnarray*}
\frac{v_{f}}{v_{r}} & = & \frac{\exp(\ln(k_{f})+F^{T}\cdot\ln(c))}{\exp(\ln(k_{r})+R^{T}\cdot\ln(c))}\\
 & = & \exp(\ln(k_{f})-\ln(k_{r})+F^{T}\cdot\ln(c)-R^{T}\cdot\ln(c))\\
 & = & \exp(\ln(k_{f})-\ln(k_{r}))\odot\exp(F^{T}\cdot\ln(c)-R^{T}\cdot\ln(c))\\
 & = & \exp(\ln(k_{f})-\ln(k_{r}))\odot\exp(-N^{T}\cdot\ln(c))\\
 & = & \left(\frac{k_{f}}{k_{r}}\right)\odot\exp(-N^{T}\cdot\ln(c)).
\end{eqnarray*}
Taking the logarithm of both sides, we have 
\begin{eqnarray*}
\ln\left(\frac{v_{f}}{v_{r}}\right) & = & \ln\left(\frac{k_{f}}{k_{r}}\right)-N^{T}\cdot\ln(c).
\end{eqnarray*}
Using the thermodynamic constraints on kinetic parameters in (\ref{eq:thermoFeasibleKineticParam})
we observe that

\begin{eqnarray*}
\ln\left(\frac{v_{f}}{v_{r}}\right) & = & -N^{T}\cdot\frac{u^{\circ}}{\mathcal{R}\mathcal{T}}-N^{T}\cdot\ln(c),\\
\mathcal{R}\mathcal{T}\ln\left(\frac{v_{f}}{v_{r}}\right) & = & -N^{T}\cdot(u^{\circ}+\mathcal{R}\mathcal{T}\ln(c)).
\end{eqnarray*}
Using the definition of chemical potential in Eq. (\ref{eq:chemicalPotential})
we then obtain

\begin{eqnarray}
\mathcal{R}\mathcal{T}\ln\left(\frac{v_{f}}{v_{r}}\right) & = & -N^{T}\cdot u\label{eq:logvfvrNtu}
\end{eqnarray}
which is a thermodynamic constraint on reaction rates that must hold
for any pair of forward and reverse rates and any potential at a given
instance, regardless of whether a system is at equilibrium or not,
and regardless if a system is in a dynamic or steady state. Eq. (\ref{eq:logvfvrNtu})
incorporates a representation of energy conservation, since each species
is assigned a single chemical potential. Eq. (\ref{eq:logvfvrNtu})
also incorporates a representation of the second law of thermodynamics,
since 
\begin{eqnarray*}
v_{f}=v_{r} & \Leftrightarrow & N^{T}\cdot u=0\\
v_{f}>v_{r} & \Leftrightarrow & N^{T}\cdot u<0\\
v_{f}<v_{r} & \Leftrightarrow & N^{T}\cdot u>0
\end{eqnarray*}
which ensures that the net rate of each reaction is zero at thermodynamic
equilibrium, and away from thermodynamic equilibrium the sign of net
rate is opposite to the sign of change in chemical potential, i.e.
, net rate is down a gradient in chemical potential.

\subsection{A thermodynamically open system admitting an elementary kinetic steady
state}

Given a stoichiometric matrix $N\in\mathbb{Z}^{m\times n}$ and a
moiety incidence matrix $L\in\mathbb{Z}_{\ge0}^{(m-r)\times m}$ consider
the \emph{cyclic stoichiometric matrix} \cite{ghaderi_structural_2020}
\begin{equation}
C\coloneqq\left[\begin{array}{cc}
N & -I\\
0 & L
\end{array}\right]\in\mathbb{Z}^{\left(2m-r\right)\times\left(n+m\right)}\label{eq:cyclicStoichiometricMatrix}
\end{equation}
where each of the additional $m$ columns is termed a \emph{perpetireaction},
which is a modelling construct that represents the transformation
of a single metabolite into its constituent conserved moieties. Appendix
\ref{sec:Cyclic-stoichiometric-matrix} establishes that $\textrm{rank}\left(C\right)=m$,
the matrix
\begin{equation}
\left[\begin{array}{c}
I_{n}\\
N
\end{array}\right]\in\mathbb{Z}^{(n+m)\times n}\label{eq:nullBasisC}
\end{equation}
is a basis for the right nullspace of $C$, and the matrix $\left[\begin{array}{cc}
L, & I_{m-r}\end{array}\right]\in\mathbb{N}^{(m-r)\times(2m-r)}$ is a basis for the left nullspace of $C$. Since every row of \ref{eq:nullBasisC}
contains at least one non-zero, every reaction in $C$ participates
in at least one vector in the nullspace of $C$. We shall refer back
to this property in a subsequent section. Since $\left[\begin{array}{cc}
L & I\end{array}\right]^{T}1>0$ every reaction in $C$ is stoichiometrically consistent \cite{gevorgyan_detection_2008}
so by Theorem \ref{thm:existenceODE}, the corresponding elementary
kinetic system admits at least one steady state.

Assume a system of chemical reactions defined by a cyclic stoichiometric
matrix \ref{eq:cyclicStoichiometricMatrix} with corresponding parameter
vectors $k_{f},k_{r}\in\mathbb{R}^{n}$ and $p_{f},p_{r}\in\mathbb{R}^{m}$.
If both are thermodynamically feasible then the Hadamard divisors
of forward and reverse parameters form a vector in the range of the
cyclic stoichiometric matrix, $\mathcal{R}(C^{T})$. That is, there
exists a $u\in\mathbb{R}^{m}$ and $q\in\mathbb{R}^{m-r}$ such that
\begin{eqnarray}
\left[\begin{array}{c}
\ln\left(k_{f}\oslash k_{r}\right)\\
\ln\left(p_{f}\oslash p_{r}\right)
\end{array}\right] & = & \left[\begin{array}{cc}
N & -I\\
0 & L
\end{array}\right]^{T}\cdot\left[\begin{array}{c}
u\\
q
\end{array}\right]\in\mathcal{R}(C^{T}).\label{eq:thermoFeasibleCyclic}
\end{eqnarray}
Equivalently

\begin{eqnarray*}
\ln\left(k_{f}\oslash k_{r}\right) & = & N^{T}\cdot u,\\
\ln\left(p_{f}\oslash p_{r}\right) & = & -u+L^{T}\cdot q.
\end{eqnarray*}
In contrast, if both $k_{f},k_{r}$ and $p_{f},p_{r}$ are thermodynamically
infeasible then the Hadamard divisors of forward and reverse parameters
form a vector in the nullspace of the cyclic stoichiometric matrix
$\mathcal{N}(C)$. That is

\begin{eqnarray*}
\left[\begin{array}{cc}
N & -I\\
0 & L
\end{array}\right]\cdot\left[\begin{array}{c}
\ln\left(k_{f}\oslash k_{r}\right)\\
\ln\left(p_{f}\oslash p_{r}\right)
\end{array}\right] & = & 0,
\end{eqnarray*}
or equivalently
\begin{eqnarray*}
N\ln\left(k_{f}\oslash k_{r}\right) & = & \ln\left(p_{f}\oslash p_{r}\right),\\
L\ln\left(p_{f}\oslash p_{r}\right) & = & 0.
\end{eqnarray*}
However, if $k_{f},k_{r}$ are thermodynamically feasible but $p_{f},p_{r}$
are thermodynamically infeasible, then there exists a $u\in\mathbb{R}^{m}$
such that
\begin{eqnarray*}
\ln\left(k_{f}\oslash k_{r}\right) & = & N^{T}\cdot u,\\
L\ln\left(p_{f}\oslash p_{r}\right) & = & 0.
\end{eqnarray*}

Let $\bar{N}\in\mathbb{Z}^{m\times r}$ denote a basis for the range
of the internal reactions $\mathcal{R}(N)$, then

\begin{eqnarray}
\ln\left(p_{f}\oslash p_{r}\right)=\bar{N}w & \Rightarrow & L\ln\left(p_{f}\oslash p_{r}\right)=0\label{eq:thermoInfeasiblePerpeti}
\end{eqnarray}
where $w\in\mathbb{R}^{r}$. Without loss of generality, let $\ln(p_{r})=0$,
therefore $\ln\left(p_{f}\right)=\bar{N}w.$ Let $\ell\in\mathbb{R}_{>0}^{m-r}$
denote the concentration of each conserved moiety. Given a cyclic
stoichiometric matrix (\ref{eq:cyclicStoichiometricMatrix}), with
the aforementioned parameterisation, the dynamical equation for the
corresponding elementary reaction kinetic system is
\begin{eqnarray}
\left[\begin{array}{cc}
N & -I\\
0 & L
\end{array}\right]\exp\left(\left[\begin{array}{c}
\ln(k_{f})\\
\ln(p_{f})
\end{array}\right]+\left[\begin{array}{cc}
F & I\\
0 & 0
\end{array}\right]^{T}\cdot\left[\begin{array}{c}
\ln(c)\\
\ln(\ell)
\end{array}\right]\right)\ldots\label{eq:cyclicDynamicalSystem}\\
-\left[\begin{array}{cc}
N & -I\\
0 & L
\end{array}\right]\exp\left(\left[\begin{array}{c}
\ln(k_{r})\\
\ln(p_{r})
\end{array}\right]+\left[\begin{array}{cc}
R & 0\\
0 & L
\end{array}\right]^{T}\cdot\left[\begin{array}{c}
\ln(c)\\
\ln(\ell)
\end{array}\right]\right) & = & \left[\begin{array}{c}
\dot{c}\\
\dot{\ell}
\end{array}\right],
\end{eqnarray}
and steady states satisfy
\begin{eqnarray}
(R-F)\cdot\left(\exp(\ln(k_{f})+F^{T}\cdot\ln(c))-\exp(\ln(k_{r})+R^{T}\cdot\ln(c))\right)\ldots\label{eq:cyclicSteadyState1}\\
-(R-F)\left(\exp(\ln(p_{f})+\ln(c))+\exp(\ln(p_{r})+L^{T}\cdot\ln(\ell))\right) & = & 0,\\
L\left(\exp(\ln(p_{f})+\ln(c))-\exp(\ln(p_{r})+L^{T}\cdot\ln(\ell))\right) & = & 0.\label{eq:cyclicSteadyState2}
\end{eqnarray}
By Theorem (\ref{thm:existenceODE}), we are assured a solution to
(\ref{eq:cyclicSteadyState1}) and (\ref{eq:cyclicSteadyState2})
exists. Moreover, since (\ref{eq:cyclicSteadyState2}) is implied
by (\ref{eq:cyclicSteadyState1}), then (\ref{eq:cyclicSteadyState2})
is redundant.

We now show that thermodynamic infeasibility of the exchange parameters
forces the guaranteed steady state out of equilibrium. A system of
cyclic mass-action kinetics is detailed balanced when every forward
elementary rate equals its reverse, $v_{f}=v_{r}$. The following
result shows that this cannot occur once the perpeti parameters violate
the Wegscheider conditions, even when the internal kinetic parameters
$k_{f},k_{r}$ remain thermodynamically feasible.
\begin{thm}
\label{thm:noDetailedBalance}Consider the cyclic stoichiometric system(\ref{eq:cyclicStoichiometricMatrix})
with parameters $k_{f},k_{r}\in\mathbb{R}^{n}$ and $p_{f},p_{r}\in\mathbb{R}^{m}$,
and let a steady state be guaranteed by Theorem \ref{thm:existenceODE}.
Assume $k_{f},k_{r}$ are thermodynamically feasible but $p_{f},p_{r}$
are thermodynamically infeasible, so that there exist $u\in\mathbb{R}^{m}$
and $w\in\mathbb{R}^{r}$ such that
\[
\ln\left(k_{f}\oslash k_{r}\right)=N^{T}\cdot u,\qquad\ln\left(p_{f}\oslash p_{r}\right)=\bar{N}w,
\]
where $\bar{N}$ is a basis for the range $\mathcal{R}(N)$. Then
the steady state does not satisfy detailed balance; that is, $v_{f}\neq v_{r}$.
\end{thm}

\begin{proof}
By (\ref{eq:thermoFeasibleCyclic}) the cyclic system is thermodynamically
feasible exactly when the combined log-parameter vector lies in $\mathcal{R}(C^{T})$.
Since (\ref{eq:nullBasisC}) is a basis for the right nullspace of
$C$ and $\mathcal{R}(C^{T})=\mathcal{N}(C)^{\perp}$, this membership
is equivalent to the single cycle condition obtained by pairing the
parameter vector with that basis,
\begin{equation}
a\coloneqq\left[\begin{array}{c}
I_{n}\\
N
\end{array}\right]^{T}\cdot\left[\begin{array}{c}
\ln\left(k_{f}\oslash k_{r}\right)\\
\ln\left(p_{f}\oslash p_{r}\right)
\end{array}\right]=\ln\left(k_{f}\oslash k_{r}\right)+N^{T}\cdot\ln\left(p_{f}\oslash p_{r}\right)=0.\label{eq:cycleAffinity}
\end{equation}
With $\ln\left(k_{f}\oslash k_{r}\right)=N^{T}\cdot u$ and $\ln\left(p_{f}\oslash p_{r}\right)=\bar{N}w$,
condition (\ref{eq:cycleAffinity}) reads $N^{T}\cdot(u+\bar{N}w)=0$;
the assumed thermodynamic infeasibility of $p_{f},p_{r}$ is precisely
the failure of this identity, so $a\neq0$. Suppose now, to the contrary,
that the steady state were detailed balanced, that is $v_{f}=v_{r}$.
Then the internal net flux $v\coloneqq v_{f}-v_{r}$ vanishes, and
the species balance (\ref{eq:cyclicSteadyState1}), which reads $Nv=b$
for the perpeti net flux $b\coloneqq p_{f}\odot c-p_{r}\odot\exp(L^{T}\cdot\ln\ell)$,
forces $b=0$ as well. Every forward rate then equals its reverse,
so there is a state $(c,\ell)$ with $\ln\left(k_{f}\oslash k_{r}\right)=N^{T}\cdot\ln c$
and $\ln\left(p_{f}\oslash p_{r}\right)=-\ln c+L^{T}\cdot\ln\ell$;
equivalently
\[
\left[\begin{array}{c}
\ln\left(k_{f}\oslash k_{r}\right)\\
\ln\left(p_{f}\oslash p_{r}\right)
\end{array}\right]=C^{T}\cdot\left[\begin{array}{c}
\ln c\\
\ln\ell
\end{array}\right]\in\mathcal{R}(C^{T}).
\]
Pairing this with the nullspace basis(\ref{eq:nullBasisC}) and using
$C\left[I_{n};N\right]=0$ gives $a=\left[I_{n};N\right]^{T}\cdot C^{T}\cdot\left[\ln c;\ln\ell\right]=\left(C\left[I_{n};N\right]\right)^{T}\cdot\left[\ln c;\ln\ell\right]=0$,
contradicting $a\neq0$. Hence $v_{f}\neq v_{r}$, so the internal
reactions carry a nonzero net flux and the steady state does not satisfy
detailed balance.
\end{proof}
The obstruction is the single cycle affinity $a$ in (\ref{eq:cycleAffinity}):
thermodynamic feasibility of the full cyclic system is equivalent
to $a=0$. By (\ref{eq:thermoInfeasiblePerpeti}) the infeasible exchange
parameters satisfy $L\ln\left(p_{f}\oslash p_{r}\right)=0$, so $\ln\left(p_{f}\oslash p_{r}\right)\in\mathcal{R}(N)$.
The entire thermodynamic driving therefore lies within the internal
cycle space $\mathcal{R}(N)$, and it is the perpeti reactions that
sustain the resulting nonequilibrium steady state. Note that detailed
balance is the condition $v_{f}=v_{r}$; Theorem\ref{thm:noDetailedBalance}
asserts its negation.

The strictly positive conservation vector required by Theorem\ref{thm:existenceODE}
is furnished by the cyclic construction itself. Every left-null vector
of $C$ has the form $\ell=\left[L^{T}\cdot y;\,y\right]$ with $y\in\mathbb{R}^{m-r}$:
its metabolite block is $L^{T}\cdot y$ and its moiety block is $y$.
Choosing $y>0$, and recalling that $L\geq0$ with every metabolite
belonging to at least one conserved moiety, so that each column of
$L$ contains a positive entry, one obtains $L^{T}\cdot y>0$ and
hence $\ell>0$ strictly. The cyclic matrix $C$ is therefore constructed
precisely so that the strictly positive left-null vector demanded
by Theorem\ref{thm:existenceODE} is guaranteed to exist, and this
vector depends only on the structural matrices $N$ and $L$, not
on any kinetic parameter values. In particular, arbitrarily setting
the parameters $p_{f},p_{r}$ but keeping $C$, $L$ or $N$ invariant,
means $\ell>0$ is preserved and a steady state continues to exist.
Consistently with Theorem\ref{thm:noDetailedBalance}, thermodynamic
infeasibility of the perpeti parameters is not merely permitted but
is the mechanism that drives this guaranteed steady state away from
equilibrium, so that the internal reactions carry a nonzero net flux.

\subsection{Mathematical classification of elementary reaction kinetics}

Classification of a function in mathematical terms is important when
one seeks to identify whether there exist established algorithms and
software that either enables one to obtain a numerical solution that
is in the zero set of that function, or numerically optimise over
the zero set of that function. This section attempts to mathematically
classify the function (\ref{eq:dcdt-1-1-1}), and can be omitted on
a first pass, but is a topic we shall return to in the discussion.
Recall that $Q\in\mathbb{R}^{m\times m}$ is positive definite if
$x^{T}\cdot Q\cdot x>0$ for all $x\in\mathbb{R}^{m}\ne0$ while $A$
is indefinite if there exist $x_{1}$ and $x_{2}$ such that $x_{1}^{T}\cdot A\cdot x_{1}>0$
and $x_{2}^{T}\cdot Ax_{2}<0$.

Given the cyclic stoichiometric matrix in \ref{eq:cyclicStoichiometricMatrix},
the corresponding forward and reverse stoichiometric matrices are
denoted
\[
\bar{F}\;\coloneqq\;\begin{bmatrix}F & I_{m}\\
0 & 0
\end{bmatrix}\qquad\bar{R}\;\coloneqq\;\begin{bmatrix}R & 0\\
0 & L
\end{bmatrix}\in\mathbb{R}^{(2m-r)\times(n+m)}.
\]
 Consider the following coordinate transformations 
\begin{eqnarray*}
k & \coloneqq & [\ln(k_{f});\ln(p_{f});\ln(k_{r});\ln(p_{r})],\\
x & \coloneqq & [\ln(c);\ln(\ell)].
\end{eqnarray*}
Therefore the set of steady sates is

\begin{eqnarray}
f(x) & \coloneqq & \left([\bar{R},\bar{F}]-[\bar{F},\bar{R}]\right)\cdot\exp(k+[\bar{F},\bar{R}]^{T}\cdot x)=0,\label{eq:f(x)fundamental}\\
 & = & [\bar{R},\bar{F}]\cdot\exp(k+[\bar{F},\bar{R}]^{T}\cdot x)-[\bar{F},\bar{R}]\cdot\exp(k+[\bar{F},\bar{R}]^{T}\cdot x)=0.\nonumber 
\end{eqnarray}

By definition $[\bar{F},\bar{R}]$ has non-negative entries. Furthermore,
under biochemically realistic assumptions \parencite{fleming_conditions_2016},
$[\bar{F},\bar{R}]$ is full row rank. This is a form of kinetic consistency,
in the sense that the stoichiometric signatures of the metabolites
across the forward and reverse elementary reactions, are linearly
independent, and it underpins the duality between fluxes and concentrations
in biochemical networks \parencite{fleming_conditions_2016}. Define
the following split of $f(x)$ into two parts
\[
f(x)\coloneqq\nabla\varphi(x)-a(x)
\]
where

\begin{eqnarray*}
\varphi(x) & := & \mathbf{1}^{T}\cdot\exp(k+[\bar{F},\bar{R}]^{T}\cdot x),\\
\nabla\varphi(x) & = & [\bar{F},\bar{R}]\cdot\exp(k+[\bar{F},\bar{R}]^{T}\cdot x)>0,\\
\nabla^{2}\varphi(x) & = & [\bar{F},\bar{R}]\cdot\textrm{diag}\left(\exp\left(k+[\bar{F},\bar{R}]^{T}\cdot x\right)\right)\cdot[\bar{F},\bar{R}]^{T},\\
a(x) & := & [\bar{R},\bar{F}]\cdot\exp(k+[\bar{F},\bar{R}]^{T}\cdot x)>0,\\
\nabla a(x) & = & [\bar{F},\bar{R}]\cdot\textrm{diag}\left(\exp\left(k+[\bar{F},\bar{R}]^{T}\cdot x\right)\right)\cdot[\bar{R},\bar{F}]^{T},\\
\nabla f(x) & = & \nabla^{2}\varphi(x)-\nabla a(x).
\end{eqnarray*}
Since $\exp$ is a convex function and $[\bar{F},\bar{R}]$ is full
row rank then $\varphi(x)$ is a strictly convex function of $x$,
so its gradient $\nabla\varphi(x)\in\mathbb{R}_{\ge0}^{m}$ is strictly
monotone and is equal to the rate of consumption of each species,
and $\nabla^{2}\varphi(x)\in\mathbb{R}^{m\times m}$ is a symmetric
positive definite matrix. Full row rank of $[\bar{F},\bar{R}]$ makes
the Hessian $\nabla^{2}\varphi(x)$ is positive definite rather than
merely positive semidefinite. In contrast, $a(x)\in\mathbb{R}_{\ge0}^{m}$
is the rate of production of each species and $\nabla a(x)\in\mathbb{R}^{m\times m}$
is an asymmetric indefinite matrix, so it is not the gradient of any
scalar valued function. Also, the Hessian of a strictly convex function
is positive definite and vice versa. Therefore $\nabla f(x)$ is the
difference between a positive definite and an asymmetric square matrix,
which may be indefinite, in which case $\nabla f(x)$ is not the Hessian
of any convex function \cite{boyd_convex_2004}. This makes it difficult
to solve for $x$ such that $f(x)=0$ with established algorithms
using convex optimisation or monotone variational inequality theory
\parencite{rockafellar_variational_1998}. For a restricted class
of biochemical networks, previously we established that $f(\cdot)$
is not monotone, however it is \emph{duplomonotone}, that is $f(x)^{T}\cdot\nabla f(x)\cdot f(x)\ge0$.
This enabled formulation of a globally convergent algorithm whose
stationary states solve $f(x)=0$, but it is not known if $f(\cdot)$
is duplomonotone in general \cite{artacho_globally_2014}. Moreover,
even if $f(\cdot)$ is duplomonotone, there still remains the unsolved
problem to optimise over the set $\Omega\coloneqq\{x\,|\,f(x)=0\}$.

Note that, to express concentration dynamics in the logarithmic coordinate,
differentiate $c=\exp(x)$ component-wise
\[
\dot{c}=\exp(x)\odot\dot{x}=\exp(x)\odot\dot{x}.
\]
 Because $\exp(x)$ is componentwise invertible the velocity field
in the logarithmic coordinate is
\[
\dot{x}=\exp(x)\odot\cdot\dot{c}=-\exp(x)\odot f(x).
\]
The steady-state set is identical in both coordinates,
\[
\dot{c}=0\quad\iff\quad\dot{x}=0\quad\iff\quad f(x)=0.
\]
As such, up to coordinate transformation, $f(x)=0$ is the fundamental
equation defining the set of steady states of a network.

\section{Constraint-based modelling of biochemical networks}

Instead of trying to directly find a solution to Eq (\ref{eq:massBalanceKinetics}),
or in addition Eq. (\ref{eq:steadyStateMassActionKinetics}), the
field of constraint-based modelling of biochemical networks arose,
whereby a subset of the constraints represented by (\ref{eq:steadyStateMassActionKinetics})
are either removed or relaxed, and an optimisation problem was formulated
to select an optimal vector that also satisfied the resulting simplified
system of equations. The motivation for introducing the particular
optimisation problems in this section is that each retains a subset
of the (primal) constraints on a kinetically feasible steady state
but their objectives are different and each of their optimality conditions
correspond to different subsets of the feasible set of kinetically
feasible steady states. In particular, as explained in Section \ref{sec:vkSteady},
Theorem \ref{thm:stationaryIsSteady} establishes that these subsets
form a nested chain of sets $\mathcal{S}_{4}\subseteq\mathcal{S}_{3}\subseteq\mathcal{S}_{2}\subseteq\mathcal{S}_{1}$,
with the variational kinetics algorithm optimising towards the innermost
set $\mathcal{S}_{4}$ of elementary kinetic steady states. This is
further interpreted in Section \ref{sec:vkSteady}. The first such
optimisation problem was a linear optimisation problem.

\subsection{Linear optimisation: Introduction}

The standard form of a linear optimisation problem is

\begin{eqnarray}
\underset{x}{\text{min}}\;c^{T}\cdot x\nonumber \\
\text{s.t.}\;A\cdot x\le b &  & \textrm{}\label{eq:LP}
\end{eqnarray}
where $c\in\mathbb{R}^{n}$ is a coefficient vector, \emph{$c^{T}\cdot x$}
is a \emph{linear} \emph{objective function,} $x\in\mathbb{R}^{n}$
is a vector of variables $c\in\mathbb{R}^{n}$ is a given vector of
coefficients, $A\in\mathbb{R}^{m\times n}$ is a given linear constraint
matrix and $b\in\mathbb{R}^{m}$ is a given vector of data. The constraints
in Eq. (\ref{eq:LP}) define a polyhedral convex set, which may either
be empty (no solution exists), admit one solution (hence no need for
an optimisation problem) or admit an infinite number of solutions
(well posed optimisation problem). Optimisation of the\emph{ }objective
function is expressed as minimisation by convention and results in
identification of an optimal vector $x^{\star}$ wherein the value
of $c^{T}\cdot x^{\star}$ is minimal, i.e. there does not exist another
vector satisfying $A\cdot x\le b$ such that the value of the objective
is any less. Two distinct vectors $x_{1}^{\star}$ and $x_{2}^{\star}$
are referred to as alternate optimal solutions if $c^{T}\cdot x_{1}^{\star}=c^{T}\cdot x_{2}^{\star}$.
In general, there exist an infinite number of optimal vectors $x^{\star}$
that each have the same minimal value of the\emph{ }objective function.\footnote{Such problems are solved using a numerical optimisation solver that
represents real valued numbers in finite precision, so typically a
value of $x_{\textrm{i}}^{\star}=10^{-6}$ or less should be considered
zero (though this depends on the optimisation solver).}

\subsection{Flux balance analysis}

Rather than modelling a system of reactions in terms of mass action
kinetics, an alternative approach is to represent net reaction rate
as a variable rather than as a function of kinetic parameters and
species concentrations. Recall that when assuming elementary reaction
kinetics, net rate is the difference between a pair of unidirectional
reactions, each of which is an explicit function of kinetic parameters
and species concentrations, cf Eq. (\ref{eq:netElementaryReactionFlux}).
An advantage of this approach is that the set of net rate vectors
that satisfy steady state is a polyhedral convex set. This enables
one to formulate and efficiently solve various optimisation problems
where the solution is a rate vector that is optimal with respect to
a particular objective function.

When the objective function is linear, this approach is known as \emph{flux
balance analysis} and is represented by the optimisation problem

\begin{eqnarray}
\underset{v,w}{\text{min}}\;c^{T}\cdot v+d^{T}\cdot w\nonumber \\
\text{s.t.}\;N\cdot v+B\cdot w=0 &  & \textrm{}\label{eq:FBA}\\
l\le v\le u,
\end{eqnarray}
where $v\in\mathbb{R}^{n}$ is a vector of net rates, one for each
internal reaction in the biochemical system and $c\in\mathbb{R}^{n}$
is a given vector of linear objective coefficients, one for each internal
reaction. Furthermore, to represent the exchange of mass between the
system and its environment $w\in\mathbb{R}^{k}$ is a vector of net
rates, one for each external reaction and $d\in\mathbb{R}^{k}$ is
a given vector of linear objective coefficients, one for each exchange
reaction\emph{.} Often, in applications $c=0$ and only one $d_{\textrm{i}}\ne0$.

As before, $N\in\mathbb{R}^{m\times n}$ is a stoichiometric matrix
representing internal reactions and $B$ is a stoichiometric matrix
representing external reactions. Compare the first constraint in Problem
(\ref{eq:FBA}) with the steady state constraint in Eq. (\ref{eq:massBalance}).
Both seem similar but they are fundamentally different as in Eq. (\ref{eq:massBalance})
net rate is a function of concentrations and kinetic parameters, while
in Problem (\ref{eq:FBA}) net rate is a variable vector that is not
a function of any other variables or parameters, which makes it substantially
less constrained.

In Problem (\ref{eq:FBA}) the last constraint is a set of box constraints,
represented by lower and upper bounds on the net rate for each reaction,
$l,u\in\mathbb{R}^{n}$. Qualitatively, these bounds may be set based
on known directionality of biochemical reactions. The convention is
that net rate is positive for a reaction that proceeds from substrates
to products in a left to right direction when expressed as a reaction
equation. In most metabolic networks, irreversible reactions are represented
with bounds like $0\le v\le u$ so net rate must be positive. Reversible
reactions are represented as $-\infty\le v\le\infty$ which means
that box constraint is not active. Quantitatively such bounds may
be set based on experimentally measured reaction rates, e.g., in an
\emph{in vitro} culture, by measuring the difference between fresh
and spent medium concentrations, one may estimate the rate of uptake
or secretion of a metabolite by a system, which can be used to set
quantitative bounds on exchange reaction rates. Generally, there are
also box constraints on net external reaction rates.

\subsection{Cycle-free flux balance analysis}

Inspired by Desouki et al. \cite{desouki_cyclefreeflux_2015}, we
previously observed \cite{fleming_cardinality_2023} that a thermodynamically
feasible flux may be computed by a single linear optimisation problem

\begin{equation}
\begin{array}{ll}
\underset{z,w}{\textrm{min}} & \left\Vert v\right\Vert _{1}+c^{T}\cdot w\\
\;\textrm{s.t.} & N\cdot v+B\cdot w=0\;:y\\
 & l_{}\leq v\leq u_{}\;:s\\
 & l_{w}\leq w\leq u_{w}\;:t
\end{array}\label{eq:TFBA}
\end{equation}
where $l$ and $u$ denote lower and upper bounds on internal reaction
fluxes, with the constraint that $l\in\{0,-\infty\}^{n}$ and $u\in\{0,\infty\}^{n}$,
while $l_{w}\in\mathbb{R}^{k}$ and $u_{w}\in\mathbb{R}^{k}$ denote
lower and upper bounds on external reaction fluxes, respectively.
The vectors $y\in\mathbb{R}^{m}$, $s\in\mathbb{R}^{m}$ and $t\in\mathbb{R}^{m},$
are dual variables to the steady state constraint, bounds on internal
reaction rates and bounds on external reaction rates, respectively.

The optimality conditions of Problem \ref{eq:TFBA} are
\begin{eqnarray*}
N\cdot v^{\star}+B\cdot w^{\star} & = & 0\\
\nabla\left\Vert v^{\star}\right\Vert _{1}=\textrm{sign}(v^{\star}) & = & -N^{T}\cdot y^{\star}-s^{\star}\\
c & = & -B^{T}\cdot y^{\star}-t^{\star}
\end{eqnarray*}
where $-y^{\star}$ may interpreted as a vector proportional to the
chemical potentials of each metabolite and $-N_{\textrm{j}}^{T}\cdot y^{\star}$
is proportional to the change of chemical potential for reaction $j$.
When $l\in\{0,-\infty\}^{n}$ and $u\in\{0,\infty\}^{n}$ then the
optimal dual variable $s_{\textrm{j}}^{\star}$ to the inequality
constraints on internal reaction $j$, is non-zero if and only if
$v_{\textrm{j}}^{\star}$ is zero, that is $z_{\textrm{j}}^{\star}=0\iff s_{\textrm{j}}^{\star}\ne0$
and $v_{\textrm{j}}^{\star}\ne0\iff s_{\textrm{j}}^{\star}=0$. Therefore
$v_{\textrm{j}}^{\star}\ne0\iff$$\textrm{sign}(v_{\textrm{j}}^{\star})=-N_{\textrm{j}}^{T}\cdot y^{\star}$,
which enforces energy conservation and the second law of thermodynamics
on the optimal vector of nonzero internal reaction fluxes \cite{fleming_variational_2012}.
However, when $z_{\textrm{j}}=0$, this is a relaxation of the thermodynamic
sign constraint $\textrm{sign}(v_{\textrm{j}})=-\textrm{sign}(N_{\textrm{j}}^{T}\cdot y)$,
since $v_{\textrm{j}}^{\star}=0\iff$$v_{\textrm{j}}^{\star}=N_{\textrm{j}}^{T}\cdot y^{\star}+s_{\textrm{j}}^{\star}=0$,
so $v_{\textrm{j}}^{\star}=0\iff N_{\textrm{j}}^{T}\cdot y^{\star}=-s_{\textrm{j}}^{\star}$
and therefore $v_{\textrm{j}}=0\nLeftrightarrow N_{\textrm{j}}^{T}\cdot y=0$.
Biochemically, one may interpret this relaxation as saying that a
zero internal reaction flux does not imply a zero change in chemical
potential. For example, a nonzero change in chemical potential may
still be consistent with zero net flux when an enzyme is absent for
the corresponding reaction. To summarise, herein we define thermodynamic
consistency as the requirement that any nonzero net flux be driven
by a change in chemical potential for the corresponding reaction,
that is
\begin{eqnarray*}
v_{\textrm{j}}^{\star}>0 & \Rightarrow & N_{\textrm{j}}^{T}\cdot y^{\star}<0,\\
v_{\textrm{j}}^{\star}<0 & \Rightarrow & N_{\textrm{j}}^{T}\cdot y^{\star}>0.
\end{eqnarray*}
However, a reactions must admit a non-zero flux that is thermodynamically
consistent to be deemed thermodynamically flux consistent, so we omit
reactions where zero net flux is the only thermodynamically consistent
solution obtained.

\subsection{Convex optimisation: Introduction}

\subsubsection{Convex functions}

In convex optimisation, the constraints define a polyhedral convex,
set while the objective may be non-linear, but it must be \emph{convex
function}. The \emph{epigraph} of a function is the set of points
lying on or above the function's graph. A function is convex if and
only if its epigraph is a convex set. The function $\phi(x):\mathbb{R}^{n}\rightarrow\mathbb{R}$
given by $x^{T}\cdot x$ is convex. Given $h\in\mathbb{R}^{n}$, the
quadratic function $\phi(x\mid h):\mathbb{R}^{n}\rightarrow\mathbb{R}$,
given by
\[
\phi(x\mid h)\coloneqq\nicefrac{1}{2}(h-x)^{T}\cdot(h-x)
\]
 is convex. The exponential function $\phi(x):\mathbb{R}^{n}\rightarrow\mathbb{R}$,
given by
\[
\phi(x)\coloneqq\mathbf{1}^{T}\cdot\exp(x)
\]
 is convex. The negative entropy function $\phi(x):\mathbb{R}_{\ge0}^{n}\rightarrow\mathbb{R}$,
given by
\[
\phi(x)\coloneqq x^{T}\cdot\ln(x)
\]
 is also convex.

\subsubsection{Standard form convex optimisation}

The standard form of a convex optimisation problem is

\begin{eqnarray}
\underset{x}{\text{min}}\;\phi(x)\nonumber \\
\text{s.t.}\;A\cdot x\le b &  & \textrm{}\label{eq:CO}
\end{eqnarray}
where \emph{$\phi(x)$} is a \emph{convex} \emph{objective function,}
$x\in\mathbb{R}^{n}$ is a vector of variables, $A\in\mathbb{R}^{m\times n}$
is a given linear constraint matrix and $b\in\mathbb{R}^{m}$ is a
given vector of data.

\subsection{Entropic flux balance analysis}

There are several shortcomings with flux balance analysis. An important
one is that the prediction of internal reaction rate may not, and
for genome-scale computational models usually do not, satisfy thermodynamic
constraints. In particular do not satisfy energy conservation and
the second law of thermodynamics. Previously, we developed a novel
method to satisfy the aforementioned constraints using a convex optimisation
problem \cite{fleming_variational_2012}. Since the negative entropy
function is convex, the following is a convex optimisation problem

\begin{equation}
\begin{array}{ll}
\underset{v_{f},v_{r}>0}{\textrm{min}} & v_{f}^{T}\cdot\ln(v_{f})+v_{r}^{T}\cdot\ln(v_{r})\\
\textrm{s.t.} & N\cdot(v_{f}-v_{r})+B\cdot w=0\quad:y
\end{array}\tag{}\label{eq:EP1}
\end{equation}
where we introduce the dual variable $y\in\mathbb{R}^{m}$, which
by convention is written to the right hand side of the primal constraints.
The Lagrangian for this problem is
\[
\mathcal{L}(v_{f},v_{r},w,y)=v_{f}^{T}\cdot\ln(v_{f})+v_{r}^{T}\cdot\ln(v_{r})-y^{T}\cdot(N\cdot(v_{f}-v_{r})+B\cdot w)
\]
and by setting its partial derivatives to equal zero we obtain the
optimality conditions for Problem (\ref{eq:EP1}), which are
\begin{eqnarray}
\frac{\partial\mathcal{L}}{\partial v_{f}} & = & \ln(v_{f}^{\star})+1+N^{T}\cdot y^{\star}=0\label{eq:EPvf}\\
\frac{\partial\mathcal{L}}{\partial v_{r}} & = & \ln(v_{r}^{\star})+1-N^{T}\cdot y^{\star}=0\label{eq:EPvr}\\
\frac{\partial\mathcal{L}}{\partial w} & = & B^{T}\cdot y^{\star}=0\nonumber \\
\frac{\partial\mathcal{L}}{\partial y} & = & N\cdot(v_{f}^{\star}-v_{r}^{\star})+B\cdot w^{\star}=0\nonumber 
\end{eqnarray}
By subtracting (\ref{eq:EPvr}) from (\ref{eq:EPvf}) we obtain
\begin{eqnarray}
\ln(v_{f}^{\star})-\ln(v_{r}^{\star})+2N^{T}\cdot y^{\star} & = & 0\nonumber \\
\ln\left(\frac{v_{f}^{\star}}{v_{r}^{\star}}\right) & = & -2N^{T}\cdot y^{\star}.\label{eq:lnvfvrNty}
\end{eqnarray}
Comparing (\ref{eq:lnvfvrNty}) with (\ref{eq:logvfvrNtu}), we can
see that the optimality conditions of Problem (\ref{eq:EP1}) satisfy
the desired thermodynamic constraints and $2y^{*}$ can be interpreted
as a vector proportional to the chemical potential of each species.

Consider Problem (\ref{eq:EP1}) with the addition of box constraints
and a linear objective on external net rates
\begin{eqnarray}
\underset{v_{f},v_{r},w}{\text{min}}\;v_{f}^{T}\cdot\ln(v_{f})+v_{r}^{T}\cdot\ln(v_{r})+d^{T}\cdot w\nonumber \\
\text{s.t.}\;N\cdot(v_{f}-v_{r})+B\cdot w=0 &  & :-y\label{eq:EntropicFBA}\\
l-(v_{f}-v_{r})\le0 &  & :-z_{l}\\
(v_{f}-v_{r})-u\le0 &  & :-z_{u}
\end{eqnarray}
where $l,u\in\mathbb{R}^{n}$ are given lower and upper bounds on
internal net rate and where $y\in\mathbb{R}^{m}$, $z_{l}\in\mathbb{R}_{\ge0}^{n}$
and $z_{u}\in\mathbb{R}_{\ge0}^{n}$ are dual variables to the steady
state constraints, lower bounds on internal net rate, and upper bounds
on internal net rate, respectively. A dual variable to an equality
constraint may be positive or negative, while a dual variable to an
inequality constraint must be restricted in sign. By convention, they
are non-negative. The Lagrangian for problem (\ref{eq:EntropicFBA})
is
\begin{equation}
\mathcal{L}(v_{f},v_{r},w,y,z_{l},z_{u})\coloneqq v_{f}^{T}\cdot\ln(v_{f})+v_{r}^{T}\cdot\ln(v_{r})-y^{T}\cdot(N\cdot(v_{f}-v_{r})+B\cdot w)-z_{l}^{T}\cdot(l-(v_{f}-v_{r}))-z_{u}^{T}\cdot((v_{f}-v_{r})-u)+d^{T}\cdot w\label{eq:LagrangianQP-FBA1-1}
\end{equation}
and its partial derivatives are

\begin{eqnarray}
\frac{\partial\mathcal{L}}{\partial v_{f}} & = & \ln(v_{f}^{\star})+1-N^{T}\cdot y^{\star}+z_{l}^{\star}-z_{u}^{\star}=0\label{eq:EPvf-1}\\
\frac{\partial\mathcal{L}}{\partial v_{r}} & = & \ln(v_{r}^{\star})+1+N^{T}\cdot y^{\star}-z_{l}^{\star}+z_{u}^{\star}=0\label{eq:EPvr-1}\\
\frac{\partial\mathcal{L}}{\partial w} & = & d-B^{T}\cdot y^{\star}=0\nonumber \\
\frac{\partial\mathcal{L}}{\partial y} & = & N\cdot(v_{f}^{\star}-v_{r}^{\star})+B\cdot w^{\star}=0\nonumber 
\end{eqnarray}
Additionally, it can be shown (5.5 in \cite{boyd_convex_2009}) that
an optimum solution satisfies
\begin{eqnarray}
z_{l}^{\star}\odot(l-(v_{f}^{\star}-v_{r}^{\star})) & = & 0\label{eq:compSlackl}\\
z_{u}^{\star}\odot((v_{f}^{\star}-v_{r}^{\star})-u) & = & 0\label{eq:compSlacku}
\end{eqnarray}
where $\odot$ denotes the component-wise (Hadamard) product of a
pair of vector arguments. These are referred to as \emph{complementary
slackness} conditions, because 
\begin{eqnarray*}
z_{lj}^{\star}>0 & \Leftrightarrow & l_{\textrm{j}}-(v_{fj}^{\star}-v_{rj}^{\star})=0,\\
z_{lj}^{\star}=0 & \Leftrightarrow & l_{\textrm{j}}-(v_{fj}^{\star}-v_{rj}^{\star})<0,
\end{eqnarray*}
and the same for the upper bound constraints, under strict complementarity.
By subtracting (\ref{eq:EPvr-1}) from (\ref{eq:EPvf-1}) we obtain
\begin{eqnarray}
\ln(v_{f}^{\star})-\ln(v_{r}^{\star})-2N^{T}\cdot y^{\star}+2(z_{l}^{\star}-z_{u}^{\star}) & = & 0,\nonumber \\
\ln\left(\frac{v_{f}^{\star}}{v_{r}^{\star}}\right) & = & 2\left(N^{T}\cdot y^{\star}-(z_{l}^{\star}-z_{u}^{\star})\right).\label{eq:lnvfvrNty-1}
\end{eqnarray}
The dual variables to the box constraints in (\ref{eq:lnvfvrNty-1})
could potentially interfere with satisfaction of \ref{eq:logvfvrNtu},
so context-specific models must be generated in a thermodynamically
consistent way \cite{preciat_xomicstomodel_2025}, and bounds on net
rate must be set in a thermodynamically consistent way to avoid this
issue \cite{fleming_cardinality_2023}. Essentially, a solution to
Problem (\ref{eq:EntropicFBA}) will be thermodynamically feasible
provided Problem (\ref{eq:EntropicFBA}) is posed in a way that admits
a thermodynamically feasible solution \cite{fleming_cardinality_2023}.
The existence of a thermodynamically feasible solution is necessary,
but not sufficient, for the existence of a kinetically feasible solution,
as elaborated further below (cf Section (\ref{sec:vkSteady})).

\emph{Entropic flux balance analysis} is a parameterised variant of
Problem (\ref{eq:EP1}) that also enables penalisation of deviation
from measured rates \cite{preciat_mechanistic_2023}. Compared with
a variety of constraint-based modelling approaches, entropic flux
balance analysis has been shown to enable superior predictions of
reaction rates in a context specific model of dopaminergic neuronal
metabolism. . However, entropy maximisation alone it has several shortcomings.
Each primal solution ($v_{f}^{\star},v_{r}^{\star}$) to an entropic
flux balance analysis problem is a unique function of input parameters,
but it is not obvious what the most appropriate parameterisation is.
Maximisation of the relative entropy of unidirectional fluxes, with
a prior derived from transcriptomic data, has been demonstrated to
further increase prediction accuracy \parencite{amestica-toledo_thermodynamically_2026},
which partly addresses the question of the ideal parameters. However,
entropy maximisation tends to bias net rate to reactions whose stoichiometric
coefficients are large in magnitude because, all else being equal,
for a single unit of rate, a reaction with large stoichiometric coefficient
will move more mass than one with a small stoichiometric coefficient.
This issue arises because of reactions in a model with high molecularity,
which are lumped representations of sets of reactions.

With (relative) entropy maximisation, the predicted potentials are
linearly dependent, since the stoichiometric matrix is row rank deficient.
Therefore, only the predicted change in chemical potential ($N^{T}\cdot y^{\star}$)
is unique, not the potential vector $y^{\star}$ itself. The variables
experimentally measured most frequently are species concentrations
while (entropic) flux balance analysis predicts rates and change in
chemical potential. This makes comparison of measurements and predictions
difficult. When concentrations of molecular species are not represented
as variables, it makes it difficult to incorporate data on metabolite
concentrations. While predicted rates are thermodynamically feasible,
they may not satisfy known reaction rate laws and as such they are
not guaranteed to be kinetically feasible. Because known reaction
rate laws are not represented, an important established feature of
(bio)chemistry is not represented, potentially resulting in prediction
artefacts. This motivates a search for novel modelling methods that
also incorporate kinetic constraints, yet retain the theoretical and
numerical advantages of convex optimisation methods.

\section{\protect\label{sec:Variational-elementary-kinetics}Variational elementary
kinetics}

\subsection{Conic optimisation: Introduction}

In \emph{conic optimisation}, the objective is linear, but the constraints
are an intersection of a polyhedral convex set and one or more\emph{
convex cones}. A cone $\mathcal{K}$ is proper when it is (a) closed
(contains its boundary or more technically limit points), (b) pointed,
$\mathcal{K}\cap(-\mathcal{K})=\{0\}$ and (c) it has nonempty interior.
A \emph{proper} \emph{convex cone} defines a convex set. Henceforth,
for brevity, cone means proper convex cone. A \emph{conic inequality
}is a constraint 
\[
x\in\mathcal{K}
\]
where $\mathcal{K}$ is a proper convex cone. Each conic inequality
satisfies certain properties, e.g., a conic inequality is preserved
by non-negative linear combinations, that is
\[
x\in\mathcal{K},y\in\mathcal{K}\Rightarrow\alpha x+\beta y\in\mathcal{K}\;\forall\alpha,\beta\ge0.
\]
A simple example of a cone is the set defined by the non-negative
orthant, 
\[
\mathcal{K}_{\ge0}\coloneqq\left\{ x\in\mathbb{R}^{n}\mid x\ge0\right\} \Leftrightarrow x\in\mathcal{K}_{\ge0}.
\]
The boundary of a cone is denoted $\partial\mathcal{K}$. The boundary
of the non-negative orthant cone is where there exists at least one
coordinate that is zero, that is
\[
\partial\mathcal{K}_{\ge0}\coloneqq\left\{ x\in\mathbb{R}^{n}\mid x\ge0,\exists x_{\textrm{j}}=0\right\} 
\]
The interior of a cone is denoted $\textrm{int}\mathcal{K}$. The
interior of the non-negative orthant cone is the set of points where
all coordinates are strictly positive, that is

\[
\textrm{int}\mathcal{K}_{\ge0}\coloneqq\left\{ x\in\mathbb{R}^{n}\mid x>0\right\} .
\]
Conic optimisation is focused on optimisation over certain types of
cones that admit the expression of a well behaved barrier function,
that is a smooth function with mathematical properties that enable
an interior point algorithms to enforce feasibility with respect to
a conic equality, while optimising within the cone. For example, the
\emph{exponential cone} is a $3$ dimensional cone defined\footnote{Warning: some articles in the literature define the same exponential
cone but with a different convention for the order of the variables
$x_{1},x_{2},x_{3}$. Throughout, we stick to the convention adopted
by the MOSEK conic optimisation solver an its associated documentation,
e.g., \href{https://docs.mosek.com/modeling-cookbook/expo.html}{https://docs.mosek.com/modeling-cookbook/expo.html}} as the closure of the set of points that satisfy
\begin{equation}
\mathcal{K}_{exp}\coloneqq\left\{ (x_{1},x_{2},x_{3})\mid x_{1}\ge x_{2}\exp\left(\frac{x_{3}}{x_{2}}\right),x_{1},x_{2}>0\right\} \Leftrightarrow\left(\begin{array}{c}
x_{1}\\
x_{2}\\
x_{3}
\end{array}\right)\in\mathcal{K}_{exp}\subset\mathbb{R}^{3}.\label{eq:Kexp}
\end{equation}
Observe that the epigraph of an exponential function is a two dimensional
slice of an exponential cone where $x_{2}=1$, that is
\[
\left(\begin{array}{c}
x_{1}\\
1\\
x_{3}
\end{array}\right)\in\mathcal{K}_{exp}\Leftrightarrow x_{1}\ge\exp\left(x_{3}\right),x_{1}>0.
\]
Equivalently, the exponential cone may be defined in logarithmic rather
than exponential terms as the closure of the set of points that satisfy
\begin{equation}
\mathcal{K}_{exp}\coloneqq\left\{ (x_{1},x_{2},x_{3})\mid-x_{3}\ge x_{2}\ln\left(\frac{x_{2}}{x_{1}}\right),x_{1},x_{2}>0\right\} \Leftrightarrow\left(\begin{array}{c}
x_{1}\\
x_{2}\\
x_{3}
\end{array}\right)\in\mathcal{K}_{exp}\subset\mathbb{R}^{3}.\label{eq:KexpLogTerms}
\end{equation}
The \emph{rotated quadratic cone} is a $2+n$ dimensional cone
\[
\mathcal{Q}^{2+n}\coloneqq\left\{ (x,y,z)\mid2xy\ge z_{1}^{2}+z_{2}^{2}+\ldots+z_{n}^{2},x\ge0,y\ge0\right\} \Leftrightarrow\left(\begin{array}{c}
x\\
y\\
z_{1}\\
\vdots\\
z_{n}
\end{array}\right)\in\mathcal{Q}^{2+n},
\]
which is a $2+n$ dimensional cone. Any positive semidefinite matrix
may be factorised as $H^{T}\cdot H$, where $H\in\mathbb{R}^{k\times n}$
and $k=\textrm{rank}(H)$. Given such a factor $H$ and a vector $h\in\mathbb{R}^{n}$
a convex quadratic set can be represented by a rotated quadratic cone
of the form
\begin{equation}
t\ge\nicefrac{1}{2}\left(h-x\right)^{T}\cdot H^{T}\cdot H\cdot\left(h-x\right)\Leftrightarrow\left(\begin{array}{c}
t\\
1\\
H\cdot(x-h)
\end{array}\right)\in\mathcal{Q}^{2+n},\label{eq:affineQuadratic}
\end{equation}
which is an affine quadratic constraint.

Any convex constraint can be represented as a conic inequality, with
minor modifications to make $\mathcal{K}$ proper. A cone can be constructed
from any convex function $\phi(x):\mathbb{R}^{n}\rightarrow\mathbb{R}$
because the epigraph of the perspective of a convex function is a
convex cone. It follows, that any constraint involving a convex function
$\phi(x_{1}):\mathbb{R}^{m}\rightarrow\mathbb{R}$ can be represented
as a conic linear inequality, that is
\[
x_{3}\ge\phi(x_{1})\Leftrightarrow\left(\begin{array}{c}
x_{1}\\
1\\
x_{3}
\end{array}\right)\in\mathcal{K}_{\phi}.
\]
Let $A\in\mathbb{R}^{m\times n}$, $b\in\mathbb{R}^{m}$, $F\in\mathbb{R}^{k\times n}$
and $g\in\mathbb{R}^{k}$, then the standard form for a \emph{conic
optimisation} problem is

\begin{eqnarray}
\underset{x}{\text{min}} & c^{T}\cdot x\nonumber \\
\text{s.t.} & A\cdot x\le b, & \textrm{}\label{eq:conicP}\\
 & F\cdot x+g\in\mathcal{K}.
\end{eqnarray}
In practice, current conic optimisation solvers support a limited
number of types of cones, therefore whether a problem can be solved
by a given solver depends on which types of cone the solver supports.

\subsection{\protect\label{subsec:Variational-kinetics:-ele}Conification of
elementary reaction kinetics}

Our novel approach is to reformulate Eq. (\ref{eq:massBalanceKinetics})
into a system of equations that represent a convex set, which is amenable
to optimisation, yet satisfy Eq. (\ref{eq:massBalanceKinetics}) at
the solution to an optimisation problem. Eq. (\ref{eq:massBalanceKinetics})
may be rewritten as
\begin{align}
N\cdot(v_{f}-v_{r})+B\cdot w=0,\label{eq:NvfvfBw}\\
v_{f}=\exp(\ln(k_{f})+F^{T}\cdot\ln(c)),\label{eq:vf=00003D}\\
v_{r}=\exp(\ln(k_{r})+R^{T}\cdot\ln(c)),\label{eq:vr=00003D}
\end{align}
which still represents the same non-convex set as Eq. (\ref{eq:massBalanceKinetics}).
By introducing logarithmic variables in place of kinetic parameters
and concentration, 
\begin{eqnarray}
lnk_{f} & \coloneqq & \ln\left(k_{f}\right),\label{eq:lnkf}\\
lnk_{r} & \coloneqq & \ln\left(k_{r}\right),\label{eq:lnkr}\\
lnc & \coloneqq & \ln\left(c\right),\label{eq:lnc}
\end{eqnarray}
then Eq. (\ref{eq:NvfvfBw}), (\ref{eq:vf=00003D}) and (\ref{eq:vr=00003D})
become

\begin{align}
N\cdot(v_{f}-v_{r})+B\cdot w=0,\label{eq:NvBw0}\\
v_{f}=\exp(lnk_{f}+F^{T}\cdot lnc),\label{eq:vf}\\
v_{r}=\exp(lnk_{r}+R^{T}\cdot lnc),\label{eq:vr}\\
k_{f}=\exp(lnk_{f}),\label{eq:kf}\\
k_{r}=\exp(lnk_{r}),\label{eq:kr}\\
c=\exp(lnc),\label{eq:c}
\end{align}
which, again, still represents the same non-convex set as Eq. (\ref{eq:massBalanceKinetics}).
Assume we have a solution $\left(v_{f},v_{r},lnk_{f},lnk_{r},lnc\right)$
to Eq. (\ref{eq:NvBw0})-(\ref{eq:vr}), then it is straightforward
to compute the exponentials of $lnk_{f},lnk_{r},lnc$ to obtain $k_{f},k_{r},c$
so for the sake of clarity, we will momentarily omit (\ref{eq:kf}),
(\ref{eq:kr}) and (\ref{eq:c}) from consideration.

Eq. (\ref{eq:vf=00003D}) and (\ref{eq:vr=00003D}) each involves
an exponential term, which can be relaxed by replacing each equality
by an inequality to give

\begin{align}
N\cdot(v_{f}-v_{r})+B\cdot w=0,\label{eq:NvBw}\\
v_{f}\ge\exp(lnk_{f}+F^{T}\cdot lnc),\label{eq:vf_ge}\\
v_{r}\ge\exp(lnk_{r}+R^{T}\cdot lnc),\label{eq:vr_ge}
\end{align}
It is clear that Eq. (\ref{eq:NvBw}) represents a (polyhedral) convex
set. However, each of the inequalities (\ref{eq:vf_ge}) and (\ref{eq:vr_ge})
also represents a convex set, even though it is nonlinear. This can
be appreciated by recognising that the exponential is a convex function
which is the same as saying that the epigraph of an exponential function
is a convex set. In fact, each term of the form $x_{1}\ge\exp(x_{3})$
is a two dimensional plane within an exponential cone (\ref{eq:Kexp})
with $x_{2}=1.$ That is, one may express each of the constraints
involving an exponential term as either of the two equivalent forms
\[
x_{1}\ge\exp\left(x_{3}\right)\iff\left(\begin{array}{c}
x_{1}\\
1\\
x_{3}
\end{array}\right)\in\mathcal{K}_{exp},
\]
where it is implicit in this representation that $x_{2}=1$. Therefore,
Eq. (\ref{eq:NvBw}), (\ref{eq:vf_ge}) and (\ref{eq:vr_ge}) may
be expressed as

\begin{align}
N\cdot(v_{f}-v_{r})+B\cdot w=0,\label{eq:NvBw-1}\\
\left(\begin{array}{c}
v_{f}\\
1\\
F^{T}\cdot lnc+lnk_{f}
\end{array}\right)\in\mathcal{K}_{exp}^{n},\label{eq:vf_ge-1}\\
\left(\begin{array}{c}
v_{r}\\
1\\
R^{T}\cdot lnc+lnk_{r}
\end{array}\right)\in\mathcal{K}_{exp}^{n},\label{eq:vr_ge-1}
\end{align}
where there are a set of $n$ exponential cone constraints, one for
each forward reaction, and a second set $n$ exponential cone constraints,
one for each reverse reaction. The order of the variables is to be
interpreted as set of $n$ component-wise exponential cone constraints,
where $1$ is a $n\times1$ dimensional vector of constants. That
is 
\[
\left(\begin{array}{c}
v_{f}\\
1\\
F^{T}\cdot lnc+lnk_{f}
\end{array}\right)\in\mathcal{K}_{exp}^{n}\iff\left(\begin{array}{c}
v_{fj}\\
1\\
F_{\textrm{:,j}}^{T}\cdot lnc+lnk_{f,\textrm{j}}
\end{array}\right)\in\mathcal{K}_{exp}\;\forall\;\textrm{j}\;\in1\ldots n.
\]
Without loss of generality, assume that we are given logarithmic kinetic
parameters $lnk_{f}$ and $lnk_{r}$. Constraints (\ref{eq:NvBw-1})-(\ref{eq:vr_ge-1})
define a nonlinear, yet convex set, at the boundary of which are solutions
to Eq. (\ref{eq:NvBw0})-(\ref{eq:vr}), which may be obtained at
the optimum of the following \emph{conic optimisation} problem

\begin{eqnarray}
\underset{v_{f},v_{r}w,lnc}{\text{min}}\qquad c_{v_{f}}^{T}\cdot v_{f}+c_{v_{r}}^{T}\cdot v_{r}+c_{lnc}^{T}\cdot lnc\nonumber \\
\text{s.t.}\qquad N\cdot(v_{f}-v_{r})+B\cdot w=0 &  & \textrm{}\label{eq:vk1}\\
\left(\begin{array}{c}
v_{f}\\
1\\
F^{T}\cdot lnc+lnk_{f}
\end{array}\right)\in\mathcal{K}_{exp}^{n} & \iff & v_{f}\ge\exp\left(F^{T}\cdot lnc+lnk_{f}\right)\\
\left(\begin{array}{c}
v_{r}\\
1\\
R^{T}\cdot lnc+lnk_{r}
\end{array}\right)\in\mathcal{K}_{exp}^{n} & \iff & v_{r}\ge\exp\left(R^{T}\cdot lnc+lnk_{r}\right).
\end{eqnarray}
This begs the question, how does one choose the parameters $c_{v_{f}}$,
$c_{v_{r}}$ and $c_{lnc}$ such that the optimal solution to this
problem lies at the boundary of each exponential cone, where $v_{f}^{\star}=\exp\left(F^{T}\cdot lnc^{\star}+lnk_{f}\right)$
and $v_{r}^{\star}=\exp\left(R^{T}\cdot lnc^{\star}+lnk_{r}\right)$?
An answer to this questions requires some additional conic optimisation
theory, which is framed in terms of generic conic optimisation problem.

\subsection{\protect\label{subsec:Correspondence-with-reaction-kinetics-minor}Correspondence
with reaction kinetics}

The correspondence between the reaction kinetic problem \ref{eq:vk1}
and the generic conic optimisation problem \ref{eq:conicP} can be
understood from the following

\[
x\coloneqq\left[\begin{array}{c}
v_{f}\\
v_{r}\\
1\\
1\\
lnc\\
w
\end{array}\right],\qquad c\coloneqq\left[\begin{array}{c}
c_{v_{f}}\\
c_{v_{r}}\\
0\\
0\\
c_{lnc}\\
0
\end{array}\right],\qquad A\coloneqq\left[\begin{array}{ccccc}
N & -N & 0 & 0 & 0\end{array}\right],\qquad b\coloneqq B\cdot w,
\]

and

\[
F\coloneqq\left[\begin{array}{cccccc}
I & 0 & 0 & 0 & 0 & 0\\
0 & 0 & 0 & I & 0 & 0\\
0 & 0 & F^{T} & 0 & 0 & 0\\
0 & I & 0 & 0 & 0 & 0\\
0 & 0 & 0 & 0 & I & 0\\
0 & 0 & R^{T} & 0 & 0 & 0
\end{array}\right],g\coloneqq\left[\begin{array}{c}
0\\
0\\
lnk_{f}\\
0\\
0\\
lnk_{r}
\end{array}\right],
\]

where 
\begin{eqnarray*}
F_{1} & \coloneqq & \left[\begin{array}{cccccc}
I & 0 & 0 & 0 & 0 & 0\\
0 & I & 0 & 0 & 0 & 0
\end{array}\right],\qquad g_{1}\coloneqq\left[\begin{array}{c}
0\\
0
\end{array}\right],\\
F_{2} & \coloneqq & \left[\begin{array}{cccccc}
0 & 0 & 0 & I & 0 & 0\\
0 & 0 & 0 & 0 & I & 0
\end{array}\right],\qquad g_{1}\coloneqq\left[\begin{array}{c}
0\\
0
\end{array}\right],\\
F_{3} & \coloneqq & \left[\begin{array}{cccccc}
0 & 0 & F^{T} & 0 & 0 & 0\\
0 & 0 & R^{T} & 0 & 0 & 0
\end{array}\right],\qquad g_{1}\coloneqq\left[\begin{array}{c}
lnk_{f}\\
lnk_{f}
\end{array}\right],
\end{eqnarray*}
and therefore the exponential cone constraints are
\[
F_{1}\cdot x+g_{1}\ge\exp\left(F_{3}\cdot x+g_{3}\right),\quad F_{2}\cdot x+g_{2}=1,
\]
with equality of each row when the corresponding exponential cone
constraint is active.

\subsection{Conic optimisation: optimality conditions}

\subsubsection{Dual cone}

We define a \emph{dual cone }as 
\[
\mathcal{K}^{\star}\coloneqq\left\{ s\in\mathbb{R}^{n}\mid x^{T}\cdot s\ge0\;\forall x\in\mathcal{K}\right\} .
\]
Since $x\in\mathcal{K}$, $s\in\mathcal{K}^{\star}$ and $x^{T}\cdot s\ge0$
then $x$ and $s$ are referred to as \emph{primal} and \emph{dual}
variables, as they reside within primal and dual cones. Where $\mathcal{K}$
is a proper convex cone, the dual cone $\mathcal{K}^{\star}$ has
the following properties, (a) it is a proper cone, (b) the dual of
the dual cone is the primal cone, that is $\left(\mathcal{K}^{\star}\right)^{\star}=\mathcal{K}$,
and (c) the interior of the dual cone $\textrm{int}\mathcal{K}^{\star}$
is given by
\[
\textrm{int}\mathcal{K}^{\star}=\left\{ s\in\mathbb{R}^{n}\mid s\ne0,\:x\ne0,\:x^{T}\cdot s>0\;\forall x\in\mathcal{K}\right\} .
\]
Let $A\in\mathbb{R}^{m\times n}$, then dual of a linear subspace
is the orthogonal subspace. For example let the nullspace be $\mathcal{L}\coloneqq\left\{ x\in\mathbb{R}^{n}\mid Ax=0\right\} $
then the dual to the nullspace is $\mathcal{L}^{\perp}=\mathcal{L}^{\star}$,
which is the row space $\mathcal{L}^{\perp}=\mathcal{L}^{\star}=\left\{ s\in\mathbb{R}^{n}\mid s=A^{T}\cdot z,z\in\mathbb{R}^{m}\right\} $.
The dual of the exponential cone is the closure 
\begin{equation}
\mathcal{K}_{exp}^{\star}\coloneqq\textrm{cl}\left\{ (s_{1},s_{2},s_{3})\mid s_{1}\ge-s_{3}\exp\left(\frac{s_{2}}{s_{3}}-1\right),s_{1}>0,s_{3}<0\right\} \Leftrightarrow\left(\begin{array}{c}
s_{1}\\
s_{2}\\
s_{3}
\end{array}\right)\in\mathcal{K}_{exp}^{\star}.\label{eq:dualExpCone}
\end{equation}
The non-negative orthant cone 
\[
\mathcal{K}\coloneqq\mathbb{R}_{\ge0}^{n}
\]
is \emph{self-dual} since the dual of the non-negative orthant is
the non-negative orthant, that is

\[
\mathbb{R}_{\ge0}^{n}=\left\{ s\mid x^{T}\cdot s\ge0,\forall x\in R_{\ge0}^{n}\right\} =\mathcal{K}^{\star}.
\]

\subsubsection{Optimality conditions}

Let $A\in\mathbb{R}^{m\times n}$, $b\in\mathbb{R}^{m}$, $F\in\mathbb{R}^{3k\times n}$
and $g\in\mathbb{R}^{3k}$, then consider the primal \emph{exponential
conic linear optimisation} problem

\begin{eqnarray}
\underset{x}{\text{min}} & c^{T}\cdot x\nonumber \\
\text{s.t.} & A\cdot x=b & :y\label{eq:primalConicLPv}\\
 & F\cdot x+g\in\mathcal{K}_{exp}^{k} & :s
\end{eqnarray}
where $y\in\mathbb{R}_{}^{m}$ is a vector of dual variables to the
linear equality constraints and $s\in\mathbb{R}^{3k}$ is a vector
of dual variables to each conically constrained term. The \emph{Lagrangian}
analogue of problem (\ref{eq:primalConicLPv}) is
\begin{eqnarray}
\mathcal{L}(x,y,s) & \coloneqq & c^{T}\cdot x-y^{T}\cdot(A\cdot x-b)-s^{T}\cdot(F\cdot x+g)\label{eq:LagrangianConicLP}\\
 &  & F\cdot x+g\in\mathcal{K}\nonumber \\
 &  & s\in\mathcal{K}^{\star}
\end{eqnarray}
The gradient of the Lagrangian with respect to $x$ and $y$ are zero
at the optimum of problem (\ref{eq:primalConicLPv}), that is
\begin{eqnarray*}
\frac{\partial\mathcal{L}(x,y,s)}{\partial x} & = & c-A^{T}\cdot y^{\star}-F^{T}\cdot s^{\star}=0\\
\frac{\partial\mathcal{L}(x,y,s)}{\partial y} & = & b-Ax^{\star}=0.
\end{eqnarray*}
Since $F\cdot x+g\in\mathcal{K}$ and $s\in\mathcal{K}^{\star}$ we
have $s^{T}\cdot(F\cdot x+g)\ge0$ because they are primal and dual
variables. At the optimum of \ref{eq:LagrangianConicLP} we have the
complementarity condition
\[
s^{\star T}\cdot(F\cdot x^{\star}+g)=0,
\]
because if $s^{T}\cdot(F\cdot x+g)>0$ then it would be possible to
further minimise the Lagrangian, which contradicts minimality, therefore
$s^{T}\cdot(F\cdot x+g)=0$. This is a complementary condition, rather
than complementary slackness as in (\ref{eq:compSlackl}) or (\ref{eq:compSlacku})
as $s^{T}\cdot(F\cdot x+g)=0$ does not imply that $s\odot(F\cdot x+g)=0.$

Actually, there is one complementarity condition involving 3 primal
and 3 dual terms per primal exponential conic constraint. Specifically,
each cone corresponds to one complementarity constraint of the form
\[
s_{1}(F\cdot x+g)_{1}+s_{2}(F\cdot x+g)_{2}+s_{2}(F\cdot x+g)_{2}=0.
\]
We now express these conditions for a set of $k$ exponential cones
in matrix-vector form. Let $B\in\{0,1\}^{k\times3k}$ be a matrix
where $k$ is the number of exponential cones in the cone product
$\mathcal{K}_{exp}^{k}$, and $3k$ is the number of rows of $F$,
each corresponding to a conically constrained term. Let $B_{\textrm{i,j}}=1$
if cone $i$ involves conically constrained term $j$ and zero otherwise.
Then the set of complementarity conditions may be expressed as $B\cdot(s^{\star}\odot(F\cdot x^{\star}+g))=0$.
Combining the aforementioned constraints then $x$, $y$ and $s$
are optimal if and only if 
\begin{eqnarray}
A\cdot x^{\star}-b & = & 0,\quad:y^{\star}\label{eq:COoptimality1}\\
c-A^{T}\cdot y^{\star}-F^{T}\cdot s^{\star} & = & 0,\quad:x^{\star}\\
F\cdot x^{\star}+g & \in & \mathcal{K},\quad:s^{\star}\\
s^{\star} & \in & \mathcal{K}^{\star},\\
B\cdot(s^{\star}\odot(F\cdot x^{\star}+g)) & = & 0,\label{eq:COoptimality5}
\end{eqnarray}
which represent the optimality conditions of Problem \ref{eq:primalConicLPv},
and the corresponding dual variables.

\subsubsection{Primal and dual problems in conic optimisation}

Let $F\in\mathbb{R}^{k\times n}$ and $g\in\mathbb{R}^{k}$. Note
that we follow the notation in the conic optimisation community by
using $F$ to denote the matrix in the affine conic constraints, which
clashes with the use of the same upper case roman letter for forward
stoichiometric matrix. Let $c^{\star}\coloneqq c^{T}\cdot x^{\star}$
denote the optimal value of the objective of the following primal
\emph{conic linear optimisation} problem

\begin{eqnarray}
\underset{x}{\text{min}} & c^{T}\cdot x\nonumber \\
\text{s.t.} & F\cdot x+g\in\mathcal{K} & :\textrm{s}\label{eq:primalConicLP}
\end{eqnarray}
Let $d^{\star}\coloneqq g^{T}\cdot y^{\star}$ be the optimal value
of the following \emph{dual conic linear optimisation} problem

\begin{eqnarray}
\underset{y}{\text{max}} & g^{T}\cdot y\nonumber \\
\text{s.t.} & c-F^{T}\cdot y=0 & \textrm{:}x\label{eq:dualConicLP}\\
 & y\in\mathcal{K}^{\star}
\end{eqnarray}
then, without exception, \emph{weak duality} holds
\[
c^{\star}=c^{T}\cdot x^{\star}\ge d^{\star}=-g^{T}\cdot y^{\star}.
\]
If the primal or dual problem is strictly feasible, then \emph{strong
duality} holds, that is $c^{\star}=d^{\star}$. If the primal is \emph{strictly}
feasible, then the dual optimum is attained if, $d^{\star}$ is finite.
If $d^{\star}$ is not finite, a solver will report that the dual
problem is unbounded. If the dual is \emph{strictly} feasible, then
the primal optimum is attained, if $c^{\star}$ is finite. If $c^{\star}$
is not finite, a solver will report that the primal problem is unbounded.
This is more restrictive than linear programming duality, where the
strong duality holds if the primal or dual problem is feasible (not
strictly feasible).

\subsection{Variational elementary kinetics: optimality conditions}

The primal optimisation problem in Problem (\ref{eq:vk1}), with the
addition of dual variables is

\begin{eqnarray}
\underset{v_{f},v_{r},w,lnc}{\text{min}}\qquad c_{v_{f}}^{T}\cdot v_{f}+c_{v_{r}}^{T}\cdot v_{r}+c_{lnc}^{T}\cdot lnc\nonumber \\
\text{s.t.}\qquad N\cdot(v_{f}-v_{r})+B\cdot w=0 &  & :-y\label{eq:vk1_s}\\
\left(\begin{array}{c}
v_{f}\\
1\\
F^{T}\cdot lnc+lnk_{f}
\end{array}\right)\in\mathcal{K}_{exp}^{n} &  & :-\left(\begin{array}{c}
s_{vf}\\
s_{f1}\\
s_{F}
\end{array}\right)\\
\left(\begin{array}{c}
v_{r}\\
1\\
R^{T}\cdot lnc+lnk_{r}
\end{array}\right)\in\mathcal{K}_{exp}^{n} &  & :-\left(\begin{array}{c}
s_{vr}\\
s_{r1}\\
s_{R}
\end{array}\right)
\end{eqnarray}
Note that $s$ is used to denote any dual variable to a cone constraint,
its subscript is reflective of the corresponding primal term and $s_{vf},s_{vr}s_{f1},s_{r1},s_{F},s_{R}\in\mathbb{R}^{n}$.
The \emph{Lagrangian} corresponding to Problem (\ref{eq:vk1_s}) is

\begin{eqnarray}
\mathcal{L}(v_{f},v_{r},w,lnc,y,s) & \coloneqq & c_{v_{f}}^{T}\cdot v_{f}+c_{v_{r}}^{T}\cdot v_{r}+c_{lnc}^{T}\cdot lnc-y^{T}\cdot(N\cdot(v_{f}-v_{r})+B\cdot w)\label{eq:LagrangianVK1}\\
 &  & -\left(\begin{array}{c}
s_{vf}\\
s_{f1}\\
s_{F}
\end{array}\right)^{T}\cdot\left(\begin{array}{c}
v_{f}\\
1\\
F^{T}\cdot lnc+lnk_{f}
\end{array}\right)-\left(\begin{array}{c}
s_{vr}\\
s_{r1}\\
s_{R}
\end{array}\right)^{T}\cdot\left(\begin{array}{c}
v_{r}\\
1\\
R^{T}\cdot lnc+lnk_{r}
\end{array}\right)\\
 &  & \left(\begin{array}{c}
v_{f}\\
1\\
F^{T}\cdot lnc+lnk_{f}
\end{array}\right)\in\mathcal{K}_{exp}^{n},\left(\begin{array}{c}
v_{r}\\
1\\
R^{T}\cdot lnc+lnk_{r}
\end{array}\right)\in\mathcal{K}_{exp}^{n}\\
 &  & \left(\begin{array}{c}
s_{vf}\\
s_{f1}\\
s_{F}
\end{array}\right)\in\mathcal{K}_{exp}^{n\star},\left(\begin{array}{c}
s_{vr}\\
s_{r1}\\
s_{R}
\end{array}\right)\in\mathcal{K}_{exp}^{n\star}\nonumber 
\end{eqnarray}
where the last two pairs of terms represent the requirement for primal
and dual terms to be constrained to lie within primal and dual exponential
cones. The optimality conditions may be obtained by (a) setting the
partial derivatives of the Lagrangian with respect to the variables
$v_{f},v_{r},w,lnc,y$ to zero, that is

(\ref{eq:vk1_s})
\begin{eqnarray}
\frac{\partial\mathcal{L}}{\partial v_{f}} & = & c_{v_{f}}-N^{T}\cdot y^{\star}-s_{vf}^{\star}=0\label{eq:dL_vf}\\
\frac{\partial\mathcal{L}}{\partial v_{r}} & = & c_{v_{r}}+N^{T}\cdot y^{\star}-s_{vr}^{\star}=0\label{eq:dL_vr}\\
\frac{\partial\mathcal{L}}{\partial w} & = & -B^{T}\cdot y^{\star}=0\nonumber \\
\frac{\partial\mathcal{L}}{\partial lnc} & = & c_{lnc}-F\cdot s_{F}^{\star}-R\cdot s_{R}^{\star}=0\label{eq:dL_p}\\
\frac{\partial\mathcal{L}}{\partial y} & = & N\cdot(v_{f}^{\star}-v_{r}^{\star})+B\cdot w^{\star}=0\nonumber 
\end{eqnarray}
(b) expressing the complementarity conditions between primal and dual
variables, and

\begin{eqnarray}
\left(\begin{array}{c}
s_{vf}^{\star}\\
s_{f1}^{\star}\\
s_{F}^{\star}
\end{array}\right)^{T}\cdot\left(\begin{array}{c}
v_{f}^{\star}\\
1\\
F^{T}\cdot lnc^{\star}+lnk_{f}
\end{array}\right) & = & 0\label{eq:complemenatarityF}\\
\left(\begin{array}{c}
s_{vr}^{\star}\\
s_{r1}^{\star}\\
s_{R}^{\star}
\end{array}\right)^{T}\cdot\left(\begin{array}{c}
v_{r}^{\star}\\
1\\
R^{T}\cdot lnc^{\star}+lnk_{r}
\end{array}\right) & = & 0\nonumber 
\end{eqnarray}
(c) specifying that the primal and dual terms are constrained to reside
within primal and dual conic cones, respectively, with

\begin{eqnarray*}
\left(\begin{array}{c}
v_{f}^{\star}\\
1\\
F^{T}\cdot lnc^{\star}+lnk_{f}
\end{array}\right)\in\mathcal{K}_{exp}^{n} &  & \left(\begin{array}{c}
v_{r}^{\star}\\
1\\
R^{T}\cdot lnc^{\star}+lnk_{r}
\end{array}\right)\in\mathcal{K}_{exp}^{n}\\
\left(\begin{array}{c}
s_{vf}^{\star}\\
s_{f1}^{\star}\\
s_{F}^{\star}
\end{array}\right)\in\mathcal{K}_{exp}^{^{\star}n} &  & \left(\begin{array}{c}
s_{vr}^{\star}\\
s_{r1}^{\star}\\
s_{R}^{\star}
\end{array}\right)\in\mathcal{K}_{exp}^{^{\star}n}
\end{eqnarray*}
where by the definition of the primal exponential cone we have $v_{f}^{\star},v_{r}^{\star}>0$
and from the definition of the closure of the dual exponential cone
we have $s_{vf}^{\star},s_{vr}^{\star}>0$ and $s_{F}^{\star},s_{R}^{\star}<0$.

Elementary reaction kinetics requires satisfaction of the constraints
\begin{eqnarray}
\left[\begin{array}{c}
v_{f}\\
v_{r}
\end{array}\right] & = & \exp\left(\left[\begin{array}{cc}
F & ,R\end{array}\right]^{T}\cdot lnc+\left[\begin{array}{c}
lnk_{f}\\
lnk_{r}
\end{array}\right]\right),\label{eq:elementaryKinetics}
\end{eqnarray}
while the optimality conditions of Problem (\ref{eq:vk1_s}) has relaxed
these constraints to
\begin{eqnarray*}
\left(\begin{array}{c}
v_{f}^{\star}\\
1\\
F^{T}\cdot lnc^{\star}+lnk_{f}
\end{array}\right)\in\mathcal{K}_{exp}^{n} & \Leftrightarrow & v_{f}^{\star}\ge\exp\left(F^{T}\cdot lnc^{\star}+lnk_{f}\right),\\
\left(\begin{array}{c}
v_{r}^{\star}\\
1\\
R^{T}\cdot lnc^{\star}+lnk_{r}
\end{array}\right)\in\mathcal{K}_{exp}^{n} & \Leftrightarrow & v_{r}^{\star}\ge\exp\left(R^{T}\cdot lnc^{\star}+lnk_{r}\right).
\end{eqnarray*}
so when any one of these inequalities is strict, the corresponding
unidirectional flux is in the interior of an exponential cone. Let
$(v_{f}^{\bullet},v_{r}^{\bullet},lnc^{\bullet})$ denote an optimal
solution of Problem (\ref{eq:vk1_s}) that is also at the boundary
of the exponential cone, that is

\begin{eqnarray}
\left(\begin{array}{c}
v_{f}^{\bullet}\\
1\\
F^{T}\cdot lnc^{\bullet}+lnk_{f}
\end{array}\right)\in\partial\mathcal{K}_{exp}^{n} & \Leftrightarrow & v_{f}^{\bullet}=\exp\left(F^{T}\cdot lnc^{\bullet}+lnk_{f}\right),\label{eq:expConeBoundaryvf}\\
\left(\begin{array}{c}
v_{r}^{\bullet}\\
1\\
R^{T}\cdot lnc^{\bullet}+lnk_{r}
\end{array}\right)\in\partial\mathcal{K}_{exp}^{n} & \Leftrightarrow & v_{r}^{\bullet}=\exp\left(R^{T}\cdot lnc^{\bullet}+lnk_{r}\right),\label{eq:expConeBoundaryvr}
\end{eqnarray}
which are identical to constraints in Eq. (\ref{eq:elementaryKinetics})
required for elementary reaction kinetics to hold. Problem (\ref{eq:vk1_s})
has two free parameter vectors $c_{v_{f}}\in\mathbb{R}^{n}$ and $c_{v_{r}}\in\mathbb{R}^{n}$.

The optimality conditions (\ref{eq:dL_vf}) and (\ref{eq:dL_vr})
relate the free parameters $c_{v_{f}}\in\mathbb{R}^{n}$ and $c_{v_{r}}\in\mathbb{R}^{n}$
to the dual variables via
\begin{eqnarray*}
-c_{v_{r}}+s_{v_{r}}^{\star}= & N^{T}\cdot y^{\star} & =c_{v_{f}}-s_{v_{f}}^{\star}
\end{eqnarray*}
and since $s_{v_{f}}^{\star},s_{v_{r}}^{\star}\mathbb{R}_{\ge0}^{n}$,
we have 
\begin{eqnarray*}
-c_{v_{r}}\le & N^{T}\cdot y^{\star} & \le c_{v_{f}}
\end{eqnarray*}
which demonstrates that the values of $-c_{v_{r}}$ and $c_{v_{f}}$
are lower and upper bounds on $N^{T}\cdot y^{\star}$, therefore $-c_{v_{r}}\le c_{v_{f}}$,
therefore $0\le c_{v_{f}}+c_{v_{r}}$. A sufficient condition for
the latter is $c_{v_{f}},c_{v_{r}}\in\mathbb{R}_{\ge0}^{n}$. The
optimality condition (\ref{eq:dL_p}) relates the free parameters
$c_{lnc}$ to the dual variables via
\begin{eqnarray}
c_{lnc} & = & \left[\begin{array}{cc}
F & ,R\end{array}\right]\cdot\left[\begin{array}{c}
s_{F}^{\star}\\
s_{R}^{\star}
\end{array}\right]\label{eq:cp-1}
\end{eqnarray}
but since $s_{F}^{\star},s_{R}^{\star}\le0$ and $\left[\begin{array}{cc}
F & ,R\end{array}\right]\in\mathbb{R}_{\ge0}^{m\times2n}$, then $c_{lnc}\in\mathbb{R}_{\le0}^{m}$, that is, the objective
coefficients corresponding to logarithmic concentration must be negative.
An intuitive explanation of signs of the linear objective coefficients
is that prior to a stationary point
\[
\left[\begin{array}{c}
v_{f}\\
v_{r}
\end{array}\right]>\exp\left(\left[\begin{array}{cc}
F & ,R\end{array}\right]^{T}\cdot lnc+\left[\begin{array}{c}
lnk_{f}\\
lnk_{r}
\end{array}\right]\right),
\]
therefore minimising $v_{f}$ and $v_{r}$ and maximising $lnc$ will
encourage each exponential cone constraint to be active at an optimal
solution. For any choice of positive values for the entries of $c_{f},c_{r}$
and negative values for the entries of $c_{lnc}$ one can obtain an
optimal solution to Problem (\ref{eq:vk1_s}) where (\ref{eq:expConeBoundaryvf})
and (\ref{eq:expConeBoundaryvr}) are satisfied, provided the linear
constraints are omitted. Therefore, in Section (\ref{sec:SCLPstationarity})
we introduce an algorithm, consisting of a iterative sequence of conic
optimisation problems, to optimise these parameters and prove its
convergence and in Section (\ref{sec:vkSteady}) we demonstrate that
it converges to satisfy (\ref{eq:expConeBoundaryvf}) and (\ref{eq:expConeBoundaryvr}).

\section{\protect\label{sec:SCLPstationarity}Convergence of a sequence of
conic optimisation problems}

\subsection{Conic optimisation: classes of optimal solutions}

In this section, we introduce an algorithm that considers the linear
objective coefficient vectors, $(c_{f},c_{r},c_{lnc})$ in (\ref{eq:vk1_s}),
as parameters to be optimised such that elementary kinetics is satisfied.
This algorithm is explained as an abstract sequence of exponential
conic linear optimisation Problems, each as in (\ref{eq:primalConicLPv}),
rather than directly in terms of kinetics because the result is more
general and the explanation more concise. In Problem \ref{eq:primalConicLPv}
a constraint is said to be \emph{active} if perturbing it would change
the value of the optimal linear objective. Given the input data $\{c,A,b,F,g\}$,
optimality conditions in Equations \ref{eq:COoptimality1} - \ref{eq:COoptimality5}
define an optimal solution to Problem \ref{eq:primalConicLPv} and
each non-zero entry in one of optimal variable vectors in the set
$\{y^{\star},x^{\star},s^{\star}\}$ indicates a constraint that is
active at an optimal solution. Equivalently, the dual variable corresponding
to the active primal constraint, or primal variable corresponding
to the active dual constraint is nonzero. For Problem \ref{eq:primalConicLPv},
assuming the input data $\{A,b,F,g\}$ are invariant, it is the vector
of linear objective coefficients $c\in\mathbb{R}^{n}$ that determines
the constraints that are active at an optimal solution. From this
perspective, the set of optimal variable vectors $\{y^{\star},x^{\star},s^{\star}\}$
is a nonlinear function of a parameter $c$.

The solutions to Problem \ref{eq:primalConicLPv}, each a function
of a particular linear objective coefficient vector $c\in\mathbb{R}^{n}$,
may be classed by the combination of optimality constraints that are
active, equivalently the set of optimal variable vectors $\{y^{\star},x^{\star},s^{\star}\}$
that are non-zero. In certain circumstances, we seek a class of optimal
solution where each of the conic constraints are active, equivalently
each $s^{\star}\ne0$. For example, a nonlinear yet convex conic constraint
is a relaxation of a desired nonlinear and non-convex constraint,
one may seek a optimal solution where that conic constraint is active.
For example, given Problem (\ref{eq:primalConicLPv}), we may seek
to identify a $c^{\bullet}$ such that each exponential cone constraint
is active at an optimal solution, equivalently the corresponding dual
variables are nonzero, that is $s_{\textrm{j}}^{\star}(c^{\bullet})\ne0$
in (\ref{eq:COoptimality1}-\ref{eq:COoptimality5}) for all $j\in1\ldots3k$.
In the following, we approach the problem of identifying a $c^{\bullet}$
such that $s_{\textrm{j}}^{\star}(c^{\bullet})\ne0$ as a major optimisation
problem, where a merit function $\phi:\mathbb{R}^{n}\to\mathbb{R}$
is minimised subject to constraints, denoted $\mathcal{X}$, by solving
an iterative sequence of minor conic optimisation problems, each of
the form \ref{eq:primalConicLPv}. First, in Section \ref{sec:SCLPstationarity}
Theorem \ref{thm:lyap} demonstrates, in an abstract sense, that this
iterative sequence of conic optimisation problems converges to a stationary
point of the merit function, subject to the constraints. We do not
attempt to demonstrate, for the abstract case, that each stationary
point corresponds to a solution where a set of exponential conic constraints
is active, because it depends on the particular properties of the
input data $\{A,b,F,g\}$. However, in Section \ref{sec:vkSteady},
Theorem \ref{thm:stationaryIsSteady} demonstrates, for the particular
constraints that appear in variational elementary kinetics, each stationary
point must correspond to an optimal solution where every exponential
cone constraint in Problem \ref{eq:vk1_s} is active.

\begin{sidewaysfigure}
\includegraphics[width=1\textwidth]{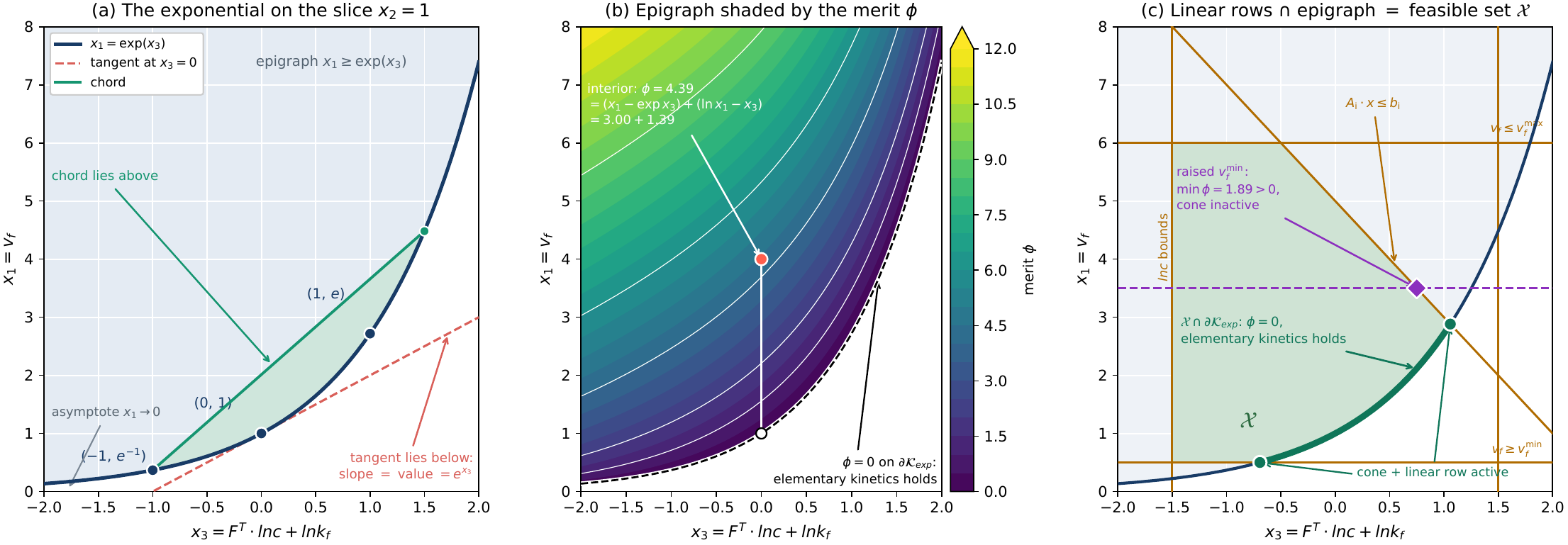}

\caption{\protect\label{fig:expConeFeasibleSet}Building the variational kinetics
feasible set from the exponential cone. All three panels are drawn
on the two-dimensional slice $x_{2}=1$ of the exponential cone $\mathcal{K}_{exp}$
of (\ref{eq:Kexp}), on which cone membership reduces to the epigraph
of the exponential, $x_{1}\ge\exp(x_{3})$. For a single forward reaction
the slice is read kinetically by $x_{1}=v_{f}$ and $x_{3}=F^{T}\cdot lnc+lnk_{f}$,
so the curve $x_{1}=\exp(x_{3})$ is the elementary forward rate law
(\ref{eq:vf}) and the region above it is its conic relaxation $v_{f}\ge\exp(F^{T}\cdot lnc+lnk_{f})$
of (\ref{eq:vf_ge}). (a) The exponential, with the points $(-1,e^{-1})$,
$(0,1)$ and $(1,e)$ marked. Two constructions establish that the
region above the curve is convex, which is what permits the relaxation
to be conic at all: every chord joining two points of the curve lies
above it (green), and every tangent lies below it (red, drawn at $x_{3}=0$,
where the slope equals the value because $\exp$ is its own derivative).
The curve approaches $x_{1}\to0$ as $x_{3}\to-\infty$ without attaining
it, the kinetic content being that the elementary rate law returns
a strictly positive forward flux at every finite concentration. (b)
The same epigraph, shaded by the boundaryseeking merit function (\ref{eq:phi-def}),
which for one forward reaction is $\phi=(v_{f}-\exp(F^{T}\cdot lnc+lnk_{f}))+(\ln v_{f}-F^{T}\cdot lnc-lnk_{f})$.
It is nonnegative throughout the epigraph and vanishes exactly on
the boundary $\partial\mathcal{K}_{exp}$ (dashed), where the rate
law holds with equality. The white segment decomposes $\phi$ at the
interior point $(x_{3},x_{1})=(0,4)$ into its two additive contributions,
the linear gap $v_{f}-\exp(F^{T}\cdot lnc+lnk_{f})=3.00$ and the
logarithmic gap $\ln v_{f}-F^{T}\cdot lnc-lnk_{f}=1.39$, so that
$\phi=4.39$; the two vanish together and only on the boundary, which
is why a strictly positive $\phi$ certifies that at least one exponential
cone constraint is inactive. Note that $\phi$ is strictly \emph{concave},
its Hessian on this slice being $\mathrm{diag}(-1/x_{1}^{2},-\exp(x_{3}))$,
so minimising it is not a convex programme and its minimum over a
compact convex set is attained at an extreme point of that set. (c)
The feasible set $\mathcal{X}$ of (\ref{eq:feasibleSet}), formed
as the intersection of the linear rows $A\cdot x\le b$ (orange) with
the epigraph. Four rows are the box bounds any kinetic model carries,
$v_{f}^{\mathrm{min}}\le v_{f}\le v_{f}^{\mathrm{max}}$, and the
two bounds on $lnc$; the fifth, $A_{\mathrm{i}}\cdot x\le b_{\mathrm{i}}$,
is a general row coupling flux to the rate exponent, as a steady state
row does once projected onto this slice. The bold arc is $\mathcal{X}\cap\partial\mathcal{K}_{exp}$,
the kinetically consistent subset on which $\phi=0$. It is a continuum
rather than a single point, so the merit alone does not determine
the solution and the remaining constraints select within it; its two
endpoints, $(-0.693,0.500)$ and $(1.059,2.883)$, are points at which
the exponential cone and a linear row are simultaneously active. The
dashed line shows the consequence of raising $v_{f}^{\mathrm{min}}$
until the polyhedron no longer reaches the boundary: no feasible point
then satisfies elementary kinetics, and $\min\phi$ over $\mathcal{X}$
is attained at the vertex marked by the diamond, $(0.75,3.50)$, with
$\phi^{\star}=1.89>0$ and the exponential cone inactive throughout
$\mathcal{X}$. This is the strictly positive merit stationary point
that additional assumptions on the data $\{A,b,F,g\}$ are required
to exclude, and detecting it is precisely what the sign of $\phi$
at convergence is for.}
\end{sidewaysfigure}

\subsection{Convergence to stationarity}

The theorem below introduces an iterative sequence of conic optimisation
problems that converges to a stationary point.
\begin{thm}
\label{thm:lyap} Let $A\in\mathbb{R}^{m\times n}$, $b\in\mathbb{R}^{m}$,
$F\in\mathbb{R}^{3k\times n}$, and $g\in\mathbb{R}^{3k}$. Partition
$F$ and $g$ into $v$-row blocks 
\[
F=\begin{bmatrix}F_{1}\\
F_{2}\\
F_{3}
\end{bmatrix},\qquad g=\begin{bmatrix}g_{1}\\
g_{2}\\
g_{3}
\end{bmatrix},\qquad F_{\textrm{i}}\in\mathbb{R}^{v\times n},\ g_{\textrm{i}}\in\mathbb{R}^{v}\ (\textrm{i}=1,2,3).
\]
Define the convex feasible set 
\begin{equation}
\mathcal{X}\coloneqq\left\{ x\in\mathbb{R}^{n}:\ A\cdot x\le b,\ \ F\cdot x+g\in\mathcal{K}_{exp}^{v},\ \ F_{2}\cdot x+g_{2}=1\right\} .\label{eq:feasibleSet}
\end{equation}
Assume $\mathcal{X}$ is nonempty and compact. Define the boundary-seeking
merit function $\phi:\mathbb{R}^{n}\to\mathbb{R}$ by 
\begin{align}
\phi(x) & \coloneqq1^{T}\cdot\left(\,(F_{1}x+g_{1})-\exp(F_{3}x+g_{3})+\ln(F_{1}x+g_{1})-(F_{3}x+g_{3})\,\right),\label{eq:phi-def}
\end{align}
and the major optimisation problem 
\begin{eqnarray}
\underset{x}{\text{min}} & \phi(x)\nonumber \\
\text{s.t.} & x\in\mathcal{X}, & \textrm{}\label{eq:majorOpt}
\end{eqnarray}
where $\phi(x)\ge$0 and 
\[
\phi(x)=0\iff F_{1}x+g_{1}=\exp\left(F_{3}x+g_{3}\right).
\]
Consider the iterative scheme: choose any $x_{0}\in\mathcal{X}$,
and for each $k\ge0$ select 
\begin{equation}
x_{k+1}\in\arg\min_{x\in\mathcal{X}}\ \nabla\phi(x_{\textrm{k}})^{T}\cdot x.\label{eq:VIiteration}
\end{equation}
Then\\
\textbf{ (i)(Concavity)} $\phi$ is concave on $\mathcal{X}$, and
$\nabla\phi(x)$ exists for all $x\in\mathcal{X}$.\textbf{ }\\
\textbf{(ii)(Lyapunov descent)} The sequence $\{\phi(x_{\textrm{k}})\}$
is nonincreasing, and with 
\[
\delta_{\textrm{k}}\coloneqq\nabla\phi(x_{\textrm{k}})^{T}\cdot(x_{\textrm{k}}-x_{k+1})\ \ge\ 0,
\]
one has the one-step decrease bound 
\begin{equation}
\phi(x_{k+1})\ \le\ \phi(x_{\textrm{k}})-\delta_{\textrm{k}}\qquad\forall k\ge0.\label{eq:descent}
\end{equation}
\textbf{(iii) (Summability of stationarity gaps)} $\sum_{\textrm{k}=0}^{\infty}\delta_{\textrm{k}}<\infty$,
hence $\delta_{\textrm{k}}\to0$.\\
\textbf{(iv) (Limit points are stationary, equivalently Variational
Inequality solutions)} Every accumulation point $x^{\bullet}$ of
$\{x_{\textrm{k}}\}$ is a stationary point of Problem \ref{eq:majorOpt},
equivalently, it satisfies the variational inequality 
\begin{equation}
\nabla\phi(x^{\bullet})^{T}\cdot(x-x^{\bullet})\ \ge\ 0\qquad\forall x\in\mathcal{X}.\label{eq:VI}
\end{equation}
\end{thm}

\begin{proof}
Note that $x_{k+1}$ is the optimum of Problem (\ref{eq:primalConicLPv})
with $c\coloneqq\nabla\phi(x_{\textrm{k}})$.\textbf{ Step 1: Concavity
and explicit gradient.} When $x\in\mathcal{X}$, the constraint $F\cdot x+g\in\mathcal{K}_{exp}^{v}$
implies $F_{1}x+g_{1}>0$ componentwise, hence $\ln(F_{1}x+g_{1})$
and $(F_{1}x+g_{1})^{-1}$ are well-defined. In (\ref{eq:phi-def})
each summand is a sum of: a linear function of $x$, plus $\ln(\cdot)$
composed with an affine map, plus $-\exp(\cdot)$ composed with an
affine map, plus another linear function of $x$. Since $\ln$ is
concave on $(0,\infty)$ and $\exp$ is convex on $\mathbb{R}$, the
function $-\exp$ is concave, and composing concave functions with
affine maps preserves concavity on their domains. Therefore $\phi$
is concave on any set where $F_{1}x+g_{1}>0$ componentwise; in particular,
it is concave on $\mathcal{X}$. The gradient of $\phi$ is
\begin{equation}
\nabla\phi(x)=F_{1}^{T}\cdot\Bigl(1+(F_{1}x+g_{1})^{-1}\Bigr)\;-\;F_{3}^{T}\cdot\Bigl(\exp(F_{3}x+g_{3})+1\Bigr).\label{eq:gradPhi}
\end{equation}
For $x\in\mathcal{X}$, $F_{1}x+g_{1}>0$, so $(F_{1}x+g_{1})^{-1}$
is well-defined and $\nabla\phi(x)$ exists on $\mathcal{X}$.

\medskip{}
\textbf{Step 2: Supporting hyperplane inequality for concave $\phi$.}
A standard consequence of concavity and differentiability is that
for all $x,y$ in the domain of $\phi$, 
\begin{equation}
\phi(y)\ \le\ \phi(x)\ +\ \nabla\phi(x)^{T}\cdot(y-x).\label{eq:support}
\end{equation}
We will apply (\ref{eq:support}) with $x=x_{\textrm{k}}$ and $y=x_{k+1}$.

\medskip{}
\textbf{Step 3: Optimality of $x_{k+1}$ in the linearised exponential-conic
subproblem.} By definition (\ref{eq:VIiteration}), $x_{k+1}$ minimises
the linear functional $x\mapsto\nabla\phi(x_{\textrm{k}})^{T}\cdot x$
over $\mathcal{X}$. Hence, for every $x\in\mathcal{X}$, 
\begin{equation}
\nabla\phi(x_{\textrm{k}})^{T}\cdot x_{k+1}\ \le\ \nabla\phi(x_{\textrm{k}})^{T}\cdot x.\label{eq:linopt}
\end{equation}
In particular, taking $x=x_{\textrm{k}}\in\mathcal{X}$ gives 
\begin{equation}
\nabla\phi(x_{\textrm{k}})^{T}\cdot(x_{k+1}-x_{\textrm{k}})\ \le\ 0.\label{eq:stepneg}
\end{equation}
Define the stationarity gap 
\begin{equation}
\delta_{\textrm{k}}\coloneqq\nabla\phi(x_{\textrm{k}})^{T}\cdot(x_{\textrm{k}}-x_{k+1})\ \ge\ 0,\label{eq:delta}
\end{equation}
which is nonnegative by (\ref{eq:stepneg}).

\medskip{}
\textbf{Step 4: Lyapunov descent inequality.} Apply the supporting
inequality (\ref{eq:support}) with $x=x_{\textrm{k}}$ and $y=x_{k+1}$:
\[
\phi(x_{k+1})\ \le\ \phi(x_{\textrm{k}})\ +\ \nabla\phi(x_{\textrm{k}})^{T}\cdot(x_{k+1}-x_{\textrm{k}}).
\]
Using (\ref{eq:stepneg}) (or equivalently (\ref{eq:delta})) yields
\[
\phi(x_{k+1})\ \le\ \phi(x_{\textrm{k}})-\delta_{\textrm{k}},
\]
which is exactly (\ref{eq:descent}) (cf. Figure \ref{fig:concaveDescent}).
In particular, $\phi(x_{k+1})\le\phi(x_{\textrm{k}})$, so $\{\phi(x_{\textrm{k}})\}$
is nonincreasing. Thus $\phi$ is a (discrete-time) Lyapunov function
for the iteration (\ref{eq:VIiteration}).
\end{proof}
\begin{figure}
\includegraphics[width=1\columnwidth]{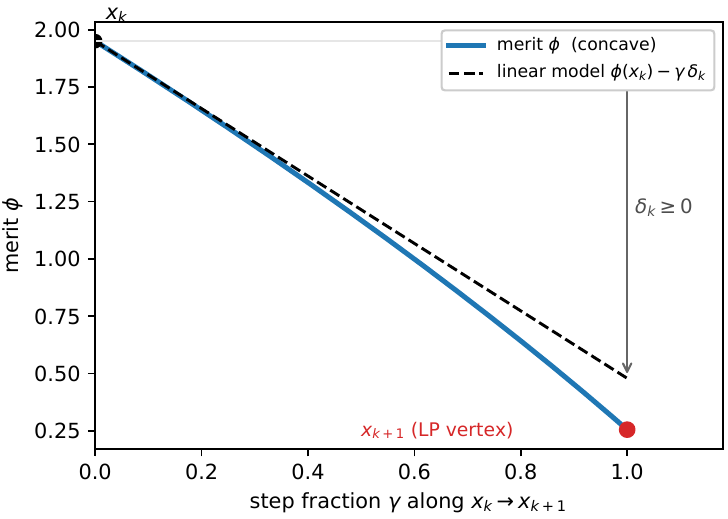}\caption{\protect\label{fig:concaveDescent}Concavity, the full step, and the
guaranteed decrease underlying the convergence proof. Along the segment
from the current iterate $x_{\textrm{k}}$ to the linear-minimisation
vertex $x_{k+1}$, parameterised by the step fraction $\gamma\in[0,1]$,
the merit $\phi$ (blue) is concave, so it lies below its supporting
line $\phi(x_{\textrm{k}})-\gamma\,\delta_{\textrm{k}}$ (dashed),
where $\delta_{\textrm{k}}=\nabla\phi(x_{\textrm{k}})^{T}\cdot(x_{\textrm{k}}-x_{k+1})\ge0$
is the stationarity gap. At the full step $\gamma=1$ (reaching $x_{k+1}$)
this gives $\phi(x_{k+1})\le\phi(x_{\textrm{k}})-\delta_{\textrm{k}}$,
the monotone decrease underlying the Lyapunov argument.}
\end{figure}

\begin{proof}
\medskip{}
\textbf{Step 5: Summability of $\delta_{\textrm{k}}$ and $\delta_{\textrm{k}}\to0$.}
Because $\mathcal{X}$ is compact and $\phi$ is continuous on $\mathcal{X}$,
$\phi$ attains a finite lower bound on $\mathcal{X}$: 
\[
\phi_{\inf}\coloneqq\inf_{x\in\mathcal{X}}\phi(x)\ >\ -\infty.
\]
Sum (\ref{eq:descent}) from $k=0$ to $\textrm{k}=n-1$: 
\[
\phi(x_{\textrm{n}})\ \le\ \phi(x_{0})-\sum_{\textrm{k}=0}^{n-1}\delta_{\textrm{k}}.
\]
Rearrange and use $\phi(x_{\textrm{n}})\ge\phi_{\inf}$: 
\[
\sum_{\textrm{k}=0}^{n-1}\delta_{\textrm{k}}\ \le\ \phi(x_{0})-\phi(x_{\textrm{n}})\ \le\ \phi(x_{0})-\phi_{\inf}.
\]
Letting $N\to\infty$ shows $\sum_{\textrm{k}=0}^{\infty}\delta_{\textrm{k}}<\infty$.
Since each $\delta_{\textrm{k}}\ge0$, it follows that $\delta_{\textrm{k}}\to0$.

\medskip{}
\textbf{Step 6: Accumulation points satisfy the variational inequality.}
Let $x^{\bullet}$ be any accumulation point of $\{x_{\textrm{k}}\}$.
Since $\mathcal{X}$ is compact and all iterates lie in $\mathcal{X}$,
there exists a subsequence $\{x_{\textrm{k}_{\textrm{j}}}\}$ such
that 
\[
x_{\textrm{k}_{\textrm{j}}}\ \to\ x^{\bullet}\in\mathcal{X}\qquad\text{as }j\to\infty.
\]
We prove that $x^{\bullet}$ satisfies (\ref{eq:VI}). Assume for
contradiction that (\ref{eq:VI}) fails. Then there exists $\hat{x}\in\mathcal{X}$
and a scalar $\varepsilon>0$ such that 
\begin{equation}
\nabla\phi(x^{\bullet})^{T}\cdot(\hat{x}-x^{\bullet})\ \le\ -2\varepsilon.\label{eq:contr0}
\end{equation}
Because $\nabla\phi$ is continuous on $\mathcal{X}$ (see (\ref{eq:gradPhi})
and continuity of $\exp$ and reciprocal on $F_{1}x+g_{1}>0$, we
have $\nabla\phi(x_{\textrm{k}_{\textrm{j}}})\to\nabla\phi(x^{\bullet})$.
Also $x_{\textrm{k}_{\textrm{j}}}\to x^{\bullet}$. Therefore the
scalar sequence 
\[
\nabla\phi(x_{\textrm{k}_{\textrm{j}}})^{T}\cdot(\hat{x}-x_{\textrm{k}_{\textrm{j}}})
\]
converges to $\nabla\phi(x^{\bullet})^{T}\cdot(\hat{x}-x^{\bullet})$.
Hence, for all sufficiently large $j$, 
\begin{equation}
\nabla\phi(x_{\textrm{k}_{\textrm{j}}})^{T}\cdot(\hat{x}-x_{\textrm{k}_{\textrm{j}}})\ \le\ -\varepsilon.\label{eq:contr1}
\end{equation}
Rearranging (\ref{eq:contr1}) gives 
\begin{equation}
\nabla\phi(x_{\textrm{k}_{\textrm{j}}})^{T}\cdot(x_{\textrm{k}_{\textrm{j}}}-\hat{x})\ \ge\ \varepsilon.\label{eq:contr2}
\end{equation}
Now use the defining optimality property (\ref{eq:linopt}) for the
step from $x_{\textrm{k}_{\textrm{j}}}$ to $x_{k_{\textrm{j}}+1}$,
choosing $x=\hat{x}$: 
\[
\nabla\phi(x_{\textrm{k}_{\textrm{j}}})^{T}\cdot x_{k_{\textrm{j}}+1}\ \le\ \nabla\phi(x_{\textrm{k}_{\textrm{j}}})^{T}\cdot\hat{x},
\]
which is equivalent to 
\begin{equation}
\nabla\phi(x_{\textrm{k}_{\textrm{j}}})^{T}\cdot(x_{\textrm{k}_{\textrm{j}}}-x_{k_{\textrm{j}}+1})\ \ge\ \nabla\phi(x_{\textrm{k}_{\textrm{j}}})^{T}\cdot(x_{\textrm{k}_{\textrm{j}}}-\hat{x}).\label{eq:lincomp}
\end{equation}
The left-hand side is $\delta_{\textrm{k}_{\textrm{j}}}$ by (\ref{eq:delta}),
so (\ref{eq:lincomp}) and (\ref{eq:contr2}) imply 
\[
\delta_{\textrm{k}_{\textrm{j}}}\ \ge\ \varepsilon\qquad\text{for all sufficiently large }j.
\]
This contradicts $\delta_{\textrm{k}}\to0$ proven in Step~5. Therefore
the assumption (\ref{eq:contr0}) was false, and $x^{\bullet}$ satisfies
(\ref{eq:VI}). Since $x^{\bullet}$ was an arbitrary accumulation
point, every accumulation point satisfies (\ref{eq:VI}). This completes
the proof.
\end{proof}
Theorem \ref{thm:lyap} demonstrates the iteration \ref{eq:VIiteration}
converges to a stationary point $x^{\bullet}$ of the merit function
$\phi$ subject to the constraints $\mathcal{X}$, defined as a solution
to the variational inequality (\ref{eq:VI}). Furthermore, there may
exist multiple solutions to the variational inequality (\ref{eq:VI}).
It is a standard result in variational analysis ($\mathdollar$6.13
in \parencite{rockafellar_variational_1998}) that, for any convex
set $\mathcal{X}$, and any mapping $f(x):\mathcal{X}\rightarrow\mathbb{R}^{n}$,
a solution $x^{\bullet}$ to the variational inequality (\ref{eq:VI})
may be interpreted as
\[
x^{\bullet}\in\underset{x\in\mathcal{X}}{\textrm{arg\;min}}f(x^{\bullet})^{T}\cdot x.
\]
The affine function $x\rightarrow f(x^{\bullet})^{T}\cdot x$ defines
a supporting hyperplane to the convex set $\mathcal{X}$ at $x^{\bullet}$.
The halfspace $\{x:f(x^{\bullet})^{T}\cdot(x-x^{\bullet})\ \ge\ 0\}$
contains all of $\mathcal{X}$ while the hyperplane $f(x^{\bullet})^{T}\cdot(x-x^{\bullet})=0$
touches $\mathcal{X}$ at $x^{\bullet}$, so $x^{\bullet}$ cannot
be a strict interior point of $\mathcal{X}$ because an interior point
admits feasible perturbations in both directions of any vector. A
solution to the variational inequality (\ref{eq:VI}) does not imply
that any $x^{\bullet}$ is on the boundary of any particular combination
of the constraints that define $\mathcal{X}$. Equivalently, it does
not imply that any particular combination of the constraints that
define $\mathcal{X}$ are active. Note that a stationary point may
be interpreted as a fixed point of the iteration (\ref{eq:VIiteration}),
where $x^{\bullet}\coloneqq x_{k+1}=x_{\textrm{k}}.$

There are two types of constraints on the feasible set \ref{eq:feasibleSet}:
(i) a linear constraint, defined by $A_{\textrm{i}}\cdot x\le b_{\textrm{i}}$,
that is active when $A_{\textrm{i}}x=b_{\textrm{i}}$, and, (ii) an
exponential cone constraint defined by $(F_{1}x+g_{1})_{\textrm{i}}\ge\exp(F_{3}x+g_{3})_{\textrm{i}}$,
that is active when $(F_{1}x+g_{1})_{\textrm{i}}=\exp(F_{3}x+g_{3})_{\textrm{i}}$.
The merit function \ref{eq:phi-def}, illustrated in Figure \ref{fig:expConeFeasibleSet},
is zero when all exponential cone constraints are active and strictly
positive when at least one exponential cone constraint is inactive,
that is
\begin{eqnarray*}
\phi(x) & = & 0\iff(F_{1}x+g_{1})=\exp(F_{3}x+g_{3}),\\
\phi(x) & > & 0\iff\exists\;i\in1\ldots v,\;s.t.\;(F_{1}x+g_{1})_{\textrm{i}}>\exp(F_{3}x+g_{3})_{\textrm{i}}.
\end{eqnarray*}
Therefore, a stationary point with at least one exponential cone constraint
inactive (strict interior) may be recognised by a strictly positive
merit function $\phi(x^{\bullet})>0$. To ensure that every stationary
point corresponds to activity of every exponential cone constraint,
one requires additional assumptions on the input data $\{A,b,F,g\}$.
That is, additional assumptions are required to eliminate the existence
of a stationary point where one or more exponential cone constraints
is inactive. It may be that exponential cone and linear constraints
are simultaneously active. Theorem \ref{thm:lyap} proves stationarity
of accumulation points, not convergence of the entire sequence to
a unique point.

Supplementary Section \ref{sec:Adaptive-Sequential-Conic} describes
an adaptive sequential conic linear approximation algorithm that numerically
implements the iterative mathematical algorithm in Theorem \ref{thm:lyap}.
It is but one approach to implements the iterative mathematical algorithm.
It is included for completeness, underlies the numerical experiments
in \ref{sec:Numerical-experiments}, but is agnostic to the biochemical
origins that motivated it and is purely a numerical optimisation construct.
The solver was developed with the assistance of AI coding tools (Claude,
Anthropic; OpenAI Codex).

\subsection{\protect\label{subsec:Correspondence-with-reaction-kinetics-iterative}Correspondence
with reaction kinetics}

Theorem (\ref{thm:lyap}) may be applied reaction kinetics by defining
the exponential boundary-seeking merit function
\begin{equation}
\phi(v_{f},v_{r},lnc)\coloneqq1^{T}\cdot\left(\left[\begin{array}{c}
v_{f}\\
v_{r}
\end{array}\right]-\exp\left(\left[F,R\right]^{T}\cdot lnc+\left[\begin{array}{c}
lnk_{f}\\
lnk_{r}
\end{array}\right]\right)+\ln\left(\left[\begin{array}{c}
v_{f}\\
v_{r}
\end{array}\right]\right)-\left[F,R\right]^{T}\cdot lnc-\left[\begin{array}{c}
lnk_{f}\\
lnk_{r}
\end{array}\right]\right),\label{eq:VKmerit}
\end{equation}
with partial derivatives
\begin{eqnarray*}
\nabla_{v_{f}}\phi & = & 1+1\oslash v_{f}\\
\nabla_{v_{r}}\phi & = & 1+1\oslash v_{r}\\
\nabla_{lnc}\phi & = & -\left[F,R\right]\cdot\left(\exp\left(\left[F,R\right]^{T}\cdot lnc+\left[\begin{array}{c}
lnk_{f}\\
lnk_{r}
\end{array}\right]\right)+1\right).
\end{eqnarray*}
Theorem (\ref{thm:lyap}) proves that an iterative sequence of exponential
conic optimisation problems each of the form of (\ref{eq:conicP}),
generates descent of this merit function over the feasible set, with
convergence to a stationary state. The correspondence between the
general exponential conic optimisation problem in (\ref{eq:conicP})
and reaction kinetic optimisation Problem (\ref{eq:vk1_s}) is provided
in Section \ref{subsec:Correspondence-with-reaction-kinetics-minor}.
In particular, at the $(k+1)^{th}$ iteration of the iterative sequence,
the linear objective coefficients in Problem (\ref{eq:vk1_s}) are
\begin{eqnarray}
c_{v_{f}}^{(k+1)} & \coloneqq\nabla_{v_{f}}\phi= & 1+1\oslash v_{f}^{(k)},\label{eq:cfk+1}\\
c_{v_{r}}^{(k+1)} & \coloneqq\nabla_{v_{r}}\phi= & 1+1\oslash v_{r}^{(k)},\label{eq:crk+1}\\
c_{lnc}^{(k+1)} & \coloneqq\nabla_{lnc}\phi= & -\left[F,R\right]\cdot\left(\exp\left(\left[F,R\right]^{T}\cdot lnc^{(k)}+\left[\begin{array}{c}
lnk_{f}\\
lnk_{r}
\end{array}\right]\right)+1\right).\label{eq:clnck+1}
\end{eqnarray}
where $\{v_{f}^{(k)},v_{r}^{(k)},c_{lnc}^{(k)}\}$ are the optimal
values of the previous minor exponential conic optimisation problem.

\section{\protect\label{sec:vkSteady}Variational kinetics: convergence to
a steady state}

The following theorem establishes sufficient conditions on the input
data $\{F,R\}$ such that $\{v_{f}^{\bullet},v_{r}^{\bullet},lnc^{\bullet}\}$
is a stationary point of the merit function $\phi$ subject to the
constraints in problem (\ref{eq:vk1_s}) (with $b\coloneqq B\cdot w$),
implies that elementary reaction kinetics is satisfied, that is,
\[
\left[\begin{array}{c}
v_{f}^{\bullet}\\
v_{r}^{\bullet}
\end{array}\right]=\exp\left(\left[F,R\right]^{T}\cdot lnc^{\bullet}+\left[\begin{array}{c}
lnk_{f}\\
lnk_{r}
\end{array}\right]\right).
\]

\begin{thm}
\label{thm:stationaryIsSteady}Let $N\in\mathbb{R}^{m\times n}$ and
$b\in\mathbb{R}^{m}$. Let $F,R\in\mathbb{R}_{\ge0}^{m\times n}$,
where $N=R-F$, and let the variables be $v_{f},v_{r}\in\mathbb{R}^{n}$
and $lnc\in\mathbb{R}^{m}$. Define the slack vectors
\[
\delta_{f}:=v_{f}-\exp(F^{T}\cdot lnc+lnk_{f})\ge0,\qquad\delta_{r}:=v_{r}-\exp(R^{T}\cdot lnc+lnk_{r})\ge0,
\]
and the sets $\mathcal{S}_{4}\subseteq\mathcal{S}_{3}\subseteq\mathcal{S}_{2}\subseteq\mathcal{S}_{1}$
by: 
\begin{align*}
\mathcal{S}_{1} & :=\Big\{(v_{f},v_{r},lnc):N(v_{f}-v_{r})=b,\;v_{f}\ge\exp(F^{T}\cdot lnc+lnk_{f}),\;v_{r}\ge\exp(R^{T}\cdot lnc+lnk_{r}),\;v_{f}>0,\;v_{r}>0\Big\},\\[1mm]
\mathcal{S}_{2} & :=\Big\{(v_{f},v_{r},lnc)\in\mathcal{S}_{1}:\ \exists u\in\mathbb{R}^{m}\ \text{s.t.}\ \textrm{sign}(v_{f}-v_{r})=-\textrm{sign}(N^{T}\cdot u)\Big\},\\[1mm]
\mathcal{S}_{3} & :=\Big\{(v_{f},v_{r},lnc)\in\mathcal{S}_{1}:\ \exists u\in\mathbb{R}^{m}\ \text{s.t.}\ \ln(v_{f}\oslash v_{r})=-N^{T}\cdot u\iff\ln(v_{f}\oslash v_{r})\in\mathcal{R}(N^{T})\Big\},\\[1mm]
\mathcal{S}_{4} & :=\Big\{(v_{f},v_{r},lnc)\in\mathcal{S}_{1}:\ v_{f}=\exp(F^{T}\cdot lnc+lnk_{f}),\ v_{r}=\exp(R^{T}\cdot lnc+lnk_{r})\Big\},
\end{align*}
where $lnc\in\mathbb{R}^{m}$. On the domain $\{v_{f},v_{r},lnc\}\in\mathcal{S}_{1}$,
define the continuously differentiable merit function 
\begin{equation}
\phi(v_{f},v_{r},lnc):=\mathbf{1}^{T}\cdot\Big(\begin{bmatrix}v_{f}\\
v_{r}
\end{bmatrix}-\exp([F,R]^{T}\cdot lnc+\left[\begin{array}{c}
lnk_{f}\\
lnk_{r}
\end{array}\right])+\ln\!\Big(\begin{bmatrix}v_{f}\\
v_{r}
\end{bmatrix}\Big)-[F,R]^{T}\cdot lnc-\left[\begin{array}{c}
lnk_{f}\\
lnk_{r}
\end{array}\right]\Big).\label{eq:meritvk}
\end{equation}
A point $(v_{f}^{\bullet},v_{r}^{\bullet},lnc^{\bullet})\in\mathcal{S}_{1}$
is called \emph{a first-order stationary point of $\phi$ over $\mathcal{\mathcal{S}}_{1}$}
if, for every direction $(d_{f},d_{r},d_{c})\in\mathbb{R}^{2n+m}$
for which there exists $\epsilon>0$ such that 
\[
(v_{f}^{\bullet}+\alpha d_{f},\ v_{r}^{\bullet}+\alpha d_{r},\ lnc^{\bullet}+\alpha d_{c})\in\mathcal{S}_{1},\qquad\forall\alpha\in[0,\epsilon],
\]
the directional derivative is non-negative, that is
\begin{equation}
\nabla\phi(v_{f}^{\bullet},v_{r}^{\bullet},lnc^{\bullet})^{T}\cdot\begin{bmatrix}d_{f}\\
d_{r}\\
d_{c}
\end{bmatrix}\ \ge\ 0.\label{eq:directionalDerivative}
\end{equation}
Assume the data satisfy the following:

(2.1) \textbf{(Independent forward and reverse stoichiometry)} For
every $j\in\{1,\dots,n\}$,
\begin{equation}
\textrm{supp}\left(F_{\textrm{:,j}}\right)\cap\textrm{supp}\left(R_{\textrm{:,j}}\right)=\emptyset,\qquad F_{\textrm{:,j}}\ne0,\;R_{\textrm{:,j}}\ne0.\label{eq:Condition2.1}
\end{equation}

(2.2) \textbf{(Cyclic}\textbf{\emph{ }}\textbf{flux}\textbf{\emph{
consistency}}\textbf{)} Every reaction participates in at least one
stoichiometrically balanced cycle. That is, for every $j$, there
exists a $z\in\mathbb{R}^{n}$ such that 
\begin{equation}
Nz=0,\qquad z_{\textrm{j}}\neq0.\label{eq:Condition2.2}
\end{equation}
Then every first-order stationary point $(v_{f}^{\bullet},v_{r}^{\bullet},lnc^{\bullet})\in\mathcal{S}_{1}$
lies in $\mathcal{S}_{4}$. Equivalently, 
\[
\left[\begin{array}{c}
v_{f}^{\bullet}\\
v_{r}^{\bullet}
\end{array}\right]=\exp\left(\left[F,R\right]^{T}\cdot lnc^{\bullet}+\left[\begin{array}{c}
lnk_{f}\\
lnk_{r}
\end{array}\right]\right).
\]
so there is no bidirectional slack and no unidirectional slack in
the kinetic inequalities at $(v_{f}^{\bullet},v_{r}^{\bullet},lnc^{\bullet})$.
\end{thm}

\begin{proof}
We prove three claims:\smallskip{}
\textbf{Claim 1:} No stationary point lies in $\mathcal{\mathcal{S}}_{1}\setminus\mathcal{\mathcal{S}}_{2}$.\smallskip{}
\textbf{Claim 2:} No stationary point lies in $\mathcal{\mathcal{S}}_{2}\setminus\mathcal{\mathcal{S}}_{3}$.\smallskip{}
\textbf{Claim 3:} No stationary point lies in $\mathcal{\mathcal{S}}_{3}\setminus\mathcal{\mathcal{S}}_{4}$.\smallskip{}
Since $\mathcal{\mathcal{S}}_{4}\subseteq\mathcal{\mathcal{S}}_{3}\subseteq\mathcal{\mathcal{S}}_{2}\subseteq\mathcal{\mathcal{S}}_{1}$,
these claims imply any stationary point in $\mathcal{\mathcal{S}}_{1}$
must lie in $\mathcal{\mathcal{S}}_{4}$. To rule out stationary of
the merit function (\ref{eq:meritvk}) at a point $(v_{f}^ {},v_{r}^ {},lnc^ {})$,
it suffices to establish the existence of one feasible direction with
strictly negative directional derivative (\ref{eq:directionalDerivative}).
All stationary statements are made with respect to the feasible set
$\mathcal{\mathcal{S}}_{1}$, even if it leaves $\mathcal{\mathcal{S}}_{2}$
or $\mathcal{\mathcal{S}}_{3}$. Thus, in each claim it suffices to
construct a descent direction that remains feasible in $\mathcal{\mathcal{S}}_{1}$.
For clarity, henceforth in this section we omit the constants $k_{f},k_{r}$.\medskip{}
\textbf{Preliminaries (first order feasibility for kinetic inequalities).}
For each $j$, in the forward direction, if 
\begin{equation}
v_{f}(j)=\exp(F_{\textrm{:,j}}^{T}\cdot lnc)\label{eq:tightForward}
\end{equation}
then we require the direction $d_{c}$ to satisfy $v_{f}(j)\ge\exp(F_{\textrm{:,j}}^{T}\cdot lnc)$
for a small step $\alpha d_{c}$, where $\alpha>0$, that is 
\begin{eqnarray}
v_{f}(j)+\alpha d_{f}(j) & \ge & \exp(F_{\textrm{:,j}}^{T}\cdot(lnc+\alpha d_{c})).\label{eq:forwardInequality}
\end{eqnarray}
The first-order expansion of $\exp(F_{\textrm{:,j}}^{T}\cdot lnc)$
is 
\[
\exp(F_{\textrm{:,j}}^{T}\cdot(lnc+\alpha d_{c}))=\exp(F_{\textrm{:,j}}^{T}\cdot lnc)\big(1+\alpha F_{\textrm{:,j}}^{T}\cdot d_{c}+o(\alpha)\big).
\]
where $o(\alpha)$ denotes a remainder satisfying $o(\alpha)/\alpha\rightarrow0$
as $\alpha\downarrow0$. Adding $\alpha d_{f}(j)$ and then subtracting
$\exp(F_{\textrm{:,j}}^{T}\cdot(lnc+\alpha d_{c}))$ from \ref{eq:tightForward},
we have
\begin{eqnarray*}
v_{f}(j)+\alpha d_{f}(j)-\exp(F_{\textrm{:,j}}^{T}\cdot(lnc+\alpha d_{c})) & = & \exp(F_{\textrm{:,j}}^{T}\cdot lnc)+\alpha d_{f}(j)-\exp(F_{\textrm{:,j}}^{T}\cdot lnc)\big(1+\alpha F_{\textrm{:,j}}^{T}\cdot d_{c}+o(\alpha)\big),\\
 & = & \alpha\left(d_{f}(j)-\exp(F_{\textrm{:,j}}^{T}\cdot lnc)\left(F_{\textrm{:,j}}^{T}\cdot d_{c}\right)\right)+o(\alpha).
\end{eqnarray*}
Thus a sufficient first-order condition to ensure \ref{eq:forwardInequality}
holds for sufficiently small $\alpha$ is
\begin{equation}
d_{f}(j)-\exp(F_{\textrm{:,j}}^{T}\cdot lnc)\,\left(F_{\textrm{:,j}}^{T}\cdot d_{c}\right)\ \ge\ 0.\label{eq:Fcond_nosymbols}
\end{equation}
Similarly, for a reverse direction, if $v_{r}(j)=\exp(R_{\textrm{:,j}}^{T}\cdot lnc)$,
a sufficient first-order condition to ensure $v_{r}(j)+\alpha d_{r}(j)\ge\exp(R_{\textrm{:,j}}^{T}\cdot(lnc+\alpha d_{c}))$
for a sufficiently small $\alpha$ is
\begin{equation}
d_{r}(j)-\exp(R_{\textrm{:,j}}^{T}\cdot lnc)\,\left(R_{\textrm{:,j}}^{T}\cdot d_{c}\right)\ \ge\ 0.\label{eq:Rcond_nosymbols}
\end{equation}
Slack forward or reverse directions impose no restriction on $d_{c}$.\medskip{}
\textbf{Gradient sign facts.} Since $\nabla_{v_{f}}\phi=1+v_{f}^{-1}>0$
and $\nabla_{v_{r}}\phi=1+v_{r}^{-1}>0$, for any $d_{f}\le0$, $d_{r}\le0$
with $(d_{f},d_{r})\neq0$, 
\[
(\nabla_{v_{f}}\phi)^{T}\cdot d_{f}+(\nabla_{v_{r}}\phi)^{T}\cdot d_{r}<0.
\]
\medskip{}
\textbf{Justification of Claim 1.} Take any $(v_{f},v_{r},lnc)\in\mathcal{\mathcal{S}}_{1}\setminus\mathcal{\mathcal{S}}_{2}$
and orthogonally decompose the net flux as $v_{f}-v_{r}=v+z$ with
$v\in\mathcal{R}(N^{T})$, $z\in\mathcal{N}(N)$, so $Nv=b$ and $Nz=0$.
If $z=0$ then $v_{f}-v_{r}=N^{T}\cdot\bar{u}$ for some $\bar{u}$;
with $u:=-\bar{u}$ we get $\mathrm{sign}(v_{f}-v_{r})=-\mathrm{sign}(N^{T}\cdot u)$,
which is the membership condition for $\mathcal{\mathcal{S}}_{2}$,
contradicting $(v_{f},v_{r},lnc)\in\mathcal{\mathcal{S}}_{1}\setminus\mathcal{\mathcal{S}}_{2}$.
Hence $z\neq0$. Decompose $z$ as follows:
\begin{equation}
z=z_{f}-z_{r},\qquad z_{f}\ge0,z_{r}\ge0,\qquad z_{f}\odot z_{r}=0,\label{eq:decomposez}
\end{equation}
so $z_{f}$ and $z_{r}$ are non-negative and have disjoint support.
Set $d_{f}:=-z_{f}$, $d_{r}:=-z_{r}$. Since $z\neq0$ and $z_{f},z_{r}$
have disjoint support, $(d_{f},d_{r})\neq0$; moreover $N(d_{f}-d_{r})=N(-z)=-Nz=0$,
so the steady-state equality $N(v_{f}-v_{r})=b$ is preserved to first
order. The log-concentration $lnc$ does not appear in this equality,
so its direction $d_{c}$ is unconstrained by it: rescaling $d_{c}\mapsto\lambda d_{c}$
leaves $N(d_{f}-d_{r})=0$ intact and $d_{c}$ enters only the exponential
cone tightness conditions, which are positively homogeneous in $d_{c}$.
Hence there is $\epsilon>0$ with $(v_{f}+\alpha d_{f},v_{r}+\alpha d_{r},lnc+\alpha d_{c})\in\mathcal{\mathcal{S}}_{1}$
for $\alpha\in[0,\epsilon]$, and scaling $d_{c}$ then the whole
direction by $\tau\in(0,1)$ yields $\nabla\phi(v_{f},v_{r},lnc)^{T}\cdot(d_{f},d_{r},d_{c})<0$,
the strictness using $(d_{f},d_{r})\neq0$. This contradicts stationarity,
so no stationary point lies in $\mathcal{\mathcal{S}}_{1}\setminus\mathcal{\mathcal{S}}_{2}$.\medskip{}
\textbf{Claim 2.} Take any $(v_{f},v_{r},lnc)\in\mathcal{\mathcal{S}}_{2}\setminus\mathcal{\mathcal{S}}_{3}$.
By the characterisation of $\mathcal{\mathcal{S}}_{3}$ this membership
gives $\ln(v_{f}\oslash v_{r})\notin\mathcal{R}(N^{T})=\mathcal{N}(N)^{\perp}$,
so the component of $\ln(v_{f}\oslash v_{r})$ in $\mathcal{N}(N)$
is nonzero: for any $Z$ whose columns span $\mathcal{N}(N)$ (so
$NZ=0$), $Z^{T}\cdot\ln(v_{f}\oslash v_{r})\neq0$. Hence there is
$z\in\mathcal{N}(N)$ with $z_{\textrm{j}}\neq0$ for some $j$; in
particular $z\neq0$. Split $z=z_{f}-z_{r}$ as in \ref{eq:decomposez}
and set $d_{f}:=-z_{f}$, $d_{r}:=-z_{r}$. As in Claim 1, $(d_{f},d_{r})=(-z_{f},-z_{r})\neq0$
and $N(d_{f}-d_{r})=-Nz=0$, so the direction is nonzero and preserves
the steady-state equality to first order. Choosing $d_{c}$ as in
Claim 1 so the exponential cone tightness conditions (\ref{eq:Fcond_nosymbols})–(\ref{eq:Rcond_nosymbols})
hold with slack gives $\epsilon>0$ with the perturbed point in $\mathcal{\mathcal{S}}_{1}$
for $\alpha\in[0,\epsilon]$, and rescaling $d_{c}$ yields $\nabla\phi^{T}\cdot(d_{f},d_{r},d_{c})<0$.
The strict inequality requires $(d_{f},d_{r})\neq0$: the first-order
descent is carried by the flux terms, whose contribution is proportional
to $(d_{f},d_{r})$ and would vanish if $(d_{f},d_{r})=0$; nonzeroness
is exactly what the $\mathcal{\mathcal{S}}_{2}\setminus\mathcal{\mathcal{S}}_{3}$
membership provides. This contradicts stationarity, so no stationary
point lies in $\mathcal{\mathcal{S}}_{2}\setminus\mathcal{\mathcal{S}}_{3}$.

\medskip{}
\textbf{Claim 3.} Take any $(v_{f},v_{r},lnc)\in\mathcal{\mathcal{S}}_{3}\setminus\mathcal{\mathcal{S}}_{4}$.
Then $\delta_{f}\neq0$ or $\delta_{r}\neq0$, so there exists at
least one index $j$ such that either
\[
v_{f}(j)>\exp(F_{\textrm{:,j}}^{T}\cdot(lnc))\qquad\textrm{or}\qquad v_{r}(j)>\exp(R_{\textrm{:,j}}^{T}\cdot(lnc)).
\]
We treat the case $\delta_{r}\neq0\iff v_{r}(j)>\exp(R_{\textrm{:,j}}^{T}\cdot(lnc))$;
the other case is symmetric. Because $v_{r}(j)$ is strictly above
$\exp(R_{\textrm{:,j}}^{T}\cdot(lnc))$, there exists $\epsilon>0$
such that for all $\alpha\in(0,\epsilon]$, 
\[
v_{r}(j)-\alpha\ge\exp(R_{\textrm{:,j}}^{T}\cdot(lnc)).
\]
Thus decreasing $v_{r}(j)$ slightly does not violate the reverse
inequality at index $j$. By (\ref{eq:Condition2.2}), choose $z\in\mathcal{N}\left(N\right)$
and $z_{\textrm{j}}\ne0$. If necessary replace $z\leftarrow-z$ so
that $z_{\textrm{j}}<0$. Decompose $z$ as in \ref{eq:decomposez}
and define 
\[
d_{f}:=-z_{f},\qquad d_{r}:=-z_{r}.
\]
Then $N(d_{f}-d_{r})=0$, so the equality constraint is preserved
to first order. Since $z_{\textrm{j}}<0$ then $z_{r}(j)=-z_{\textrm{j}}>0$
so $d_{r}(j)=-z_{r}(j)<0$ decreases $v_{r}(j)$. Choose $d_{c}$
as in Claim 1 so that (\ref{eq:Fcond_nosymbols})–(\ref{eq:Rcond_nosymbols})
hold on the tight sets, ensuring feasibility in $\mathcal{\mathcal{S}}_{1}$
for small steps. As in Claim 1, scaling $d_{c}$ if necessary yields
a feasible direction with 
\[
\nabla\phi(v_{f},v_{r},lnc)^{T}\cdot\begin{bmatrix}d_{f}\\
d_{r}\\
d_{c}
\end{bmatrix}<0,
\]
contradicting stationarity. Hence no stationary point lies in $\mathcal{S}_{3}\setminus\mathcal{S}_{4}$.\medskip{}
Combining Claims 1–3, every first-order stationary point $(v_{f}^{\bullet},v_{r}^{\bullet},lnc^{\bullet})\in\mathcal{S}_{1}$
must lie in $\mathcal{S}_{4}$, so all kinetic inequalities are tight,
that is
\[
v_{f}^{\bullet}=\exp(F^{T}\cdot lnc^{\bullet}),\qquad v_{r}^{\bullet}=\exp(R^{T}\cdot lnc^{\bullet}).
\]
\end{proof}

\subsection{Interpretation of Theorem (\ref{thm:stationaryIsSteady}).}

The right hand side of both kinetic inequalities $v_{f}\ge\exp(F^{T}\cdot lnc+lnk_{f})$
and $v_{r}\ge\exp(R^{T}\cdot lnc+lnk_{r})$ are lower bounds on each
one-way rate implied by elementary kinetics. The slack vectors $\delta_{f}$
and $\delta_{r}$ therefore quantify the extent to which the chosen
rates are \emph{in excess} of what is kinetically implied by $lnc$.
In biochemical terms, nonzero slack corresponds to a ``rate assignment''
that cannot be attributed to the stated substrate/product dependencies
alone (e.g., it would implicitly require unmodelled activation, inhibition,
regulation, or additional species).

The merit function $\phi$ is a separable barrier-like penalty that
strictly prefers smaller positive one-way fluxes while also penalising
positive slack vectors. Its gradient with respect to $v_{f}$ and
$v_{r}$ is strictly positive componentwise; consequently, any feasible
perturbation that decreases any component of $v_{f}$ or $v_{r}$
produces an immediate decrease in the merit value, unless doing so
violates feasibility. The proof leverages this fact by explicitly
constructing feasible ``rate-reducing'' directions and showing that,
unless all kinetic inequalities are tight, such a direction always
exists.

The three-stage argument admits a direct biochemical interpretation.
In $\mathcal{S}_{1}\setminus\mathcal{S}_{2}$, the net flux vector
$v_{f}-v_{r}$ contains a component lying in the internal flux cycle
space $\mathcal{N}(N)$. Such cycle flux can circulate without changing
the external exchange balance $b$ \cite{desouki_cyclefreeflux_2015}.
The proof shows that whenever such a cyclic component is present,
one can reduce a subset of one-way rates along the cycle while maintaining
the steady-state balance $N(v_{f}-v_{r})=b$. Because $\phi$ strictly
decreases when one-way rates decrease, any point with a removable
cyclic component cannot be stationary. This corresponds to the biochemical
notion that purely internal futile cycling is disfavoured by the merit:
it is ``unproductive'' with respect to meeting the exchange demands
$b$, yet it increases one-way turnover. $\mathcal{S}_{1}$ corresponds
to the feasible set of Problem \ref{eq:FBA}. $\mathcal{S}_{2}$ corresponds
to the set of optimal solutions to Problem \ref{eq:TFBA}.

In $\mathcal{S}_{2}\setminus\mathcal{S}_{3}$, the natural log ratio
$\ln(v_{f}\oslash v_{r})$ is not compatible with a potential-like
representation in the stoichiometric row space. Biochemically, $\ln(v_{f}\oslash v_{r})$
plays the role of a force term (affinity-like quantity) driving the
net direction of each reaction; requiring $\ln(v_{f}\oslash v_{r})\in\mathcal{R}(N^{T})$
enforces that these forces are consistent with a globally defined
set of chemical potentials. The proof shows that if this compatibility
fails, then there again exists a nullspace direction along which one
can decrease one-way rates without affecting the steady state constraints,
contradicting stationarity. $\mathcal{S}_{3}$ corresponds to optimality
condition \ref{eq:lnvfvrNty} in one of the optimality conditions
of Problem \ref{eq:EntropicFBA}, which requires $\ln(v_{f}\oslash v_{r})\in\mathcal{R}(N^{T})$.

Finally, in $\mathcal{S}_{3}\setminus\mathcal{S}_{4}$, the system
is already ``potential-consistent'' (the ln forward/reverse ratios
can be written as $-N^{T}\cdot u$), but at least one kinetic inequality
is slack. In biochemical terms, this means that even though the directionality
pattern is consistent with a global potential, at least one reaction
has an excess one-way rate beyond what the concentrations would imply.
The cyclic flux consistency assumption ensures that for any reaction
index $j$ exhibiting slack, there exists an internal cycle that includes
reaction $j$. This guarantees a feasible cycle-based perturbation
that reduces the slack one-way rate at $j$ while preserving the net
exchanges $b$. Hence slack cannot persist at stationarity: the only
stationary configurations are those in which every one-way rate is
exactly matched to the monomial kinetics implied by $\ln c$, i.e.,
$\delta_{f}=\delta_{r}=0$. The innermost set $\mathcal{S}_{4}$ is
the set of elementary kinetic steady states, in which both kinetic
inequalities hold with equality (\ref{eq:elementaryKineticsF}, \ref{eq:elementaryKineticsR});
it is the feasible set of the variational kinetics problem \ref{eq:vk1},
over which the algorithm optimises.

The proof is structured in stages because different types of non-physical
or non-minimal behaviour arise from distinct geometric features of
the feasible set. Internal cycling ($\mathcal{S}_{1}\setminus\mathcal{S}_{2}$),
thermodynamic inconsistency ($\mathcal{S}_{2}\setminus\mathcal{S}_{3}$),
and kinetic slack ($\mathcal{S}_{3}\setminus\mathcal{S}_{4}$) correspond
to progressively stronger notions of feasibility that may be difficult
to be ruled out by a single argument. Each stage identifies one source
of excess in the merit function and constructs a descent direction
tailored to that mechanism, while maintaining feasibility with respect
to all constraints. This staged approach mirrors a biochemical hierarchy
from mass-balance yet net flux direction inconsistent with thermodynamics
$\left(\mathcal{S}_{1}\setminus\mathcal{S}_{2}\right)$, to steady
state and net flux direction consistent with thermodynamics, yet inconsistent
with the ratio of forward over reverse unidirectional fluxes $\left(\mathcal{S}_{2}\setminus\mathcal{S}_{3}\right)$,
to steady state and full thermodynamic consistency yet kinetic inconsistency
$\left(\mathcal{S}_{3}\setminus\mathcal{S}_{4}\right)$, and makes
explicit how the merit function eliminates each form of inconsistency
in turn.

Overall, the proof formalises the following biochemical principle:
if a flux configuration contains any removable internal cycling or
any excess one-way turnover beyond that implied by concentrations,
then the merit function provides a direction of feasible improvement
that reduces total one-way turnover while maintaining the same external
demands. Under assumptions (\ref{eq:Condition2.1})-(\ref{eq:Condition2.2}),
the only points where this is no longer possible are those where all
kinetic inequalities are tight, meaning the rates are fully explained
by the kinetic monomials for the same $lnc$. This justifies the method
as a constructive mechanism for eliminating futile internal cycling
and enforcing kinetically coherent one-way rates consistent with steady-state
exchange requirements.

At a stationary point of Theorem (\ref{thm:lyap}) the objective in
Problem (\ref{eq:vk1_s}) is 
\begin{eqnarray*}
c_{v_{f}}^{T}\cdot v_{f}^{\bullet} & \coloneqq & \left(1+1\oslash v_{f}^{\bullet}\right)^{T}\cdot v_{f}^{\bullet}=\mathbf{1}^{T}\cdot v_{f}^{\bullet}+n\\
c_{v_{r}}^{T}\cdot v_{r}^{\bullet} & \coloneqq & \left(1+1\oslash v_{r}^{\bullet}\right)^{T}\cdot v_{r}^{\bullet}=\mathbf{1}^{T}\cdot v_{r}^{\bullet}+n\\
c_{lnc}^{T}\cdot lnc^{\bullet} & \coloneqq & -\left(\left[F,R\right]\cdot\exp\left(\left[F,R\right]^{T}\cdot lnc^{\bullet}+\left[\begin{array}{c}
lnk_{f}\\
lnk_{r}
\end{array}\right]\right)+1\right)^{T}\cdot lnc^{\bullet}=\nabla\phi(lnc^{\bullet})^{T}\cdot lnc^{\bullet}
\end{eqnarray*}
which may be interpreted as minimising the sum of unidirectional fluxes
and maximising the rate of consumption of every species weighted by
the logarithm of species concentration.

\section{Additional constraints and regularisation}

Section \ref{sec:Variational-elementary-kinetics}, \ref{sec:SCLPstationarity}
and \ref{sec:vkSteady} establish convergence properties for an algorithm
to obtain a steady state, as defined by Eq. \ref{eq:massBalanceKinetics}.
Beyond that, there are additional constraints that can be added, such
as moiety conservation and thermodynamic constraints on kinetic parameters.
Strictly, Theorem \ref{thm:stationaryIsSteady} applies to solutions
to Eq. \ref{eq:massBalanceKinetics}, but in practice, addition of
the constraints below is observed to be numerically compatible with
convergence to satisfaction of elementary kinetics also.

\subsection{Moiety conservation constraints}

The moiety conservation constraints in Eq. (\ref{eq:moietyConservation})
are linear in linear concentrations, while the formulation of variational
elementary kinetics in Problem (\ref{eq:vk1}) is expressed in terms
of logarithmic concentrations, therefore to add moiety conservation
to Problem (\ref{eq:vk1}), one can employ an exponential cone to
constrain the relationship between linear and logarithmic concentration
variables then add terms to optimise to the boundary of this cone,
to give the conic optimisation problem

\begin{eqnarray}
\underset{c,lnc}{\text{min}}\qquad c_{c}^{T}\cdot c+c_{lnc}^{T}\cdot lnc\nonumber \\
\text{s.t.}\qquad L\cdot c=L\cdot c(0),\label{eq:moietyConservationVK}\\
\left(\begin{array}{c}
c\\
1\\
lnc
\end{array}\right)\in\mathcal{K}_{exp}^{m},
\end{eqnarray}
where at an optimum we have $c_{\textrm{i}}^{\star}=\exp\left(lnc_{\textrm{i}}^{\star}\right)$
when $c_{\textrm{i}}$ and $lnc_{\textrm{i}}$ lie on the exponential
face of the $i^{th}$ exponential cone.

\subsection{Thermodynamic constraints on elementary kinetic parameters}

In an extension to (\ref{eq:vk1}), logarithmic forward and reverse
kinetic parameters may be modelled as variables, in which case it
is possible to implement thermodynamic constraints on kinetic parameters
using
\begin{eqnarray}
lnk_{f}-lnk_{r}+N{}^{T}\cdot u^{\circ} & = & 0,\label{eq:thermoConstrainedKineticParam-1}
\end{eqnarray}
which is linear system of equations in logarithmic variables.

\subsection{Regularisation of kinetic steady states}

Assuming a solution exists to (\ref{eq:massBalanceKinetics}) implies
that there exists an optimal solution to (\ref{eq:vk1}), where each
exponential cone constraint is active, that is, the optimal solution
is on the exponential face of each exponential cone. However, it may
occur that there does not exist a steady state solution to (\ref{eq:massBalance})
that also satisfies the elementary reaction rate laws (\ref{eq:vf=00003D})
and (\ref{eq:vr=00003D}), as well as box constraints on internal
and external reaction rates. As described in \ref{subsec:Moiety-conservation},
replacement of stoichiometrically inconsistent exchange reactions
with perpetireactions, that are stoichiometrically consistent but
driven by intentionally thermodynamically infeasible kinetic parameters,
is the principled and theoretically supported approach to ensure there
exists a non-equilibrium steady state. However, in practice there
are a large cadre of established modes that employ exchange reactions
so it is useful to have the option to solve for a regularised steady
state solution that penalises deviation from steady state in case
the given exchange reactions are not compatible with a kinetically
feasible steady state. This can be achieved by the addition of a quadratic
penalty on a regularisation variable, $r\in\mathbb{R}^{m}$, which
is conically representable with a rotated quadratic cone, that is
\begin{eqnarray}
\underset{v_{f,}v_{r},w,r,t_{r},}{\text{min}}\qquad c_{t_{r}}\cdot t_{r}\nonumber \\
\text{s.t.}\qquad N\cdot(v_{f}-v_{r})+r+B\cdot w=0,\label{eq:qpSteadyStateRegularisation}\\
\left(\begin{array}{c}
t_{r}\\
1\\
H_{r}\cdot r
\end{array}\right)\in\mathcal{Q}^{2+m}.
\end{eqnarray}
where $t_{r}\in\mathbb{R}_{\ge0}$ is an auxiliary variable. Penalisation
can be weighted for or against different molecular species by selection
of a suitable set of weights on the diagonal of $H_{r}\in\mathbb{R}^{m\times m}$,
but the default is $H_{r}\coloneqq I$, Regularisation enables one
to relax one or more steady state constraints in (\ref{eq:massBalance})
yet still satisfy elementary reaction kinetics, in the case where
a solution to (\ref{eq:massBalanceKinetics}) does not exist.

\subsection{Optimisation of network states}

A key feature of constraint-based modelling is the ability to optimise
over a feasible set of network states. In flux balance analysis, such
optimisation is over a feasible set of steady state fluxes. Typically,
optimisation is of one or more exchange fluxes, rather than internal
fluxes. Although optimisation over any linear combination of internal
and external fluxes is possible, internal fluxes with support in the
nullspace of the internal stoichiometric matrix must be bounded by
given lower or upper bounds to avoid an unbounded optimisation problem,
if any of those reactions is optimised. In entropic flux balance analysis
\parencite{fleming_variational_2012}, every thermodynamically feasible
steady state net flux can be obtained as a function of parameters
corresponding to internal reactions that may be interpreted as prior
information in a relative entropy optimisation problem \parencite{amestica-toledo_thermodynamically_2026}.
Entropic flux balance analysis can also include optimisation of any
linear combination of external net fluxes, via a trade off between
(relative) entropy optimisation of internal fluxes and optimisation
of exchange fluxes. However, this is not optimisation \emph{over}
a set of thermodynamically feasible fluxes, because that set is non-convex.
In variational kinetics, it is currently not possible to arbitrarily
optimise internal net fluxes, concentrations, or kinetic parameters,
without interfering with convergence to satisfaction of elementary
reaction kinetics, because it is the objective coefficients corresponding
to unidirectional fluxes, concentrations ($\pm$ kinetic parameters)
that are iteratively optimised to ensure that elementary reaction
kinetics is satisfied. However, one can optimise over the set of net
external fluxes as in flux balance analysis, by adding a linear objective
over exchange reactions, that is $c_{w}^{T}\cdot w$. In practice,
in each conic optimisation problem in the terms in the linear objective
compete with one another, therefore it is beneficial to add scalar
parameter to balance optimisation of exchange fluxes with satisfaction
of kinetics, that is \textbackslash alpha\_\{w\}c\_\{w\}\textasciicircum\{T\}\textbackslash cdot
w.

each term in the linear objective. Therefore, to establish priorities
between the different objectives within the combined formulation of
variational elementary kinetics, we introduce non-negative scalar
weights, denoted $\alpha_{x}\in\mathbb{R}_{\ge0}$ where the subscript
$x$ is replaced by a symbol that corresponds to the primal variable.
That is, the objective in the combined formulation of variational
elementary kinetics becomes

\begin{eqnarray}
\underset{}{\text{min}}\qquad\alpha_{v_{f}}c_{v_{f}}^{T}\cdot v_{f}+\alpha_{v_{r}}c_{v_{r}}^{T}\cdot v_{r}+\alpha_{lnk_{f}}c_{lnk_{f}}^{T}\cdot lnk_{f}+\alpha_{lnk_{r}}c_{lnk_{r}}^{T}\cdot lnk_{r}\ldots\label{eq:VKObjectiveKineticsAlpha}\\
+\alpha_{lnc}c_{lnc}^{T}\cdot lnc\ldots+\alpha_{c}c_{c}^{T}\cdot c\ldots\label{eq:VKObjectiveLinConcAlpha}\\
+\alpha_{t_{r}}c_{t_{r}}^{T}\cdot t_{r}\ldots\label{eq:VKObjectiveMassBalanceRegulariserAlpha}\\
+\alpha_{w}c_{w}^{T}\cdot w.\label{eq:VKObjectiveOptExternalFluxAlpha}
\end{eqnarray}
The relative values of these scalar weights substantially affects
the type of variational kinetic solution obtained.

\section{\protect\label{subsec:Comprehensive-formulation-of-VK}Variational
elementary kinetics}

In this section, each of the conic optimisation problems in the preceding
sections is combined into a single conic optimisation problem. Each
constraint is represented with the corresponding dual variables, where
$y$ denotes a dual variable to a linear equality constraint, $s$
denotes a dual variable to a cone constraint and $z$ denotes a dual
variable to a box constraint. The subscript to these dual variables
is chosen to reflect a correspondence to a primal term. The combined
formulation of variational elementary kinetics is

\begin{eqnarray}
\underset{}{\text{min}}\qquad c_{v_{f}}^{T}\cdot v_{f}+c_{v_{r}}^{T}\cdot v_{r}+c_{lnc}^{T}\cdot lnc+c_{lnk_{f}}^{T}\cdot lnk_{f}+c_{lnk_{r}}^{T}\cdot lnk_{r}\ldots\label{eq:VKObjectiveKinetics}\\
+c_{c}^{T}\cdot c\ldots\label{eq:VKObjectiveLinConc}\\
+c_{t_{r}}^{T}\cdot t_{r}\ldots\label{eq:VKObjectiveMassBalanceRegulariser}\\
+c_{w}^{T}\cdot w\label{eq:VKObjectiveOptExternalFlux}\\
\text{s.t.}\qquad N\cdot(v_{f}-v_{r})+r+B\cdot w=0\;:y_{N}\label{eq:VKConstraintMassBalance}\\
lnk_{f}-lnk_{r}+N{}^{T}\cdot u^{\circ}=0\;:y_{u^{\circ}}\label{eq:VKConstraintThermoKinParam}\\
L\cdot c=L\cdot c(0)\;:y_{L}\label{eq:VKConstraintMoietyConservation}\\
\left(\begin{array}{c}
v_{f}\\
1\\
F^{T}\cdot lnc+lnk_{f}
\end{array}\right)\in\mathcal{K}_{exp}^{n}:\left(\begin{array}{c}
s_{vf}\\
s_{f1}\\
s_{F}
\end{array}\right),\label{eq:Pexp_F}\\
\left(\begin{array}{c}
v_{r}\\
1\\
R^{T}\cdot lnc+lnk_{r}
\end{array}\right)\in\mathcal{K}_{exp}^{n}:\left(\begin{array}{c}
s_{vr}\\
s_{r1}\\
s_{R}
\end{array}\right),\label{eq:eq:Pexp_R}\\
\left(\begin{array}{c}
c\\
1\\
lnc
\end{array}\right)\in\mathcal{K}_{exp}^{m}:\left(\begin{array}{c}
s_{c}\\
s_{c1}\\
s_{lnc}
\end{array}\right),\label{eq:eq:Pexp_c}\\
\left(\begin{array}{c}
t_{r}\\
1\\
H_{r}\cdot r
\end{array}\right)\in\mathcal{Q}^{2+m}:\left(\begin{array}{c}
s_{t_{r}}\\
s_{r1}\\
s_{H_{r}}
\end{array}\right).\label{eq:Prquad_Hr}
\end{eqnarray}
The following are the intent of the terms in the objective. The linear
coefficients in the terms (\ref{eq:VKObjectiveKinetics}) are iteratively
updated to optimise fluxes, logarithmic concentrations, and logarithmic
kinetic parameters (unless they are fixed) so the corresponding exponential
cone constraints representing elementary kinetics are active at a
stationary state (\ref{eq:Pexp_F}). The linear coefficients in the
term (\ref{eq:VKObjectiveLinConc}) is iteratively updated to so the
corresponding exponential cone constraints ((\ref{eq:Pexp_F}),\ref{eq:eq:Pexp_R},\ref{eq:eq:Pexp_c})
are active at a stationary state and therefore $c^{\bullet}=\exp(lnc^{\bullet})$.
The term (\ref{eq:VKObjectiveMassBalanceRegulariser}) implements
regularisation of steady state constraints. The term (\ref{eq:VKObjectiveOptExternalFlux})
implements linear optimisation of net external reaction flux. The
following are the intent of the equality constraints. Equation (\ref{eq:VKConstraintMassBalance})
implements regularised steady state, in conjunction with the term
(\ref{eq:VKObjectiveMassBalanceRegulariser}) and the rotated quadratic
cone (\ref{eq:Prquad_Hr}). Equation (\ref{eq:VKConstraintThermoKinParam})
implements thermodynamic constraints on logarithmic elementary kinetic
parameters. Finally, Eq. (\ref{eq:VKConstraintMoietyConservation})
implements moiety conservation.

When implementing this conic optimisation problem numerically, one
must encode upper and lower bounds for each variable, but they may
be specified to be unbounded, except in the case where they are required
to be fixed to be a certain given value specified as prior data, or
where an unbounded variable may result in an unbounded optimisation
problem, e.g., upper bounds on elementary fluxes and lower bounds
on logarithmic concentrations. The following additional inequality
constraints are used to constrain unidirectional fluxes, net internal
reaction fluxes, net exchange reaction fluxes, logarithmic concentrations,
logarithmic kinetic parameters and logarithmic standard Gibbs energies
of formation
\begin{eqnarray}
v_{f}\le u_{v_{f}} &  & :-z_{v_{f}},\label{eq:boundsNetFlux}\\
v_{r}\le u_{v_{r}} &  & :-z_{v_{r}},\\
l_{v}\le v_{f}-v_{r}\le u_{v} &  & :z_{v}\coloneqq z_{lv}-z_{uv},\\
l_{w}\le w\le u_{w} &  & :z_{w}\coloneqq z_{lw}-z_{uw},\label{eq:boundsExchangeFlux}\\
l_{lnc}\le lnc\le u_{lnc} &  & :z_{lnc}\coloneqq z_{llnc}-z_{ulnc},\label{eq:boundsLnc}\\
l_{lnk_{f}}\le lnk_{f}\le u_{lnk_{f}} &  & :z_{lnk_{f}}\coloneqq z_{llnk_{f}}-z_{ulnk_{f}},\label{eq:boundslnkf}\\
l_{lnk_{r}}\le lnk_{r}\le u_{lnk_{r}} &  & :z_{lnk_{r}}\coloneqq z_{llnk_{r}}-z_{ulnk_{r}},\label{eq:boundslnkr}\\
l_{u^{\circ}}\le u^{\circ}\le u_{u^{\circ}} &  & :z_{u^{\circ}}\coloneqq z_{lu^{\circ}}-z_{uu^{\circ}}.\label{eq:boundsu0}
\end{eqnarray}
For each box constraint, a net dual vector and two non-negative dual
vectors corresponding to lower and upper bound constraints are introduced,
with dimensions corresponding to the primal variable concerned, e.g.,
$l_{v}\in\mathbb{R}^{n}$ and $u_{v}\in\mathbb{R}^{n}$ denote lower
and upper bounds on net flux $v_{f}-v_{r}$, $z_{lv}\in\mathbb{R}^{n}$
denotes a non-negative dual vector to the lower bound constraint and
$z_{uv}\in\mathbb{R}^{n}$ denotes a non-negative dual vector to the
upper bound constraint, with the dual vector corresponding to box
constraints on net flux defined as the difference between these two
vectors, that is $z_{v}\coloneqq z_{lv}-z_{uv}$. The upper bounds
on the unidirectional fluxes are one-sided, so each introduces a single
non-negative dual vector, $z_{v_{f}},z_{v_{r}}\in\mathbb{R}_{\ge0}^{n}$,
rather than a net dual vector. The Lagrangian corresponding to the
combined formulation of variational kinetics expressed as Problem
(\ref{eq:VKObjectiveKinetics})-(\ref{eq:boundsu0}) is provided in
Supplementary Section \ref{subsec:Comprehensive-formulation-of-VK}.

\section{\protect\label{sec:Numerical-experiments}Numerical experiments}

Numerical experiments are presented as a rendered computational narrative
(Supplementary\_File\_1.html \url{https://doi.org/10.5281/zenodo.21633862})
generated by the MATLAB R2024b (Mathworks Inc.) numerical computing
environment, using the COBRA Toolbox (v3.8beta, specifically SHA-1
shorthand: 4713424eb) \parencite{heirendt_creation_2019} that accesses
an implementation of the sequential conic solver \url{https://doi.org/10.5281/zenodo.21633862}
\href{https://github.com/Digital-Metabolic-Twin-Centre/varkin}{https://github.com/Digital-Metabolic-Twin-Centre/varkin}
(SHA-1 shorthand: b130be3) and an industrial quality conic optimisation
solver (MOSEK Version 11.2.0, feasibility and optimality tolerance
$1e-6$), on a workstation (x86\_64, Intel(R) Core(TM) i9-10980XE
CPU @ 3.00GHz) running a Linux operating system (Linux 6.17.0-35-generic
\#35\textasciitilde 24.04.1-Ubuntu). The computational narrative
can be applied to a variety of genome-scale metabolic models but in
this section, results are summarised for experiments with a stoichiometrically,
flux and thermodynamically flux consistent subset \cite{fleming_cardinality_2023,preciat_xomicstomodel_2025}
of a generic human genome-scale metabolic model (Recon3, \cite{brunk_recon3d_2018}),
containing 5,835 metabolites and 8,791 internal reactions. 

Supplementary File 1 contains three numerical experiments. The three
experiments differ in what boundary condition b is and what extra
constraints are imposed. Experiment 1 (VK1) — satisfaction of elementary
kinetics. The plainest case: find a concentration vector consistent
with the elementary kinetics at steady state, kinetic parameters fixed
at their given values, no synthetic target to recover and no moiety
conservation. Constraints are the exponential cone kinetics, the hard
steady state equality, and the concentration/flux bounds; all penalties
off. It converged trivially a single major iteration with a clean
conic certificate (less than kinetic equality tolerance $tolEXP\coloneqq5e-5$).

Experiment 2 (VK2) — recover a steady state from a known boundary
condition. Here a random, thermodynamically feasible kinetic steady
state is generated, from a random concentration vector and fixed kinetic
parameters, $lnk_{f}=lnk_{r}=0$, the corresponding boundary flux
$b=N\cdot(\exp(\ln(k_{f})+F^{T}\cdot\ln(c))-\exp(\ln(k_{r})+R^{T}\cdot\ln(c)))$
is computed from and the solver must find a steady state satisfying
that same $b$ (with kinetic parameters fixed, but no knowledge of
the generating $c$). This is the controlled test of whether the algorithm
can recover it. Constraints are the same as VK1, the steady-state
equality against the generated $b$, and bounds scaled to include
the test concentrations and fluxes, but no moiety conservation. It
converged, with the help of several reflected kinetic leaf iterates
(cf. Supplementary Section \ref{sec:Adaptive-Sequential-Conic}):
to less than the kinetic equality tolerance $tolEXP\coloneqq5e-5$.
An interesting point to note is that, at least for this model, the
constraints defined uniquely the net internal flux.

Experiment 3 (VK3) — recover a moiety-conserved steady state. Same
generate-and-recover setup as VK2, with the added requirement that
the recovered steady state also respect moiety conservation, so the
feasible set gains the conserved-moiety linear equalities on top of
VK2's kinetics, steady-state, and bound blocks. Converged maximum
kinetic equality violation to approximately the kinetic equality tolerance
$tolEXP\coloneqq5e-5$, at major iteration 18. An interesting point
to note is that, at least for this model, the constraints defined
uniquely the net internal flux, and less so the unidirectional fluxes,
but not concentrations.

\section{Discussion}

Variational kinetics is the first computationally tractable modelling
method that enables satisfaction of elementary reaction kinetic rate
law constraints at genome-scale without recourse to mathematical approximation.
Given a biochemical network with $m$ molecular species and $n$ reactions,
the problem of finding an $m$ dimensional concentration vector within
the non-convex set satisfying elementary reaction kinetics and equilibrium
or non-equilibrium steady state constraints is expressed as the minimisation
of a strictly concave function over the intersection of convex exponential
cones and an affine subspace defined by linear equalities and inequalities.
The steady state constraint is the affine subspace, the concentration
and unidirectional flux bounds are the linear inequalities, and the
relaxations of the elementary kinetic constraints are the exponential
cones. These coexist in one convex conic feasible set, and it is over
that set that the strictly concave function is minimised. That function
is strictly positive on the feasible set and attains a zero local
minimum if, and only if, the $2n$ elementary kinetic constraints
are satisfied, a pair for each reversible reaction. Its minimisation
by a particular sequence of conic optimisation problems is guaranteed
to find a solution to a variational kinetic problem, provided that
such a solution exists.

The difficulty that this addresses is intrinsic to the problem. A
set is non-convex when there exists a line between two points in that
set where an interval of that line lies outside the set. It has long
been recognised that the set of thermodynamically feasible steady
state fluxes is non-convex \parencite{qian_stoichiometric_2003},
and the set of concentrations that also satisfy elementary kinetic
rate law constraints is likewise non-convex. Furthermore, the fundamental
equation defining the set of non-equilibrium steady states of a biochemical
network, Eq. (\ref{eq:f(x)fundamental}), has an asymmetric gradient,
that is a transposed Jacobian that is not symmetric, so a non-equilibrium
steady state cannot, in general, be obtained by minimisation of scalar
valued function. This fact alone eliminates a wide variety of established
optimisation algorithms. Numerically, Eq. (\ref{eq:f(x)fundamental})
is not a monotone function for a wide variety of genome-scale metabolic
models, and previously we demonstrated a specific biochemical network
that was provably not monotone, but rather a generalised monotone
function termed a duplomonotone function \parencite{artacho_globally_2014}.
Whether all biochemical networks give rise to a duplomonotone function
is an open question, though numerical tests support such a conclusion.
It is the difficulty of categorising this fundamental equation in
terms specific enough to admit established algorithms at high dimension
that motivates the development of tailored algorithms.

An optimisation problem is convex when a convex function is minimised,
equivalently when a concave function is maximised, over a convex set.
It is conic when a linear function is optimised over a convex conic
set, that is an intersection of convex cones. In the optimisation
literature conic optimisation is therefore described as a subclass
of convex optimisation, in which the objective is linear and each
constraint defines a convex set, and since the minimisation of any
convex function can be expressed as the minimisation of a linear function
subject to conic constraints \parencite{ben-tal_lectures_2001}, that
subclass is a structured one rather than a restrictive one, admitting
polynomial-time algorithms whenever the cone is tractable. Minimising
a strictly concave function over a conic set satisfies the convex
set requirement but violates the convex objective requirement, so
the minimisation lies outside the convex class. The direction matters,
because maximising that same strictly concave function over the same
set is a convex problem. When a strictly concave function is minimised
over a conic set it attains its local minima on the extreme points
of the feasible set, and where those cones have curved and continuous
boundaries there is, in general, a continuum of local minima. From
a biological perspective the exponential cone has substantial value
for the mathematical modelling of biochemical networks, leveraging
fundamental and applied algorithmic developments in this area \parencites{chares_cones_2009}{andersen_mosek_2000}{dahl_primal-dual_2022}.
Conic formulations enable powerful, structure-exploiting solvers,
but they require the non-linear parts of a modelling problem to be
expressible in terms of particular convex cones, which we have established
is the case for non-linear kinetic and thermodynamic constraints.
It is likely that there are other, lesser known cones of substantial
relevance to biology. We envisage that this potential will encourage
greater appreciation of conic optimisation among the biological modelling
community, and motivate mathematical progress on the characterisation
of novel cones that admit efficient optimisation algorithms.

A key advantage of variational kinetics is that it rests on a sequence
of conic optimisation problems, and so inherits the tractability,
reliability and scalability of convex optimisation: efficient algorithms
of polynomial-time complexity, solutions that are robust and reproducible,
increasing solver support, and a well-developed duality theory. Every
kinetic steady state corresponds to a local minimum of the strictly
concave merit function, and every local minimum of that function is
a global minimum, which avoids the non-global local minima, heuristics
and starting point sensitivity that complicate both analysis and computation
in alternative approaches. It is expected that, in general, there
exist multiple kinetic steady states compatible with the given constraints,
and it is a strength of variational kinetics that the sequence of
conic optimisation problems is not intrinsically biased toward any
particular one of them. That absence of bias is not the same as independence
of the initial point. Each conic optimisation problem in the sequence
starts from an initial feasible point, and which of the admissible
kinetic steady states is returned does depend on that point, so further
analysis is required to establish the nature of this dependence. We
hypothesise that such analysis will be tractable, given the substantial
literature relating perturbations in convex optimisation parameters
to perturbations in their solutions, e.g., \parencite{dontchev_primal-dual_2001}.

Although exponential cones represent relaxations of elementary kinetics,
the convexity of the feasible set guarantees that an optimal solution
satisfies elementary reaction kinetic constraints to within the numerical
tolerances set by the optimisation solver. This presumes that a steady
state kinetic solution exists \parencite{fleming_mass_2012}, that
the objective driving satisfaction of elementary kinetics is not dominated
by a competing objective, and that the numerical values of the data,
e.g., the stoichiometric coefficients \parencite{ma_reliable_2017},
are sufficiently well scaled in comparison to the precision of the
numerical implementation. The computational experiments reported herein
exercise these properties in three regimes. The first isolates satisfaction
of elementary reaction kinetics, with the competing penalties inactive
and violation of mass balance admitted under a quadratic penalty.
The second tests whether a kinetic steady state that is known to exist
can be recovered from the boundary condition that it satisfies. The
third repeats that test with moiety conservation imposed, so that
the predicted steady state must conserve moieties as well as satisfy
the same boundary condition.

Set against these merits are several limitations, some intrinsic and
some that we envisage being overcome in future work. A disadvantage
of any comprehensive approach to modelling biochemical reaction networks
is that each type of constraint or data requires the definition of
a new variable, whose biochemical interpretation depends on prior
understanding of the mathematical form of the physicochemical and
biochemical constraints concerned. As the number of different constraints
and variables increases, any modelling formalism becomes more challenging
to understand. This is compounded in conic optimisation, where auxiliary
variables are required to express the problem in a form amenable to
solution with established conic optimisation solvers. Some of these
variables have accessible biochemical interpretations, e.g., an upper
bound on the total internal net reaction rate, while others are unfamiliar
and therefore harder to interpret biochemically. Our approach to managing
this proliferation is a consistent nomenclature, with subscripts specific
to particular instances of the same type of variable, and repeated
use of as few cones as possible. Nevertheless, understanding this
approach does require a basic familiarity with conic optimisation,
which is less widespread in the biochemical modelling community than,
say, linear optimisation.

A second limitation concerns the rate laws themselves. Elementary
kinetic rate laws based on mass-action kinetics are valid only for
true elementary steps under ideal, well-mixed, dilute conditions.
Where a reaction involves complex mechanisms, non-ideal behaviour,
heterogeneous phases, transport limitations or biological regulation,
alternative mechanistic or phenomenological rate laws are required.
Genome-scale models predominantly represent enzyme-catalysed reactions
as overall rather than elementary reactions, which is partly an artefact
of their development for prediction with early constraint-based modelling
approaches, and partly a pragmatic response to pathways whose stoichiometry
or chemical structural specification is unclear for want of experimental
data, e.g., lipid metabolism. Lumped reactions create a particular
difficulty here, because the standard transformed reaction Gibbs energy
of a lumped reaction is the sum over the series of overall reactions
lumped together. An artificially large negative standard transformed
reaction Gibbs energy makes the thermodynamic constraint on the difference
between logarithmic pseudoelementary kinetic parameters numerically
awkward, since the relative values of those parameters then differ
greatly in magnitude. In turn the unidirectional forward rate becomes
artificially large and the reverse rate artificially small, or vice
versa, either of which strains the finite precision arithmetic of
a numerical optimisation solver. This motivates continued reconstruction
effort, combining manual and algorithmic approaches, to split lumped
reactions, and it warrants further analysis of whether variational
kinetics can be extended to phenomenological kinetic rate laws \parencite{cornish-bowden_fundamentals_1981,cook_enzyme_2007}
applied systematically at genome-scale through a formalism admissible
to tractable computation, e.g., convenience kinetics \parencite{liebermeister_bringing_2006}.

Herein variational kinetics is presented primarily as a means to represent
elementary kinetic constraints for a biochemical reaction network
in which all molecular species concentrations are assumed to be at
steady state. That assumption can be relaxed by admitting deviation
from steady state under a quadratic penalty, where the deviation may
be interpreted as a discrete differential of molecular species concentrations
with a time unit consistent with the reaction rates. We likewise admit
quadratic penalisation of deviation from given molecular species concentrations,
assuming that the mean concentrations are compatible with a steady
state. One could then consider an iterative scheme in which molecular
species concentrations are fixed at the sum of the concentrations
and the concentration deviation from the previous iteration, so that
the quadratic penalty enforces a form of continuity, or smoothing,
of trajectories with respect to time. Application to biochemical network
dynamics via damped gradient flow is thus one of a variety of natural
theoretical extensions.

\section{Conclusions}

Variational kinetics is a novel, scalable biochemical network modelling
method that enables satisfaction of elementary reaction kinetics in
non-equilibrium steady states, optionally with the addition of moiety
conservation and thermodynamic constraints on elementary kinetic parameters.
These constraints couple variables representing unidirectional fluxes,
molecular species concentrations and, optionally, (pseudo)elementary
kinetic variables, unless they are specified as given parameters.
All variables can be constrained by given lower and upper bounds.
Optionally, quadratic penalisation of deviations from steady state
constraints enables a kinetically feasible state to be obtained despite
conflicting constraints. Linear optimisation of external net reaction
rates enables representation of biologically motivated objectives,
subject to the aforementioned constraints. The method is implemented
by a sequence of conic optimisation problems that is globally convergent
to an elementary reaction kinetic steady state, which is demonstrated
to exist for any stoichiometrically consistent biochemical network.
Variational kinetics is envisaged to provide a foundation for genome-scale
biochemical network modelling, extending the focus of constraint-based
modelling beyond metabolism and bringing kinetic modelling within
reach at genome-scale.

\paragraph*{CRediT authorship contribution statement}

Ronan M.T. Fleming: Conceptualisation, Methodology, Software, Formal
analysis, Validation, Visualization, Writing - original draft, Writing
- review and editing, Funding acquisition. Ines Thiele: Funding acquisition,
Resources, Writing - review and editing.

\paragraph*{Declaration of competing interest}

The authors declare that they have no known competing financial interests
or personal relationships that could have appeared to influence the
work reported in this paper.

\paragraph*{Data availability}

The genome-scale metabolic model used (code/ data/ raw/ iDopaNeuroC\_VK\_withL.mat),
the computational narrative (src/ matlab/ JTB/ driver\_optimizeVKmodel\_VK1to3.mlx)
that reproduces the numerical experiments in Section \ref{sec:Numerical-experiments},
as well as all of the source code implementing variational kinetics,
including the sequential conic solver is available at \href{https://github.com/Digital-Metabolic-Twin-Centre/varkin}{https://github.com/Digital-Metabolic-Twin-Centre/varkin}
(SHA-1 b130be3). Reproduction of numerical computations requires a
numerical computing environment MATLAB R2024b (Mathworks Inc.), the
COBRA Toolbox (v3.8beta, specifically SHA-1 shorthand: 4713424eb)
and a conic optimisation solver (MOSEK Version 11.2.0, feasibility
and optimality tolerance $1e-6$), a commercial solver for which free
academic licences are available. The absolute paths in the narrative
have to be updated to reflect the location that, e.g, iDopaNeuroC\_VK\_withL.mat
exists on each system.

\paragraph*{Acknowledgements}

This work was funded by the European Union Horizon Europe Framework
Programme within the 'Reconstruction and Computational Modelling for
Inherited Metabolic Diseases' project (Recon4IMD, 101080997), the
European Research Council under the European Union Horizon Europe
research and innovation programme for the 'Innovative whole-body metabolism
models for personalised medicine' project (AVATAR, 101125633) and
the European Commission, Research and Innovation action within the
'Systems Medicine of Mitochondrial Parkinson’s Disease' project (SysMedPD,
668738), the U.S. National Institutes of Health and Department of
Energy interagency, collaborative research award for the 'Multiscale
Molecular Systems Biology: Reconstruction and Model Optimization'
project (U01GM102098) and the U.S. Department of Energy (Office of
Biological and Environmental Research) under the 'Numerical Optimization
Algorithms and Software for Systems Biology' project (DE-FG02-09ER25917).

\paragraph*{Declaration of generative AI and AI-assisted technologies in the
manuscript preparation process}

During the preparation of this work the authors used Claude (Anthropic),
accessed through the Claude Code command-line interface, in order
to draft and revise passages of manuscript text, to draft figure captions,
and to write the Python source code that deterministically render
Figure 1,2,3 and Supplementary Figure 5 from instances of the equations
defined in this manuscript. No general-purpose generative AI image
tool was used to create or alter any figure, and no image representing
primary observed or experimental data was created or altered by AI;
all graphical output is produced by versioned, re-executable plotting
code. After using this tool, the authors reviewed and edited the content
and take full responsibility for the content of the published article.

\printbibliography
\pagebreak{}

\appendix

\section{Local rate-law approximations}

Figure \ref{fig:kineticFormats} illustrates how local rate-law approximations
lose accuracy away from the reference state.

\begin{figure}[H]
\includegraphics[width=1\columnwidth]{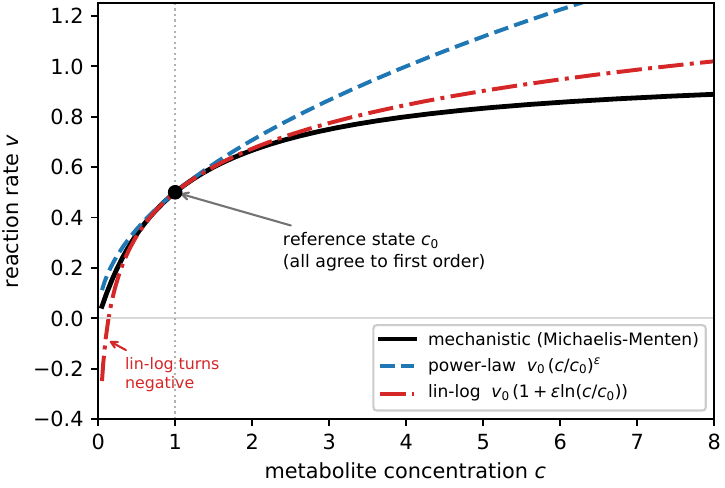}\caption{\protect\label{fig:kineticFormats}Local rate-law approximations lose
accuracy away from the reference state. The mechanistic Michaelis-Menten
rate $v=V_{\max}c/(K_{M}+c)$ (black) is approximated at a reference
concentration $c_{0}$ by the power-law / S-system form $v_{0}(c/c_{0})^{\varepsilon}$
(blue) and the linear-logarithmic (lin-log) form $v_{0}(1+\varepsilon\ln(c/c_{0}))$
(red), both built from the reference elasticity $\varepsilon=K_{M}/(K_{M}+c_{0})$.
All three agree to first order at $c_{0}$ but diverge away from it,
the lin-log form even turning negative at low concentration.}
\end{figure}

\section{Elementary exponential and logarithmic identities}

Let $x,y\in\mathbb{R}_{>0}$, then the following hold

\begin{eqnarray*}
\exp(\ln(x)) & = & x\\
\exp(x+y) & = & \exp(x)\cdot\exp(y)\\
\ln(xy) & = & \ln(x)+\ln(y)\\
xy=\exp(\ln(xy)) & = & \exp(\ln(x)+\ln(y))=\exp(\ln(x))\cdot\exp(\ln(y))\\
x^{2}y=\exp(2\ln(x)+\ln(y)) & = & \exp(2\ln(x))\cdot\exp(\ln(y))\\
\prod x_{\textrm{i}}^{a_{\textrm{i}}} & = & \exp(a^{T}\cdot\ln(x))\\
\exp(0) & = & 1\\
\ln\left(\frac{x}{y}\right) & = & \ln(x)-\ln(y)\\
\ln(x^{y}) & = & y\ln(x)\\
-\ln(x) & = & \ln\left(\frac{1}{x}\right)
\end{eqnarray*}

\section{\protect\label{sec:Cyclic-stoichiometric-matrix}Cyclic stoichiometric
matrix}
\begin{thm}
Let 
\begin{equation}
\bar{N}\;\coloneqq\;\begin{bmatrix}N & -I_{m}\\
0 & L
\end{bmatrix}\in\mathbb{R}^{(2m-r)\times(n+m)},\label{eq:cyclicNetStoich}
\end{equation}
where 
\[
N\in\mathbb{R}^{m\times n},\qquad L\in\mathbb{R}^{(m-r)\times m},
\]
satisfy 
\[
LN=0,\qquad\textrm{rank}(N)=r,\qquad\textrm{rank}(L)=m-r.
\]
Then:\\
(i) The left nullspace of $\bar{N}$ is 
\begin{equation}
\mathcal{N}_{\ell}(\bar{N})=\left\{ \begin{bmatrix}L^{T}\\
I_{m-r}
\end{bmatrix}y\;:\;y\in\mathbb{R}^{m-r}\right\} ,\label{eq:leftNullCyclic}
\end{equation}
so a basis is given by the rows of $\begin{bmatrix}L & I_{m-r}\end{bmatrix}$,
and 
\[
\dim\mathcal{N}_{\ell}(\bar{N})=m-r.
\]
(ii) The right nullspace of $\bar{N}$ is 
\begin{equation}
\mathcal{N}_{r}(\bar{N})=\left\{ \begin{bmatrix}I_{n}\\
N
\end{bmatrix}v\;:\;v\in\mathbb{R}^{n}\right\} ,\label{eq:rightNullCyclic}
\end{equation}
so a basis is given by the columns of $\begin{bmatrix}I_{n}\\
N
\end{bmatrix}$, and 
\[
\dim\mathcal{N}_{r}(\bar{N})=n.
\]
(iii) The rank of $\bar{N}$ is 
\[
\textrm{rank}(\bar{N})=m.
\]
\end{thm}

\begin{proof}
\emph{Left nullspace.} Let $y\in\mathbb{R}^{2m-r}$ be a left-null
vector of $\bar{N}$, written in block form as 
\[
y=\begin{bmatrix}a\\
b
\end{bmatrix},\qquad a\in\mathbb{R}^{m},\;b\in\mathbb{R}^{m-r}.
\]
Then $y^{T}\cdot\bar{N}=0$ is equivalent to 
\begin{equation}
\begin{cases}
a^{T}\cdot N=0,\\
-a^{T}+b^{T}\cdot L=0.
\end{cases}\label{eq:a'N}
\end{equation}
From $a^{T}\cdot N=0$ and the assumption that the rows of $L$ span
the left nullspace of $N$, there exists $t\in\mathbb{R}^{m-r}$ such
that $a^{T}=t^{T}\cdot L$. Substituting into the second equation
of \ref{eq:a'N} gives 
\[
-t^{T}\cdot L+b^{T}\cdot L=0\quad\Longleftrightarrow\quad(b^{T}-t^{T})L=0.
\]
Since $\textrm{rank}(L)=m-r$, $L$ has full row rank, hence $(b^{T}-t^{T})L=0$
implies $b^{T}=t^{T}$. Therefore 
\[
y^{T}=\begin{bmatrix}t^{T}\cdot L & t^{T}\end{bmatrix}=t^{T}\cdot\begin{bmatrix}L & I_{m-r}\end{bmatrix}.
\]
Thus every left-null vector of $\bar{N}$ is a linear combination
of the rows of $\begin{bmatrix}L & I_{m-r}\end{bmatrix}$.

Conversely, a direct computation yields 
\[
\begin{bmatrix}L & I_{m-r}\end{bmatrix}\begin{bmatrix}N & -I_{m}\\
0 & L
\end{bmatrix}=\begin{bmatrix}LN & -L+L\end{bmatrix}=\begin{bmatrix}0 & 0\end{bmatrix},
\]
so every row of $\begin{bmatrix}L & I_{m-r}\end{bmatrix}$ lies in
$\mathcal{N}_{\ell}(\bar{N})$. Because these rows are linearly independent,
they form a basis and $\dim\mathcal{N}_{\ell}(\bar{N})=m-r$. Equivalently,
\[
\mathcal{N}_{\ell}(\bar{N})=\left\{ \begin{bmatrix}L^{T}\\
I_{m-r}
\end{bmatrix}y:\ y\in\mathbb{R}^{m-r}\right\} .
\]
\emph{Right nullspace.} Let $x\in\mathbb{R}^{n+m}$ be written as
\[
x=\begin{bmatrix}v\\
w
\end{bmatrix},\qquad v\in\mathbb{R}^{n},\;w\in\mathbb{R}^{m}.
\]
Then $\bar{N}x=0$ is equivalent to the block equations 
\[
Nv-w=0,\qquad Lw=0.\tag{2}
\]
The first equation in (2) implies $w=Nv$. Substituting into the second
gives 
\[
Lw=LNv=0,
\]
which holds for every $v\in\mathbb{R}^{n}$ because $LN=0$. Hence
\[
\mathcal{N}_{r}(\bar{N})=\left\{ \begin{bmatrix}v\\
Nv
\end{bmatrix}:\ v\in\mathbb{R}^{n}\right\} ,
\]
and the columns of $\begin{bmatrix}I_{n}\\
N
\end{bmatrix}$ form a basis, so $\dim\mathcal{N}_{r}(\bar{N})=n$.

\emph{Rank.} By rank–nullity applied to $\bar{N}\in\mathbb{R}^{(2m-r)\times(n+m)}$,
\[
\textrm{rank}(\bar{N})+\dim\mathcal{N}_{\ell}(\bar{N})=2m-r,
\]
so using $\dim\mathcal{N}_{\ell}(\bar{N})=m-r$ yields 
\[
\textrm{rank}(\bar{N})=(2m-r)-(m-r)=m.
\]
\end{proof}
Given \ref{eq:cyclicNetStoich}, the forward and reverse stoichiometric
matrices are
\[
\bar{F}\;\coloneqq\;\begin{bmatrix}F & I_{m}\\
0 & 0
\end{bmatrix}\qquad\bar{R}\;\coloneqq\;\begin{bmatrix}R & 0\\
0 & L
\end{bmatrix}\in\mathbb{R}^{(2m-r)\times(n+m)},
\]
and their horizontal concatenation is
\begin{eqnarray*}
\left[\begin{array}{cc}
\bar{F}, & \bar{R}\end{array}\right] & = & \left[\begin{array}{cccc}
F & I_{m} & R & 0\\
0 & 0 & 0 & L
\end{array}\right]
\end{eqnarray*}
with $\textrm{rank}\left(\left[\begin{array}{cc}
\bar{F}, & \bar{R}\end{array}\right]\right)=2m-r$, since $\textrm{rank}\left(L\right)=m-r$. Note that $\textrm{rank}\left(\left[\begin{array}{ccc}
F, & I_{m}, & R\end{array}\right]\right)=m$ irrespective of $\textrm{rank}\left(\left[\begin{array}{cc}
F, & R\end{array}\right]\right)$.

\section{\protect\label{subsec:VKOptimality-conditions}Optimality conditions}

The Lagrangian corresponding to the combined formulation of variational
kinetics expressed as Problem (\ref{eq:VKObjectiveKinetics})-(\ref{eq:boundsu0})
in Section \ref{subsec:Comprehensive-formulation-of-VK} is

\begin{eqnarray*}
 &  & \mathcal{L}(v_{f},v_{r},lnk_{f},lnk_{r},lnc,c,u^{\circ},r,w,t_{r},\ldots\\
 &  & y_{N},y_{u^{\circ}},y_{L},s_{v_{f}},s_{f_{1}},s_{F},s_{v_{r}},s_{r_{1}},s_{R},s_{c},s_{c_{1}},s_{lnc},s_{ce},s_{t_{r}},s_{r1},s_{H_{r}})\coloneqq\\
 &  & c_{v_{f}}^{T}\cdot v_{f}+c_{v_{r}}^{T}\cdot v_{r}+c_{lnk_{f}}^{T}\cdot lnk_{f}+c_{lnk_{r}}^{T}\cdot lnk_{r}+c_{w}^{T}\cdot w+c_{lnc}^{T}\cdot lnc+c_{c}^{T}\cdot c\ldots\\
 &  & +c_{t_{r}}^{T}\cdot t_{r}\ldots\\
 &  & +y_{N}^{T}\cdot(N\cdot(v_{f}-v_{r})+r+B\cdot w)\ldots\\
 &  & +y_{u^{\circ}}^{T}\cdot(lnk_{f}-lnk_{r}+N{}^{T}\cdot u^{\circ})\ldots\\
 &  & +y_{L}^{T}\cdot(L\cdot c-L\cdot c(0))\ldots\\
 &  & +\left(\begin{array}{c}
s_{v_{f}}\\
s_{f_{1}}\\
s_{F}
\end{array}\right)^{T}\cdot\left(\begin{array}{c}
v_{f}\\
1\\
F^{T}\cdot lnc+lnk_{f}
\end{array}\right)+\left(\begin{array}{c}
s_{v_{r}}\\
s_{r1}\\
s_{R}
\end{array}\right)^{T}\cdot\left(\begin{array}{c}
v_{r}\\
1\\
R^{T}\cdot lnc+lnk_{r}
\end{array}\right)+\left(\begin{array}{c}
s_{c}\\
s_{c_{1}}\\
s_{lnc}
\end{array}\right)^{T}\cdot\left(\begin{array}{c}
c\\
1\\
lnc
\end{array}\right)\ldots\\
 &  & +\left(\begin{array}{c}
s_{t_{r}}\\
s_{r1}\\
s_{H_{r}}
\end{array}\right)^{T}\cdot\left(\begin{array}{c}
t_{r}\\
1\\
H_{r}\cdot r
\end{array}\right)\ldots\\
 &  & +z_{v_{f}}^{T}\cdot(u_{v_{f}}-v_{f})+z_{v_{r}}^{T}\cdot(u_{v_{r}}-v_{r})\ldots\\
 &  & +z_{lv}^{T}\cdot((v_{f}-v_{r})-l_{v})+z_{uv}^{T}\cdot(u_{v}-(v_{f}-v_{r}))+z_{lw}^{T}\cdot(w-l_{w})+z_{uw}^{T}\cdot(u_{w}-w)\ldots\\
 &  & +z_{llnkf}^{T}\cdot(lnkf-l_{lnkf})+z_{ulnkf}^{T}\cdot(u_{lnkf}-lnkf)+z_{llnkr}^{T}\cdot(lnkr-l_{lnkr})+z_{ulnkr}^{T}\cdot(u_{lnkr}-lnkr)\ldots\\
 &  & +z_{llnc}^{T}\cdot(lnc-l_{lnc})+z_{ulnc}^{T}\cdot(u_{lnc}-lnc)+z_{lu^{\circ}}^{T}\cdot(u^{\circ}-l_{u^{\circ}})+z_{uu^{\circ}}^{T}\cdot(u_{u^{\circ}}-u^{\circ}),
\end{eqnarray*}
where we require that the following dual variables are non-negative
$z_{v_{f}},z_{v_{r}},z_{lv},z_{uv},z_{llnkf},z_{ulnkf},z_{llnkr},z_{ulnkr}\in\mathbb{R}_{\ge0}^{n}$,
$z_{lw},z_{uw}\in\mathbb{R}_{\ge0}^{k}$, $z_{llnc},z_{ulnc},z_{lu^{\circ}}^{T},z_{uu^{\circ}}\in\mathbb{R}_{\ge0}^{m}$
and that the primal and dual variables are feasible with respect to
the primal and dual cones, that is
\begin{eqnarray}
\left(\begin{array}{c}
v_{f}\\
1\\
F^{T}\cdot lnc+lnk_{f}
\end{array}\right)\in\mathcal{K}_{exp}^{n}, &  & \left(\begin{array}{c}
s_{vf}\\
s_{f1}\\
s_{F}
\end{array}\right)\in\mathcal{K}_{exp}^{^{\star}n},\label{eq:pdConicallyFeasibleForwardFlux}\\
\left(\begin{array}{c}
v_{r}\\
1\\
R^{T}\cdot lnc+lnk_{r}
\end{array}\right)\in\mathcal{K}_{exp}^{n}, &  & \left(\begin{array}{c}
s_{vr}\\
s_{r1}\\
s_{R}
\end{array}\right)\in\mathcal{K}_{exp}^{^{\star}n},\label{eq:pdConicallyFeasibleReverseFlux}\\
\left(\begin{array}{c}
c\\
1\\
lnc
\end{array}\right)\in\mathcal{K}_{exp}^{m}, &  & \left(\begin{array}{c}
s_{c}\\
s_{c1}\\
s_{lnc}
\end{array}\right)\in\mathcal{K}_{exp}^{^{\star}m},\label{eq:eq:pdConicallyFeasibleConcentration}\\
\left(\begin{array}{c}
t_{r}\\
1\\
H_{r}\cdot r
\end{array}\right)\in\mathcal{Q}^{2+m}, &  & \left(\begin{array}{c}
s_{t_{r}}\\
s_{r1}\\
s_{H_{r}}
\end{array}\right)\in\mathcal{Q}^{^{\star}2+m}.\label{eq:pdConicallyFeasibleQuadr}
\end{eqnarray}
Derivation of the optimality conditions to a general conic optimisation
problem are described elsewhere ($\mathdollar$5.9.2 \cite{boyd_convex_2004}).
The optimality conditions to Problem (\ref{eq:VKObjectiveKinetics})-(\ref{eq:boundsu0})
may be obtained by setting the partial derivatives of the Lagrangian
with respect to the variables $v_{f},v_{r},lnk_{f},lnk_{r},lnc,c,u^{\circ},r,w,t_{r},y_{N},y_{u^{\circ}}$
and $y_{L}$ to zero, that is

\begin{eqnarray}
\frac{\partial\mathcal{L}}{\partial v_{f}} & = & c_{v_{f}}-N^{T}\cdot y_{N}^{\star}-s_{vf}^{\star}-s_{ef1}^{\star}+s_{ef}^{\star}-s_{vfe}^{\star}-z_{v_{f}}^{\star}+z_{v}^{\star}=0\label{eq:dL_vf-1}\\
\frac{\partial\mathcal{L}}{\partial v_{r}} & = & c_{v_{r}}+N^{T}\cdot y_{N}^{\star}-s_{vr}^{\star}-s_{er1}^{\star}+s_{er}^{\star}-s_{vre}^{\star}-z_{v_{r}}^{\star}-z_{v}^{\star}=0\label{eq:dL_vr-1}\\
\frac{\partial\mathcal{L}}{\partial lnk_{f}} & = & c_{lnk_{f}}-y_{u^{\circ}}^{\star}-s_{F}^{\star}-z_{lnkf}^{\star}=0\\
\frac{\partial\mathcal{L}}{\partial lnk_{r}} & = & c_{lnk_{r}}+y_{u^{\circ}}^{\star}-s_{R}^{\star}-z_{lnkr}^{\star}=0\\
\frac{\partial\mathcal{L}}{\partial lnc} & = & c_{lnc}-F\cdot s_{F}^{\star}-R\cdot s_{R}^{\star}-s_{lnc}^{\star}-z_{lnc}^{\star}=0\\
\frac{\partial\mathcal{L}}{\partial c} & = & c_{c}+s_{c_{e}}^{\star}-s_{c}^{\star}-z_{c}^{\star}+L^{T}\cdot y_{L}^{\star}=0\\
\frac{\partial\mathcal{L}}{\partial u^{\circ}} & = & c_{u^{\circ}}-z_{u^{\circ}}^{\star}+N\cdot y_{u^{\circ}}^{\star}=0\\
\frac{\partial\mathcal{L}}{\partial r} & = & c_{r}-H_{r}\cdot s_{H_{r}}^{\star}-z_{r}^{\star}+y_{N}^{\star}=0\\
\frac{\partial\mathcal{L}}{\partial w} & = & c_{w}-B^{T}\cdot y_{N}^{\star}-z_{w}^{\star}=0\\
\frac{\partial\mathcal{L}}{\partial t_{r}} & = & c_{t_{r}}-s_{t_{r}}^{\star}-z_{t_{r}}^{\star}=0\\
\frac{\partial\mathcal{L}}{\partial y_{N}} & = & N\cdot(v_{f}^{\star}-v_{r}^{\star})+r^{\star}+B\cdot w^{\star}=0\nonumber \\
\frac{\partial\mathcal{L}}{\partial y_{u^{\circ}}} & = & lnk_{f}^{\star}-lnk_{r}^{\star}+N{}^{T}\cdot u^{\circ\star}=0\\
\frac{\partial\mathcal{L}}{\partial y_{L}} & = & L\cdot c^{\star}-L\cdot c(0)=0
\end{eqnarray}
and expressing the complementarity constraints between primal and
dual terms, that is

\begin{eqnarray}
\left(\begin{array}{c}
s_{v_{f}}^{\star}\\
s_{f1}^{\star}\\
s_{F}^{\star}
\end{array}\right)^{T}\cdot\left(\begin{array}{c}
v_{f}^{\star}\\
1\\
F^{T}\cdot lnc^{\star}+lnk_{f}^{\star}
\end{array}\right) & = & 0^{nx1}\label{eq:complemenatarityF-1}\\
\left(\begin{array}{c}
s_{v_{r}}^{\star}\\
s_{r1}^{\star}\\
s_{R}^{\star}
\end{array}\right)^{T}\cdot\left(\begin{array}{c}
v_{r}^{\star}\\
1\\
R^{T}\cdot lnc^{\star}+lnk_{r}^{\star}
\end{array}\right) & = & 0^{nx1}\nonumber \\
\left(\begin{array}{c}
s_{c}^{\star}\\
s_{c1}^{\star}\\
s_{lnc}^{\star}
\end{array}\right)^{T}\cdot\left(\begin{array}{c}
c^{\star}\\
1\\
lnc^{\star}
\end{array}\right) & = & 0^{mx1}\\
\left(\begin{array}{c}
s_{t_{r}}^{\star}\\
s_{r1}^{\star}\\
s_{H_{r}}^{\star}
\end{array}\right)^{T}\cdot\left(\begin{array}{c}
t_{r}^{\star}\\
1\\
H_{r}\cdot r^{\star}
\end{array}\right) & = & 0^{1\times1}
\end{eqnarray}
where $0^{nx1}$ denotes a set of $n$ complementarity constraints,
one for each reaction, and similarly for other zero vectors as indicated
above. The final parts of the optimality conditions to Problem (\ref{eq:VKObjectiveKinetics})-(\ref{eq:boundsu0})
are to specify that the primal and dual terms are constrained to reside
within primal and dual conic cones, respectively, that is Eqs. (\ref{eq:pdConicallyFeasibleForwardFlux})-(\ref{eq:pdConicallyFeasibleQuadr}),
with the addition of $^{\star}$ to each variable to denote optimality.

\section{\protect\label{subsec:Chemical-potential}Chemical potential}

We assume constant temperature, $\mathcal{T}=310.15\;\textrm{K}$,
constant pressure, $\mathcal{P}=1\;\textrm{atm}$, and let $x\in\mathbb{R}^{m}$
be a vector of mole fractions of molecular species in a solution.
Assuming an ideal solution, the chemical potential is
\[
u\coloneqq u^{\circ}(\mathcal{T},\mathcal{P},x^{\circ})+\mathcal{R}\mathcal{T}\ln\left(x\right)
\]
where $u^{\circ}(\mathcal{T},\mathcal{P},x^{\circ})$ is the chemical
potential at a standard mole fraction $x^{\circ}\in\mathbb{R}^{m}$
(3.7.1.2.2 in \cite{smith_chemical_1982}) and $\mathcal{R}$ is the
gas constant. The molar concentration of pure water is $c_{w}^{\circ}=55.28$
mol/L, while in typical biochemical solutions the estimated total
concentration of solutes is \textasciitilde$0.3$ mol/L \cite{bennett_absolute_2009},
which is $\sim180$ times less than the molar concentration of water
solvent so a distinction between solvent and solute is appropriate
(5.3 in \parencite{atkins_atkins_2022}. We assume Raoult's law for
the solvent, where the standard mole fraction approaches unity, $x_{w}^{\circ}\rightarrow1$,
and corresponds to a molar concentration of pure water $c_{w}^{\circ}=55.28$
mol/L, consistent with pure water at the aforementioned temperature
and pressure. We assume Henry's law for each solute, where the standard
mole fraction approaches zero, $x^{\circ}\rightarrow0$, and we assume
the standard concentration of each solute is $c^{\circ}=1$ mol/L.

Following an established approach to biochemical thermodynamics \cite{alberty_thermodynamics_2003},
refined for multi-compartmental systems \cite{haraldsdottir_quantitative_2012},
we assume constant compartment-specific pH \cite{haraldsdottir_quantitative_2012},
constant compartment-specific electrical potentials \cite{haraldsdottir_quantitative_2012},
and define standard transformed chemical potential as 
\[
u^{\circ\prime}\coloneqq u^{\circ}(\mathcal{T},\mathcal{P},x^{\circ})+\mathcal{R}\mathcal{T}\ln\left(a\right),
\]
where $a\in\mathbb{R}^{m}$ is the \emph{activity coefficient} of
each molecular species, estimated using the extended Debye-Hückel
equation ($\mathsection$3.6 in \cite{alberty_thermodynamics_2003}).
This absorbs an approximation to non-ideal behaviour into the standard
transformed term, enabling definition of transformed chemical potential
as 
\begin{eqnarray}
u & \coloneqq & u^{\circ\prime}+\mathcal{R}\mathcal{T}\ln\left(x\right)\nonumber \\
 & = & u^{\circ\prime}+\mathcal{R}\mathcal{T}\ln\left(\frac{c}{\mathbf{1}^{T}\cdot c}\right)\label{eq:chemicalPotential}
\end{eqnarray}
where we assume $x=\nicefrac{c}{\mathbf{1}^{T}\cdot c},$ where $c\in\mathbb{R}^{m}$
is a vector of molar concentrations. Transformed (standard) chemical
potential is usually distinguished from (standard) chemical potential
with $^{\prime}$ but henceforth we assume (standard) chemical potential
is transformed and omit the prime for clarity. For an ideal-dilute
solution, assuming the approximation $\mathbf{1}^{T}\cdot c\approx c_{w}^{\circ},$
standard chemical potential in terms of concentration may be related
(Eq 3.7-23 in \cite{smith_chemical_1982}) to standard chemical potential
in terms of mole fraction using 
\begin{equation}
u^{\circ}(\mathcal{T},\mathcal{P},c^{\circ})\coloneqq u^{\circ}(\mathcal{T},\mathcal{P},x^{\circ})-\mathcal{R}\mathcal{T}\ln\left(c_{w}^{\circ}\right).\label{eq:mu0ConcToMoleFraction}
\end{equation}

\section{\protect\label{sec:Adaptive-Sequential-Conic}Adaptive Sequential
Conic Linear Approximation Algorithm}

\subsection{Purpose and problem class}

This section describes the adaptive sequential conic linear approximation
algorithm used to solve the combined formulation of variational elementary
kinetics of Section \ref{subsec:Comprehensive-formulation-of-VK}.
That problem is itself a conic optimisation problem, so a single conic
solve returns an optimal solution, but not in general one at which
the exponential cone constraints representing elementary kinetics
are active. It is activity of those constraints that makes an optimal
solution satisfy elementary reaction kinetics, and Theorem \ref{thm:stationaryIsSteady}
establishes that, for the constraints of variational elementary kinetics,
every stationary point of the iteration below has that property. The
algorithm therefore treats the linear objective coefficients on the
kinetic exponential-cone variables as parameters rather than as given
data, and updates them across a sequence of conic optimisation problems
until each of those constraints is active.

Every problem in that sequence is an inner problem, solved by a conic
optimiser, whose endpoint supplies only a candidate search direction.
Progress is measured instead by a nonlinear outer merit, assembled
from the exponential-cone boundary residuals and evaluated after each
candidate step, so the inner objective that generates a direction
is distinct from the outer merit that decides whether the resulting
step is accepted. Theorem \ref{thm:lyap} of Section \ref{sec:SCLPstationarity}
establishes, for the abstract sequence, convergence to a stationary
point of the outer merit subject to the constraints. What follows
is a concrete realisation of that scheme, together with the initialisation,
step safeguards and adaptive strategy portfolio that it requires in
order to converge on a genome-scale model.

The next subsection states the algorithm in full as pseudocode. The
subsections after it define each of its ingredients in turn: the base
conic model, the exponential-cone boundary residuals and the outer
merit derived from them, the working scaling and fixed columns, the
construction of an initial point, the portfolio of inner conic models,
the step safeguards and line search, the portfolio controller, and
the finalisation and audit of the returned primal-dual tuple. The
description is independent of programming language and follows the
notation of Section \ref{sec:Notation}.

\subsection{Algorithmic summary}

The following summary gives the logical flow of the algorithm, in
the order of the subsections below.
\begin{lyxcode}
Definitions.

~~~Working~feasible~set~$\mathcal{X}\coloneqq\{x\in\mathbb{R}^{n}\,:\,b_{l}\le A\cdot x\le b_{u},\;l\le x\le u,\;F\cdot x+d\in\mathcal{K}\}$.

~~~Active~reduced~exponential-cone~images~$t_{1}(x)\coloneqq F_{1}\cdot x+d_{1}$~and~$t_{3}(x)\coloneqq F_{3}\cdot x+d_{3}$.

~~~Boundary~residuals~$h(x)\coloneqq t_{1}(x)-\exp(t_{3}(x))$~and~$g(x)\coloneqq\ln(t_{1}(x))-t_{3}(x)$.

~~~Outer~merit~$\phi(x)\coloneqq\mathbf{1}^{T}\cdot h(x)+\mathbf{1}^{T}\cdot g(x)$.

~~~Convergence~residual~$\theta(x)\coloneqq\max\{\left\Vert h(x)\right\Vert _{\infty},\left\Vert g(x)\right\Vert _{\infty}\}$.

Working~scaling~and~fixed~columns.

~1~~if~explicit~row~scaling~is~active~then~$A\leftarrow S\cdot A$,~$b_{l}\leftarrow S\cdot b_{l}$~and~$b_{u}\leftarrow S\cdot b_{u}$,

~~~~~~~leaving~$F$,~$d$,~$l$,~$u$~and~the~caller~objective~$a$~unscaled.

~2~~$\mathcal{F}\leftarrow$~the~structurally~fixed~columns;~when~projection~is~enabled,~every~candidate

~~~~~~~start,~inner~endpoint~and~accepted~trial~point~is~projected~onto~$\{x:x_{\mathcal{F}}=l_{\mathcal{F}}=u_{\mathcal{F}}\}$.

Initialisation.

~3~~$a_{0}\leftarrow a$~with~the~active~merit-column~entries~set~to~zero.

~4~~for~each~start~mode,~in~the~order~centred~feasible,~elastic~feasible,~feasible,

~~~~~~~relaxed~feasible,~all-ones,~do

~5~~~~~~for~the~centred~feasible~mode,~solve~the~elastic~Phase~I~problem,~then~the

~~~~~~~~~~~centring~problem~driving~$t_{1}\to\mathbf{1}$~and~$t_{3}\to0$,~then,~if~the~raw~explicit

~~~~~~~~~~~residual~exceeds~the~repair~threshold,~the~one-norm-movement-penalised

~~~~~~~~~~~repair~problem.

~6~~~~~~$x^{0}\leftarrow$~the~candidate~point,~projected~onto~the~fixed-column~subspace.

~7~~~~~~if~$r_{\textrm{exp}}(x^{0})\le\tau_{\textrm{start}}$,~or~the~soft~repaired-centred~allowance~is~met,

~~~~~~~~~~~then~accept~$x^{0}$~and~leave~the~loop.

~8~~if~no~start~mode~is~accepted~then~stop:~there~is~no~usable~initial~point.

Outer~loop.

~9~~$k\leftarrow0$;~select~the~first~strategy~of~the~portfolio.

10~~while~$k<k_{\max}$~do

11~~~~~~evaluate~$\phi(x^{k})$,~$h$,~$g$,~$\theta(x^{k})$,~$\nabla\phi(x^{k})$~and~the~curvature~diagnostics.

12~~~~~~if~$\theta(x^{k})\le\tau_{\textrm{stop}}$~then~stop:~the~outer~loop~has~converged.

13~~~~~~repeat

14~~~~~~~~~~build~the~inner~conic~model~of~the~current~strategy~over~$\mathcal{X}$:

~~~~~~~~~~~~~~~cost~approximation,~minimising~$(c_{\textrm{in}}^{k})^{T}\cdot x$~with~$c_{\textrm{in}}^{k}$~the~active~merit~gradient;

~~~~~~~~~~~~~~~reachable-boundary~cost~approximation,~adding~the~boundary~attraction

~~~~~~~~~~~~~~~~~~~$\eta_{\textrm{bd}}^{k}(a_{\textrm{bd}}^{k})^{T}\cdot x$~and~the~target-centred~rotated~quadratic~penalty;

~~~~~~~~~~~~~~~block-curvature~regularised~cost~approximation;

~~~~~~~~~~~~~~~residual-balanced~quadratically~regularised~cost~approximation;

~~~~~~~~~~~~~~~quadratically~regularised~cost~approximation;

~~~~~~~~~~~~~~~local-box~quadratically~regularised~cost~approximation,~which~adds

~~~~~~~~~~~~~~~~~~~temporary~affine-image~bounds~on~$t_{1}$~and~$t_{3}$.

15~~~~~~~~~~solve~the~inner~conic~model~and~set~$d^{k}\leftarrow x_{\textrm{in}}^{k}-x^{k}$;~for~a~safeguarded

~~~~~~~~~~~~~~~gradient~strategy~set~$d^{k}\leftarrow-\nabla\phi(x^{k})\oslash\max\{1,\left\Vert \nabla\phi(x^{k})\right\Vert _{\infty}\}$~instead.

16~~~~~~~~~~if~the~inner~solve~failed,~or~$d^{k}$~is~not~finite,~then~advance~the~strategy~and

~~~~~~~~~~~~~~~return~to~step~14.

17~~~~~~~~~~if~$\nabla\phi(x^{k})^{T}\cdot d^{k}\ge-\tau_{\textrm{desc}}$,~and~the~curvature-rescue~model~does~not~admit

~~~~~~~~~~~~~~~$d^{k}$,~then~advance~the~strategy~and~return~to~step~14.

18~~~~~~~~~~$\lambda\leftarrow$~the~least~of~the~step~caps:~full~step,~explicit~bound,~log~domain,

~~~~~~~~~~~~~~~exponential~boundary,~coefficient~change~and~numerical~movement.

19~~~~~~~~~~while~$\lambda>\lambda_{\min}$~do

20~~~~~~~~~~~~~~$x(\lambda)\leftarrow x^{k}+\lambda d^{k}$,~projected~onto~the~fixed-column~subspace.

21~~~~~~~~~~~~~~accept~$x(\lambda)$~if~$r_{\textrm{exp}}(x(\lambda))\le\tau_{\textrm{exp}}$,~if~$t_{1}(x(\lambda))>\eta\mathbf{1}$,~if

~~~~~~~~~~~~~~~~~~~$\phi(x(\lambda))\le\phi(x^{k})+\sigma\lambda\nabla\phi(x^{k})^{T}\cdot d^{k}$,~and,~once~$\theta(x^{k})$~has~entered

~~~~~~~~~~~~~~~~~~~the~moderate-residual~regime,~if~$\theta(x(\lambda))\le\theta(x^{k})-\delta_{\theta}+\tau_{\theta}$.

22~~~~~~~~~~~~~~if~$x(\lambda)$~is~accepted~then~leave~the~line~search,~else~$\lambda\leftarrow\beta\lambda$.

23~~~~~~~~~~if~no~trial~point~was~accepted,~or~the~accepted~step~is~a~microscopic

~~~~~~~~~~~~~~~boundary-limited~step,~then~advance~the~strategy~and~return~to~step~14.

24~~~~~~until~a~step~is~accepted,~or~every~strategy~has~been~attempted

25~~~~~~if~every~strategy~has~been~attempted~without~an~accepted~step~then~stop:

~~~~~~~~~~~the~portfolio~is~exhausted.

26~~~~~~$x^{k+1}\leftarrow x^{k}+\lambda_{\textrm{k}}d^{k}$~and~$k\leftarrow k+1$.

27~~~~~~update~the~portfolio~controller:~after~repeated~calm~accepted~steps,~probe~back

~~~~~~~~~~~toward~cost~approximation,~unless~the~high-curvature,~repeated~no-descent~or

~~~~~~~~~~~small-residual~suppressors~apply;~advance~to~a~safer~strategy~on~stagnation,

~~~~~~~~~~~on~curvature~dominance,~or~on~boundary~saturation.

Finalisation.

28~~$x_{\textrm{acc}}\leftarrow$~the~last~accepted~outer~iterate.

29~~if~polish~is~permitted~then~solve~the~unrestricted~conic~polish~problem~with~the

~~~~~~~caller~objective~$a$~over~$\mathcal{X}$,~and~accept~its~endpoint~only~if~it~passes~the~final

~~~~~~~explicit-feasibility,~outer-residual~and~merit~gates.

30~~if~no~polished~point~is~accepted~then~solve~the~fixed-primal~recovery~problem~for

~~~~~~~dual~variables~at~$x_{\textrm{acc}}$,~without~moving~the~primal~point.

31~~if~that~also~fails~then~return~$x_{\textrm{acc}}$~with~zero~dual~placeholders~and~a~failure~status.

32~~audit~the~returned~tuple:~$\theta_{\textrm{fin}}$,~the~explicit~residual~in~both~the~working~and~the

~~~~~~~original~row~units,~and~the~Karush-Kuhn-Tucker~residual~of~the~returned~pair.

33~~report~convergence~from~$\theta_{\textrm{fin}}\le\tau_{\textrm{stop}}$.
\end{lyxcode}

\subsection{Base conic model}

Let the primal variable be a column vector $x\in\mathbb{R}^{n}$.
The inner conic models are all built over a common feasible set, possibly
augmented by temporary epigraph variables or temporary local rows.
The unaugmented working feasible set is

\[
\mathcal{X}\coloneqq\left\{ x\in\mathbb{R}^{n}:b_{l}\le A\cdot x\le b_{u},\;l\le x\le u,\;F\cdot x+d\in\mathcal{K}\right\} ,
\]

where

\[
A\in\mathbb{R}^{p\times n},\quad b_{l},b_{u}\in(\mathbb{R}\cup\{-\infty,+\infty\})^{p},\quad l,u\in(\mathbb{R}\cup\{-\infty,+\infty\})^{n},
\]

\[
F\in\mathbb{R}^{q\times n},\quad d\in\mathbb{R}^{q},\quad\mathcal{K}=\mathcal{K}_{1}\times\cdots\times\mathcal{K}_{r}.
\]

All vector inequalities are understood componentwise. The finite lower
and upper bounds in the first two displays define the explicit part
of primal feasibility. The conic inclusion defines the affine conic
part of primal feasibility. The caller-supplied linear objective is
represented by a vector $a\in\mathbb{R}^{n}$, but this objective
is not the outer merit used to control the iterative boundary-matching
process. For the combined formulation of variational elementary kinetics
of Section \ref{subsec:Comprehensive-formulation-of-VK}, the vector
$a$ collects the linear objective coefficients on the exponential-cone
columns, which the algorithm treats as parameters and overwrites at
every major iteration, together with the two fixed coefficients $c_{t_{r}}$
and $c_{w}$ of Eqs. (\ref{eq:VKObjectiveMassBalanceRegulariser})
and (\ref{eq:VKObjectiveOptExternalFlux}), which lie on columns that
no exponential cone touches. The inner conic objective of every strategy
below is assembled from the outer merit gradient on the exponential-cone
columns alone, and is zero on every other column, so the two fixed
terms $c_{t_{r}}^{T}\cdot t_{r}$ and $c_{w}^{T}\cdot w$ do not themselves
steer the outer iteration. They act at two points only: the initial
point is constructed from $a$ with the exponential-cone entries set
to zero, and the final polish solve minimises $a$ over the working
conic model, its endpoint being accepted only if it passes the final
feasibility and residual gates.

For diagnostics and stopping tests, the explicit primal residual is
defined as

\[
r_{\textrm{exp}}(x)\coloneqq\max\left\{ \left\Vert [b_{l}-A\cdot x]_{+}\right\Vert _{\infty},\left\Vert [A\cdot x-b_{u}]_{+}\right\Vert _{\infty},\left\Vert [l-x]_{+}\right\Vert _{\infty},\left\Vert [x-u]_{+}\right\Vert _{\infty}\right\} .
\]

Here $[\cdot]_{+}$ denotes the componentwise positive part. This
residual is used as the raw explicit feasibility gate for accepted
initial points and accepted outer trial points.

\subsection{Exponential-cone boundary residuals}

The algorithm focuses on a subset of primal exponential-cone blocks.
For each active block, the second cone coordinate is fixed at one
after reduction. Thus each active block is represented by

\[
(x_{1},x_{2},x_{3})=(t_{1}(x),1,t_{3}(x)).
\]
The MOSEK primal exponential-cone convention is

\[
(x_{1},x_{2},x_{3})\in\mathcal{K}_{exp}\iff x_{1}\ge x_{2}\exp(x_{3}/x_{2}),\quad x_{1}>0,\quad x_{2}>0.
\]
The corresponding reduced feasibility condition is therefore

\[
t_{1}(x)\ge\exp(t_{3}(x)),\quad t_{1}(x)>0.
\]
The active affine images are written as

\[
t_{1}(x)\coloneqq F_{1}\cdot x+d_{1}\in\mathbb{R}^{m},\qquad t_{3}(x)\coloneqq F_{3}\cdot x+d_{3}\in\mathbb{R}^{m},
\]
where $F_{1},F_{3}\in\mathbb{R}^{m\times n}$ select the first and
third coordinates of the active reduced exponential-cone blocks after
embedding them in the full primal space, with constant offsets $d_{1},d_{3}\in\mathbb{R}^{m}$.
The two residual vectors are

\[
h(x)\coloneqq t_{1}(x)-\exp(t_{3}(x)),
\]

\[
g(x)\coloneqq\ln(t_{1}(x))-t_{3}(x).
\]
The vector $h\in\mathbb{R}^{m}$ is an additive boundary gap. The
vector $g\in\mathbb{R}^{m}$ is a logarithmic boundary gap. Both are
zero on the reduced exponential-cone boundary $t_{1}=\exp(t_{3})$,
and both are nonnegative in the strictly feasible reduced exponential-cone
region. Their simultaneous reduction gives a scale-aware measure of
approach to the target boundary.

The dual exponential cone is interpreted with the corresponding MOSEK
convention

\[
(s_{1},s_{2},s_{3})\in\mathcal{K}_{exp}^{*}\iff s_{1}\ge-s_{3}\exp\left((s_{2}/s_{3})-1\right),
\]

where $s\in\mathbb{R}^{q}$ denotes the affine-conic dual vector.
This convention is relevant for final primal-dual auditing, but the
outer iteration itself is driven primarily by primal residuals $h$
and $g$.

\subsection{Outer merit and derivatives}

The default outer merit is the sum of the additive and logarithmic
boundary gaps:

\[
\phi(x)\coloneqq\mathbf{1}^{T}\cdot h(x)+\mathbf{1}^{T}\cdot g(x).
\]
The associated convergence residual is not the scalar merit but the
maximum componentwise residual

\[
\theta(x)\coloneqq\max\left\{ \left\Vert h(x)\right\Vert _{\infty},\left\Vert g(x)\right\Vert _{\infty}\right\} .
\]
This distinction is important. The scalar merit aggregates all active
blocks and is useful for line search, whereas $\theta(x)$ is the
reported nonlinear stopping residual. Late in the iteration, a trial
point may reduce the sum merit while worsening the largest component.
The algorithm may therefore impose an additional maximum-residual
acceptance gate after $\theta(x)$ enters a moderate residual regime.

For the default merit, the gradient is

\[
\nabla\phi(x)=F_{1}^{T}\cdot\left(\mathbf{1}+t_{1}(x)^{-1}\right)-F_{3}^{T}\cdot\left(\mathbf{1}+\exp(t_{3}(x))\right),
\]
where $t_{1}(x)^{-1}$ denotes the componentwise reciprocal. The gradient
of $h$ is written, using the convention that the first dimension
equals the number of variables and the second dimension equals the
range dimension, as

\[
\nabla h(x)=F_{1}^{T}-F_{3}^{T}\cdot\textrm{diag}(\exp(t_{3}(x)))\in\mathbb{R}^{n\times m}.
\]

\subsection{Working scaling and fixed columns}

Before the initial point is constructed, the explicit linear rows
may be scaled by a positive diagonal matrix. If row scaling is active,
the working explicit constraints become

\[
A_{\textrm{sc}}\coloneqq S\cdot A,\qquad b_{l,\textrm{sc}}\coloneqq S\cdot b_{l},\qquad b_{u,\textrm{sc}}\coloneqq S\cdot b_{u},
\]
where $S\in\mathbb{R}^{p\times p}$ is a positive diagonal matrix.
The affine conic image $F\cdot x+d$, the variable bounds $l$ and
$u$, and the original conic objective vector $a$ are not scaled.
Consequently, the exponential-cone boundary $t_{1}(x)=\exp(t_{3}(x))$
is unchanged. The entire initialisation, all inner models, all line
searches, and finalisation use a single working representation, either
scaled or unscaled. At return, explicit row duals are mapped back
to the caller's original row units. Some variables may be structurally
fixed. Let the fixed-column set be denoted by $\mathcal{F}\subseteq\{1,\ldots,n\}$.
The fixed-column affine subspace is

\[
\{x\in\mathbb{R}^{n}:x_{\mathcal{F}}=l_{\mathcal{F}}=u_{\mathcal{F}}\}.
\]
Candidate starts, inner endpoints, and accepted trial points are optionally
projected back to this affine subspace. This prevents a small fixed-bound
residual in the initial point from being inherited throughout the
outer iteration.

\subsection{Initialisation}

The initialisation stage is not intended to minimise the original
conic objective. Its purpose is to construct a numerically usable
outer starting point $x^{0}$ in the working feasible set, compatible
with fixed columns and well scaled for evaluating the nonlinear merit.
To reduce the risk of choosing an extreme starting point, the active
merit-column entries of the original conic objective are set to zero
before the start problem is solved. The remaining inactive part may
act only as a weak tie-breaker in the direct feasible start; the Phase
I, centring and repair problems replace the objective entirely. The
preferred initialisation is a centred feasible construction. First,
an elastic Phase I conic problem is solved. This introduces nonnegative
elastic variables that relax finite explicit row and variable bounds,
while leaving the original affine conic geometry unchanged. In abstract
form, the Phase I objective is

\[
\min\;\mathbf{1}^{T}\cdot r_{l}^{A}+\mathbf{1}^{T}\cdot r_{u}^{A}+\mathbf{1}^{T}\cdot r_{l}^{x}+\mathbf{1}^{T}\cdot r_{u}^{x},
\]
subject to relaxed lower and upper explicit constraints and the original
conic inclusion. The resulting point provides a feasible or nearly
feasible reference for the second stage.

Second, a centring conic problem is solved over the original feasible
set. Its role is to place the reduced exponential-cone images near
the well-scaled boundary point $(1,1,0)$. A representative centring
problem is

\[
\min_{x,r_{1},r_{3}}\;\mathbf{1}^{T}\cdot r_{1}+\mathbf{1}^{T}\cdot r_{3}
\]

\[
\text{subject to }x\in\mathcal{X},\qquad-r_{1}\le t_{1}(x)-\mathbf{1}\le r_{1},\qquad-r_{3}\le t_{3}(x)\le r_{3},
\]

\[
r_{1}\ge0,\qquad r_{3}\ge0.
\]
Thus the start is encouraged to satisfy $t_{1}$ approximately equal
to one and $t_{3}$ approximately equal to zero. These values correspond
to the reduced exponential-cone boundary point $(x_{1},x_{2},x_{3})=(1,1,0)$,
since $1=\exp(0)$.

If the centred point is good in scaled units but misses the raw explicit
feasibility gate, an optional raw-feasibility repair step solves another
conic problem in the same feasible set while penalising movement from
the centred reference. A typical repair penalty is a one-norm movement
proxy.

\[
\min_{x,p,q}\;\omega_{\textrm{rep}}\,\mathbf{1}^{T}\cdot(p+q)
\]

\[
\text{subject to }x\in\mathcal{X},\qquad x-x_{\textrm{ref}}\le p,\qquad x_{\textrm{ref}}-x\le q,\qquad p\ge0,\quad q\ge0.
\]

Here $\omega_{\textrm{rep}}>0$ is the repair penalty weight and $p,q\in\mathbb{R}^{n}$
are the nonnegative positive and negative parts of the movement from
the centred reference $x_{\textrm{ref}}$.

Other start modes can be used as fallbacks: a direct feasible conic
solve, an elastic feasible solve, a relaxed feasible solve, and a
last-resort all-ones algebraic point. Irrespective of the start mode,
acceptance is governed by the same raw explicit residual gate

\[
r_{\textrm{exp}}(x^{0})\le\tau_{\textrm{start}}.
\]
A scaled residual may be retained for diagnosis, but it does not by
itself make a raw-infeasible point acceptable. The one relaxation
of this gate is that a centred start whose raw residual exceeds $\tau_{\textrm{start}}$
is still admitted when the repair step strictly improved it and it
meets a slightly larger soft allowance. This design ensures that the
point entering the outer loop is feasible in the same row units used
later by the explicit-bound step cap and by the line-search acceptance
test.

\subsection{Adaptive portfolio of inner conic models}

At each major iteration $k$, the algorithm evaluates the current
merit state at $x^{k}$ and then chooses a strategy from an ordered
portfolio. Each strategy builds an inner conic model whose feasible
set is the working conic feasible set, possibly augmented by temporary
local rows or convex epigraph variables. The inner objective is constructed
from the current outer merit gradient, optional residual-balancing
weights, and optional convex regularisation terms. The endpoint of
the inner conic solve is not accepted directly; it defines a direction
that must pass descent and line-search safeguards.

\subsubsection{Cost approximation}

The simplest strategy is cost approximation. It solves an inner conic
problem of the form

\[
\min_{x\in\mathcal{X}}\;(c_{\textrm{in}}^{k})^{T}\cdot x,
\]
where $c_{\textrm{in}}^{k}\in\mathbb{R}^{n}$ equals the active part
of the current merit gradient, with fixed and structurally inactive
slots suppressed and with optional small inactive tie-break terms.
The endpoint $x_{\textrm{in}}^{k}\in\mathbb{R}^{n}$ defines the candidate
direction

\[
d^{k}\coloneqq x_{\textrm{in}}^{k}-x^{k}\in\mathbb{R}^{n}.
\]
In an exact conic solve, if $x^{k}$ is feasible for the same inner
model and $c_{\textrm{in}}^{k}$ is precisely the gradient model,
the inner objective should not increase along this endpoint direction.
A value of $\nabla\phi(x^{k})^{T}\cdot d^{k}$ that is not sufficiently
negative, specifically $\nabla\phi(x^{k})^{T}\cdot d^{k}\ge-\tau_{\textrm{desc}}$
for a solver-tolerance-aware threshold $\tau_{\textrm{desc}}\ge0$,
is interpreted as a no-descent event and triggers diagnostics or a
switch to a safer strategy, unless a curvature-rescue test on the
finite-step second-order model admits the direction despite the failed
first-order test.

\subsubsection{Reachable-boundary cost approximation}

The reachable-boundary cost approximation strategy augments the cost
approximation model with a linear attraction toward a locally reachable
exponential boundary. It first solves an auxiliary projection problem
over the explicit feasible region to determine whether movement toward
the local boundary is retained under the explicit constraints. Let
$w^{k}\in\mathbb{R}_{\ge0}^{m}$ be nonnegative weights selecting
boundary-relevant residual components. With the gradient convention
above, a boundary-attraction vector can be written as

\[
a_{\textrm{bd}}^{k}\coloneqq\nabla h(x^{k})\cdot w^{k}\in\mathbb{R}^{n}.
\]
The inner objective is then biased by the additional linear term

\[
\eta_{\textrm{bd}}^{k}\,(a_{\textrm{bd}}^{k})^{T}\cdot x,
\]
where the scalar $\eta_{\textrm{bd}}^{k}$ is set to zero when the
projection diagnostic indicates that useful boundary movement is not
reachable. The purpose of this linear bias is not to introduce nonconvexity.
Rather, it preserves conic representability while encouraging a direction
whose first-order effect reduces the active boundary gap.

In addition to this linear attraction, the reachable-boundary strategy
also appends a convex target-centred penalty on the selected blocks.
For each such block it introduces a nonnegative epigraph variable
$q_{\textrm{i}}\in\mathbb{R}_{\ge0}$ together with the curvature-scaled
deviations of $t_{1}(x)$ and $t_{3}(x)$ from a projected boundary
target, collected in a vector $z_{\textrm{i}}\in\mathbb{R}^{2}$,
imposes $(q_{\textrm{i}},1,z_{\textrm{i}})$ in a rotated quadratic
cone (so that $\left\Vert z_{\textrm{i}}\right\Vert _{2}^{2}\le2q_{\textrm{i}}$),
and adds the term $\mu\sum_{\textrm{i}}w_{\textrm{i}}q_{\textrm{i}}$
to the inner objective, with the same nonnegative weights $w_{\textrm{i}}$
and a penalty scale $\mu>0$; the curvature scaling is frozen at $x^{k}$.
This penalty is a convex second-order-cone term, so the inner model
remains conic-representable, and it is inert when $\mu=0$ or no block
is selected.

\subsubsection{Quadratic and local-box regularisation}

When the first-order cost approximation becomes unreliable, the algorithm
can add convex quadratic regularisation while keeping the inner model
conic representable. Let $W_{x}\in\mathbb{R}^{n\times n}$ and $W_{1},W_{3}\in\mathbb{R}^{m\times m}$
be nonnegative diagonal weight matrices. A representative quadratically
regularised model is

\[
\min\;(c_{\textrm{in}}^{k})^{T}\cdot x+\lambda_{x}\xi_{x}+\lambda_{1}\xi_{1}+\lambda_{3}\xi_{3}
\]

\[
\text{subject to }x\in\mathcal{X},
\]

\[
\left\Vert W_{x}^{1/2}\cdot(x-x^{k})\right\Vert _{2}^{2}\le\xi_{x},
\]

\[
\left\Vert W_{1}^{1/2}\cdot(t_{1}(x)-t_{1}(x^{k}))\right\Vert _{2}^{2}\le\xi_{1},
\]

\[
\left\Vert W_{3}^{1/2}\cdot(t_{3}(x)-t_{3}(x^{k}))\right\Vert _{2}^{2}\le\xi_{3}.
\]
Each squared-norm epigraph is represented by a rotated quadratic cone,
for example by imposing $(\xi,1,z)$ in a rotated quadratic cone so
that $\left\Vert z\right\Vert _{2}^{2}\le2\xi$. In a local-box variant,
temporary affine-image bounds are also added:

\[
\left|t_{1}(x)-t_{1}(x^{k})\right|\le\Delta_{1}^{k},\qquad\left|t_{3}(x)-t_{3}(x^{k})\right|\le\Delta_{3}^{k}.
\]
These local rows are temporary. They stabilise the inner endpoint
by restricting movement in the affine images that control the exponential
boundary, but they do not redefine the global feasible set. The local
widths may be adapted according to the observed agreement between
predicted and actual merit decrease.

\subsubsection{Block-curvature and residual-balanced variants}

A block-curvature variant applies conic quadratic regularisation selectively
to exponential-cone blocks whose curvature diagnostics at the current
iterate indicate that the linear model is unreliable. This avoids
adding unnecessary epigraph variables to all blocks while still damping
the directions that dominate the mismatch between the linear model
and the observed merit change. A residual-balanced variant rescales
the $h$ and $g$ contributions to the merit gradient so that one
residual family does not numerically dominate the inner objective
solely because of scale.

\subsubsection{Safeguarded gradient directions}

If the conic endpoint models fail to provide a useful direction, the
portfolio includes safeguarded gradient directions. The basic gradient
direction is

\[
d^{k}\coloneqq-\nabla\phi(x^{k})\oslash\max\{1,\left\Vert \nabla\phi(x^{k})\right\Vert _{\infty}\},
\]
with fixed and structural coordinates suppressed as required. A local-box
gradient variant combines this direction with the same affine-image
neighbourhood restrictions used by local-box cost-approximation strategies.

\subsection{Step safeguards and line search}

After a candidate direction $d^{k}$ has been constructed, the algorithm
searches along the ray

\[
x(\lambda)\coloneqq x^{k}+\lambda d^{k},\qquad\lambda\ge0.
\]
The first admissible trial step is the minimum of several independent
caps. The explicit-bound cap is the largest step satisfying

\[
b_{l}-\tau_{\textrm{exp}}\mathbf{1}\le A\cdot x(\lambda)\le b_{u}+\tau_{\textrm{exp}}\mathbf{1},\qquad l-\tau_{\textrm{exp}}\mathbf{1}\le x(\lambda)\le u+\tau_{\textrm{exp}}\mathbf{1},
\]
where $\tau_{\textrm{exp}}$ is the accepted-trial explicit feasibility
tolerance in the current working row units. A log-domain cap prevents
$t_{1}$ from approaching the logarithmic singularity:

\[
t_{1}(x(\lambda))>\eta\mathbf{1}.
\]
An exponential-boundary cap prevents a boundary-safe strategy from
stepping across the reduced exponential graph. Along the trial ray
the boundary gap is

\[
h(x(\lambda))=t_{1}(x(\lambda))-\exp(t_{3}(x(\lambda)))\in\mathbb{R}^{m}.
\]
The boundary cap is the first positive $\lambda$ at which any component
of $h(x(\lambda))$ reaches zero, multiplied by a fraction-to-boundary
safety factor. Additional caps can limit predicted coefficient changes
and numerically excessive affine-image movement. These caps are diagnostic
as well as protective: the active cap identifies whether a failed
or tiny step was caused by explicit feasibility, log-domain safety,
exponential-boundary safety, coefficient damping, or numerical movement
damping.

The primary line-search acceptance test is an Armijo condition on
the outer merit:

\[
\phi(x^{k}+\lambda d^{k})\le\phi(x^{k})+\sigma\lambda\nabla\phi(x^{k})^{T}\cdot d^{k},\qquad0<\sigma<1.
\]
A trial point must also satisfy the explicit feasibility gate and
the log-domain gate. Once the maximum residual $\theta(x)$ is sufficiently
small, the optional maximum-residual gate requires

\[
\theta(x^{k}+\lambda d^{k})\le\theta(x^{k})-\delta_{\theta}+\tau_{\theta},
\]
where $\delta_{\theta}$ is a required decrease and $\tau_{\theta}$
is an absolute or relative numerical allowance. If a trial fails,
$\lambda$ is reduced geometrically and the tests are repeated. The
accepted update is

\[
x^{k+1}\coloneqq x^{k}+\lambda_{\textrm{k}}d^{k}.
\]

For nominal no-descent directions, a curvature-rescue test may be
applied before rejection. This test evaluates a finite-step second-order
model

\[
\Delta_{2}(\lambda)\coloneqq\lambda\nabla\phi(x^{k})^{T}\cdot d^{k}-\chi_{\textrm{curv}}(\lambda d^{k}),
\]

where $\chi_{\textrm{curv}}$ is a positive curvature correction estimated
from the exponential and logarithmic nonlinearities. A direction that
is not a strict first-order descent direction can still proceed to
line search if the finite-step model and the actual merit decrease
are both favourable.

\subsection{Portfolio adaptation}

The portfolio controller maintains an ordered list of strategies from
simpler to safer or more regularised models. The principal strategies
are cost approximation, reachable-boundary cost approximation, block-curvature
regularised cost approximation, residual-balanced quadratically regularised
cost approximation, quadratically regularised cost approximation,
local-box quadratically regularised cost approximation, safeguarded
gradient, and local-box safeguarded gradient. A user-supplied order
may use all or a subset of these strategies. The default order omits
the two safeguarded gradient strategies, which are enabled only on
request.

Within each major iteration, if the current strategy fails to solve
its inner conic problem, produces a nonfinite direction, produces
a no-descent direction, or fails line search, or yields only a microscopic
boundary-limited step, the controller advances to the next available
strategy in the same major iteration. If all strategies have been
attempted without an accepted step, the portfolio is exhausted. Conversely,
after repeated calm accepted steps, the controller may probe back
toward a simpler strategy, unless recent diagnostics indicate high
curvature, repeated no-descent classifications, or a small-residual
regime in which the simpler model is likely to repeat the previous
failure mode.

The controller also uses diagnostic triggers. A stagnation detector
advances to a safer strategy when the same accepted strategy produces
too little relative merit decrease over a short window. A curvature
detector switches from ordinary cost approximation to a curvature-aware
strategy when recent accepted cost-approximation rows have a large
curvature ratio or a repeated dominant curvature block. A boundary-saturation
detector recognises regimes in which the model remains feasible but
the step is repeatedly throttled by exponential-boundary caps.

\subsection{Finalisation and returned certificate}

The outer iteration returns an accepted primal point, but a downstream
conic interface often requires a primal-dual tuple. The finalisation
stage therefore separates three concepts: the last accepted outer
primal point, the possibly polished returned primal point, and the
dual certificate for the returned point. Let $x_{\textrm{acc}}$ denote
the last accepted outer iterate. The final accepted outer residual
is

\[
\theta_{\textrm{acc}}\coloneqq\theta(x_{\textrm{acc}}).
\]
If the finalisation policy permits polish, and either the residual
is already below a polish threshold or polish is explicitly requested,
an unrestricted conic polish solve is attempted using the caller-supplied
objective $a$ over the original unboxed working conic model. A polished
point is accepted only if it passes the final explicit-feasibility
and outer-residual acceptance gates. If polish is not attempted, fails,
or is rejected, an auxiliary fixed-primal recovery problem may be
solved to recover dual variables at $x_{\textrm{acc}}$ without moving
the primal point.

The final returned point $x_{\textrm{fin}}$ is audited in the same
working row units used by the inner solves, and explicit residuals
are also reported in the original input row units when row scaling
was used. The final Karush-Kuhn-Tucker residual is evaluated on the
actual returned primal-dual tuple, not merely on the auxiliary problem
that produced it. The outer convergence status is determined from

\[
\theta_{\textrm{fin}}\coloneqq\theta(x_{\textrm{fin}})=\max\left\{ \left\Vert h(x_{\textrm{fin}})\right\Vert _{\infty},\left\Vert g(x_{\textrm{fin}})\right\Vert _{\infty}\right\} .
\]
If polish moves the primal point and brings $\theta_{\textrm{fin}}$
below the requested tolerance, the returned solution is classified
as converged at the returned primal point even if the pre-polish portfolio
stopped because all direction models were exhausted. This avoids conflating
failure of the outer step generator with failure of the returned primal
solution.

\subsection{Diagnostics and interpretation}

The algorithm records compact diagnostic rows for initial starts,
accepted steps, rejected attempts, and finalisation. These diagnostics
are not merely logging artefacts; they separate mathematically distinct
failure modes. A solve failure means that the inner conic optimiser
did not return a usable endpoint. A no-descent event means that the
endpoint did not define a descent direction for the selected outer
merit. A line-search failure means that a descent direction was found
but no trial step satisfied the explicit feasibility, domain, boundary,
merit, and maximum-residual gates. A microscopic boundary-limited
step means that a formally acceptable step was so small, typically
because of an exponential-boundary cap, that it was treated as no
real progress.

No-descent diagnostics compare the tested directional derivative,
the movement in the actual inner objective, explicit feasibility of
the current and endpoint points, and the Karush-Kuhn-Tucker residual
of the endpoint. This distinction is important because a no-descent
row can arise from stationarity, complementarity, loose endpoint feasibility,
objective scaling, masking of structural variables, or local-box constraints.
Boundary-skating diagnostics identify active components for which
$h$ is already small and the first-order derivative points toward
the exponential boundary, yielding a tiny fraction-to-boundary step.
Curvature diagnostics estimate how much the nonlinear merit deviates
from its first-order model and identify blocks that repeatedly dominate
this deviation.

The diagnostic summaries are used to interpret the solve and to guide
the adaptive portfolio. They do not replace the mathematical acceptance
gates. The reported solution should therefore be interpreted through
three simultaneous quantities: explicit primal feasibility, outer
nonlinear residual, and the final primal-dual Karush-Kuhn-Tucker residual.

\subsection{Convergence interpretation}

The algorithm should be interpreted as a safeguarded sequential approximation
approach rather than as a single monolithic conic optimisation problem.
Each inner conic solve supplies a direction for a nonlinear boundary-matching
merit; it does not by itself certify convergence of the outer problem.
Conversely, a successful final polish can return a primal point satisfying
the nonlinear residual tolerance even when the adaptive portfolio
has exhausted its direction models. The most meaningful success criterion
is therefore the final returned residual $\theta_{\textrm{fin}}$,
together with explicit primal feasibility and the final Karush-Kuhn-Tucker
audit of the returned tuple.

The adaptive design is motivated by the mixed numerical regimes typical
of large exponential-cone models. Far from the boundary, a simple
cost approximation often gives useful progress. Near the boundary,
the same first-order model may become curvature dominated or boundary
throttled. The portfolio addresses this by adding reachable-boundary
attraction, conic quadratic regularisation, local affine-image boxes,
residual balancing, and safeguarded gradient steps only when diagnostics
indicate that the simpler model is insufficient.
\end{document}